\pdfoutput=1
\documentclass[12pt,reqno]{amsart}   
\usepackage{LocalPackages}
\usepackage[top=1in, left=1.25in, right=1.25in, bottom=1in]{geometry}

\begin{document}

\title[Statistical Estimation]{Statistical Estimation: From Denoising to Sparse Regression and Hidden Cliques}

\author[E. W. Tramel]{Eric W. Tramel}

\author[S. Kumar]{Santhosh Kumar}

\author[A. Giurgiu]{Andrei Giurgiu}

\author[A. Montanari]{Andrea Montanari}

\thanks{%
These are notes from the lecture of Andrea Montanari given at the autumn school ``Statistical Physics, Optimization, Inference, and Message-Passing Algorithms'', that took place in Les Houches, France from Monday September 30th, 2013, till Friday October 11th, 2013. The school was organized by Florent Krzakala from UPMC \& ENS Paris, Federico Ricci-Tersenghi from La Sapienza Roma, Lenka Zdeborov\'a from CEA Saclay \& CNRS, and Riccardo Zecchina from Politecnico Torino.}

\thanks{A.M. was partially supported by the NSF grant CCF-1319979 and the grants AFOSR/DARPA
FA9550-12-1-0411 and FA9550-13-1-0036.}

\thanks{E.W.T. was supported by the ERC under the European Union's 7th Framework Programme Grant Agreement 307087-SPARCS}

\begin{abstract}
These notes review six lectures given by Prof. Andrea Montanari on the 
topic of statistical estimation for linear models. The first two lectures
cover the principles of signal recovery from linear measurements in terms of 
minimax risk. Subsequent lectures demonstrate the application
of these principles to several practical problems in science and engineering.
Specifically, these topics include
denoising of error-laden signals, 
recovery of compressively sensed signals, reconstruction of low-rank matrices,
and also the discovery of hidden cliques within large networks.
\end{abstract}

\maketitle
\clearpage
\tableofcontents

\section*{Preface}
These lectures provide a gentle introduction to some modern topics in
high-dimensional statistics, statistical learning and signal
processing, for an audience without any previous background in these
areas. 
The point of view we take is to connect the recent advances to
basic background in statistics (estimation, regression and the
bias-variance trade-off), and to classical --although non-elementary-- 
developments (sparse estimation and wavelet denoising).

The first three sections will cover these basic and classical
topics. We will then cover more recent research, and discuss sparse
linear regression in Section \ref{lecture_4}, and its analysis for
random designs in Section \ref{lecture_5}. Finally, in Section
\ref{lecture_6} we discuss an intriguing example of a class of
problems whereby sparse and low-rank structures have to be exploited simultaneously.

Needless to say, the selection of topics presented here is very partial.
The reader interested in a deeper understanding can choose from a number of
excellent options for further study.
Very readable introductions to the fundamentals of statistical estimation
can be found in the books by Wasserman 
\citeyear{wasserman2004all,wasserman2006all}. More advanced references
(with a focus on high-dimensional and non-parametric settings) are the
monographs by  Johnstone \citeyear{JohnstoneBook} and Tsybakov
\citeyear{tsybakov2009introduction}.
The recent book by B\"uhlmann and van de Geer
\citeyear{buhlmann2011statistics} provides a useful survey of recent
research in high-dimensional statistics. For the last part of these
notes, dedicated to most recent research topics, we will provide
references to specific papers.

\section{Statistical estimation and linear models}

\subsection{Statistical estimation}

The general problem of statistical estimation is the one of estimating
an unknown object from noisy observations. 
To be concrete, we can consider the model
\begin{equation}
	y = f(\theta; \text{noise})\, ,\label{eq:BasicModel}
\end{equation}
where $y$ is a set of observations, $\theta$ is 
the unknown object, for instance a vector, a set of parameters,  or a function. 
Finally, $f(\, \cdot\,; \noise)$ is an
observation model which links together the observations and
the unknown parameters which we wish to estimate. 
Observations are corrupted by random noise according to this model.
The objective is to produce an estimation $\htheta = \htheta(y)$ that
is accurate under some metric.
The estimation 
of $\theta$ from $y$ is commonly aided by some hypothesis about the structure, 
or behavior, of $\theta$.  Several examples
are described below. 

Statistical estimation can be regarded as a subfield of statistics,
and lies at the core of a number of areas of science and engineering, 
including data mining, signal processing, and inverse problems. Each of these disciplines
provides some information on how to model data acquisition, 
computation, and how best to exploit the hidden structure of the 
model of interest. 
Numerous techniques and algorithms have been developed over a long
period of time, and they often differ in the assumptions and
the objectives  that they try to achieve. 
As an example, a few major distinctions to keep in mind are the following.
\begin{description}
\item[Parametric versus non-parametric] In parametric estimation,
  stringent assumptions are made about the unknown object, hence
  reducing $\theta$ to be determined by a small set of parameters. In contrast,
  non-parametric estimation strives to make minimal modeling
  assumptions, resulting in 
  $\theta$ being an high-dimensional or infinite-dimensional object
  (for instance, a function).
\item[Bayesian versus frequentist] The Bayesian approach assumes 
$\theta$ to be a random variable as well, whose `prior' distribution
plays an obviously important role. From a frequentist point of view,
$\theta$ is
instead an arbitrary point in a set of possibilities. In these
lectures we shall mainly follow the frequentist point of view, but
we stress that the two are in fact closely related.
\item[Statistical efficiency versus computational efficiency]
Within classical estimation theory, a specific estimator $\htheta$ is mainly
evaluated in terms of its accuracy: How close (or far) is $\htheta(y)$
to $\theta$ for typical realizations of the noise? We can broadly refer to
this figure of merit as to `statistical efficiency.'

Within modern applications, computational efficiency has arisen as a
second central concern. Indeed $\theta$ is often high-dimensional:
it is not uncommon to fit models with millions of parameters. The
amounts of observations has grown in parallel.  It becomes therefore
  crucial to devise estimators whose complexity scales gently with
the dimensions, and with the amount of data.
\end{description}
We next discuss informally a few motivating examples.

\subsubsection{Example 1: Exploration seismology}

Large scale statistical estimation plays a key role 
in the field of exploration seismology. This technique uses seismic measurements on the
earth surface to reconstruct geological structures, composition and 
density field of a geological substrates in \cite{HFY2012}.
Measurements are acquired, generally, by sending some known seismic
wave through the ground, perhaps through a controlled explosive detonation,
and measuring the response at multiple spatially dispersed sensors. 

Below is a simple dictionary that points at the various elements of
the model (\ref{eq:BasicModel}) in this example.
\begin{center}
\begin{tabular}{l l}
\multicolumn{2}{c}{Exploration Seismology}\\
\hline
\hline
$y$ 				& seismographic measurements \\
$\theta$ 			& earth density field \\
\textit{Hypothesis} & smooth density field \\
\hline
\end{tabular}
\end{center}

The function $f(\,\cdots\,)$ in  Eq.~(\ref{eq:BasicModel}) expresses
the outcome of the seismographic measurements, given a certain density
field $\theta$ and a certain source of seismic waves (left implicit
since it is known). While this relation is of course complex, and ultimately 
determined by the physics of wave propagation, it is in principle
perfectly known.

Because of the desired resolution of the recovered earth density field,
this statistical estimation problem is often  ill-posed, as sampling
is severely limited by the cost of generating the source signal and the
distribution and set-up of the receivers.
Resolution can be substantially improved by using some structural
insights into the nature of the earth density field. For instance, one
can exploit the fact that this is mostly smooth with the exception of
some discontinuity surfaces.

\begin{figure}[h]
\centering
	\includegraphics[width=\figurewidth]{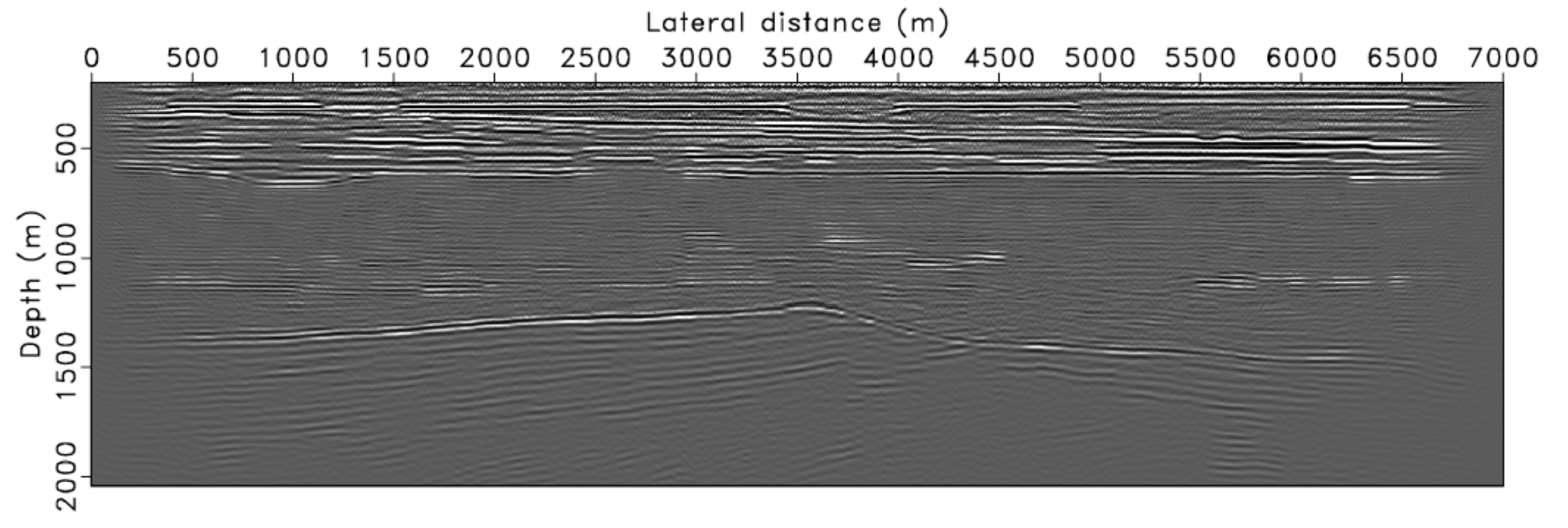}
\caption{A recovered earth density field, from \cite{HFY2012}.}
\end{figure}

\subsubsection{Example 2: Hidden structure in networks}

Many modern data sets are relational, i.e. they express pairwise
relations within a set of objects. This is the case in social
networks, communication networks, unsupervised learning and so on.

In the simplest case, for each pair of nodes in a network, we know
whether they are connected or not.
Finding a hidden structure in such a network is a recurring problem
with these datasets. A highly idealized but nevertheless very
interesting problem requires to find a highly connected subgraph in a
otherwise random graph.

\begin{center}
\begin{tabular}{l l}
\multicolumn{2}{c}{ Hidden Network Structure}\\
\hline
\hline
$y$ 				& large network \\
$\theta$ 			& hidden subset of nodes \\
\textit{Hypothesis} & hidden network is highly connected \\
\hline
\end{tabular}
\end{center}

From Figure \ref{fig:hidden_net}, it is 
apparent that the discovery of such 
networks can be  a difficult task. 

\begin{figure}[ht]
\centering
\subfigure[Subgraph easily visible.]{%
	\includegraphics[width=0.5\figurewidth]{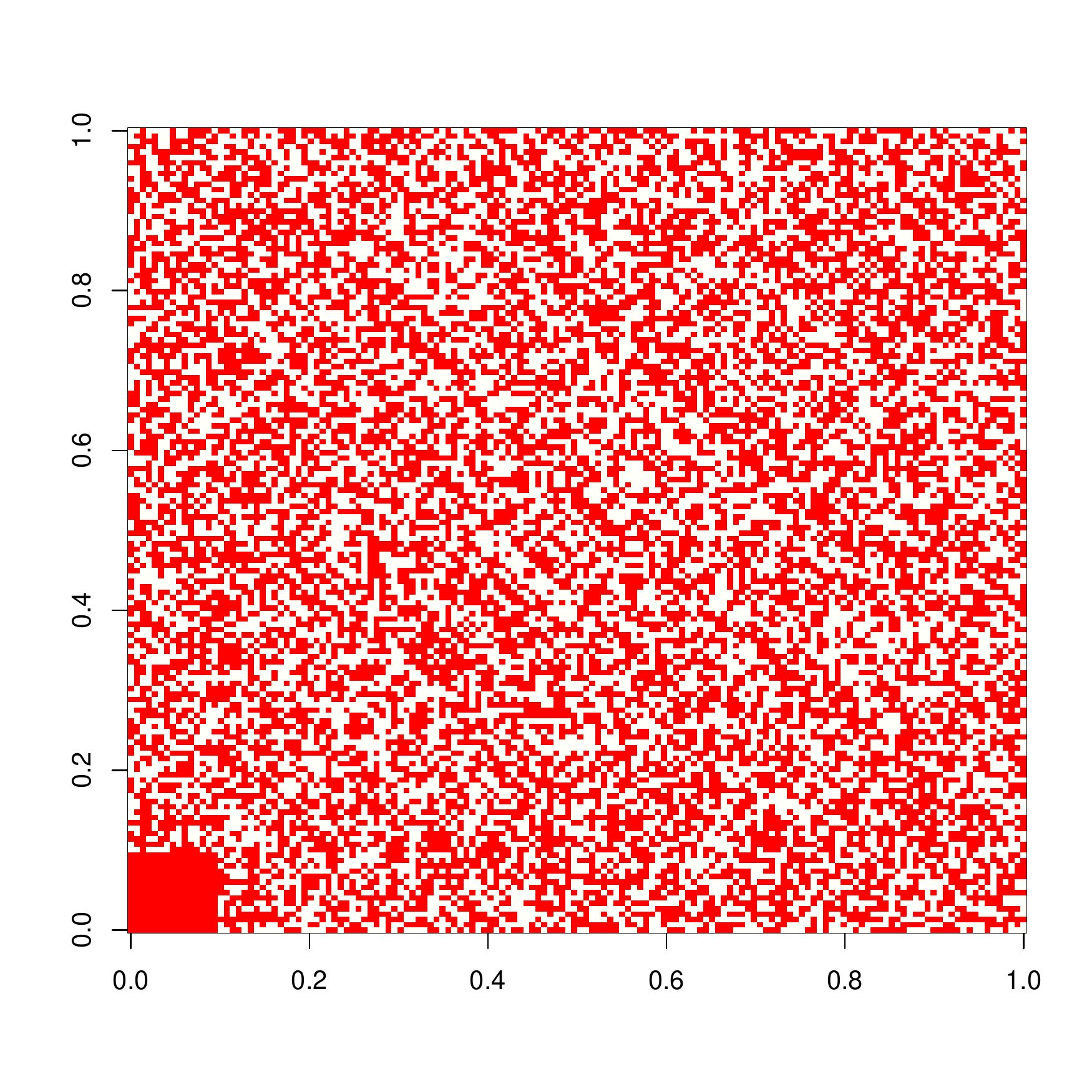}}
\quad
\subfigure[Subgraph hidden.]{%
	\includegraphics[width=0.5\figurewidth]{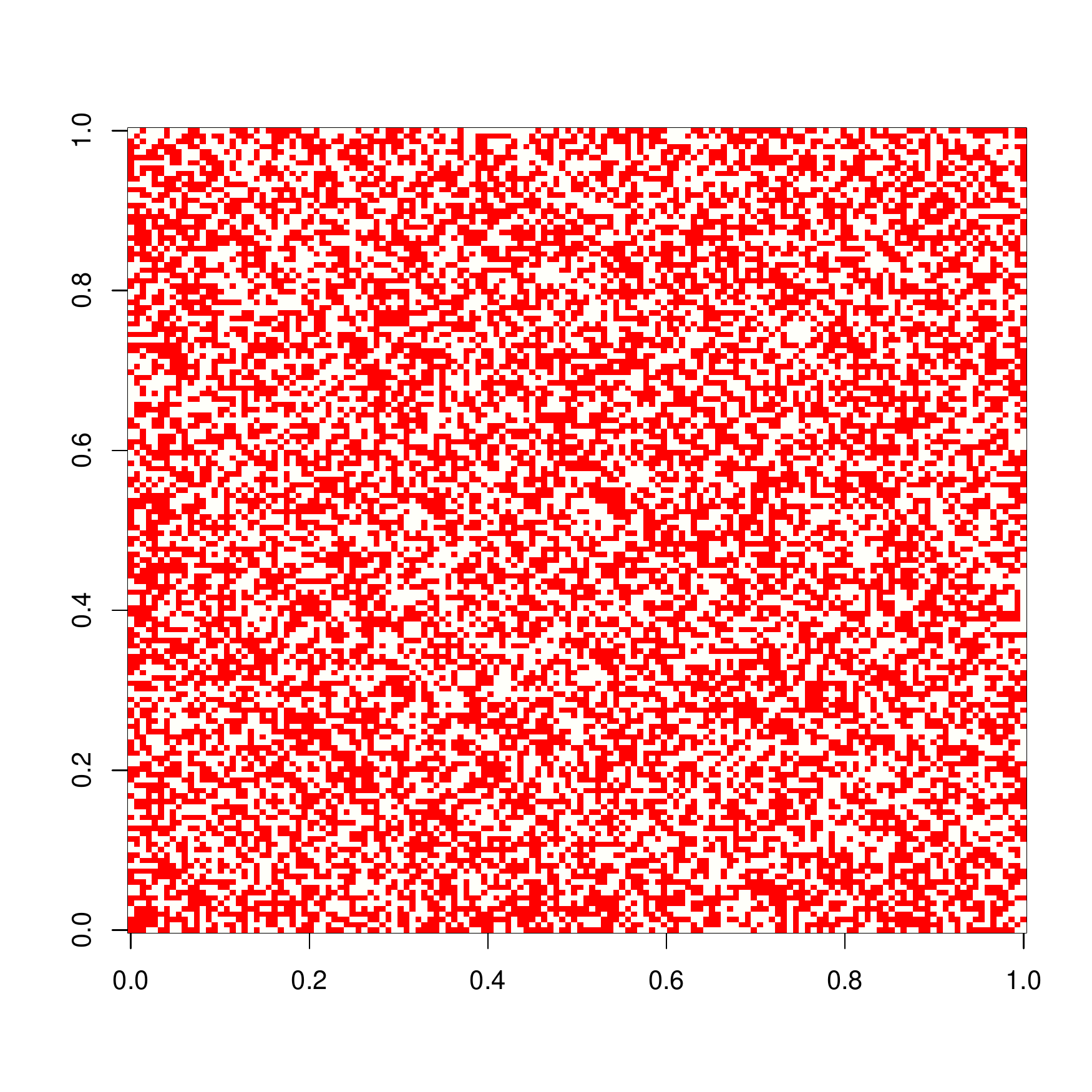}}
\quad
\subfigure[Subgraph revealed]{%
	\includegraphics[width=0.5\figurewidth]{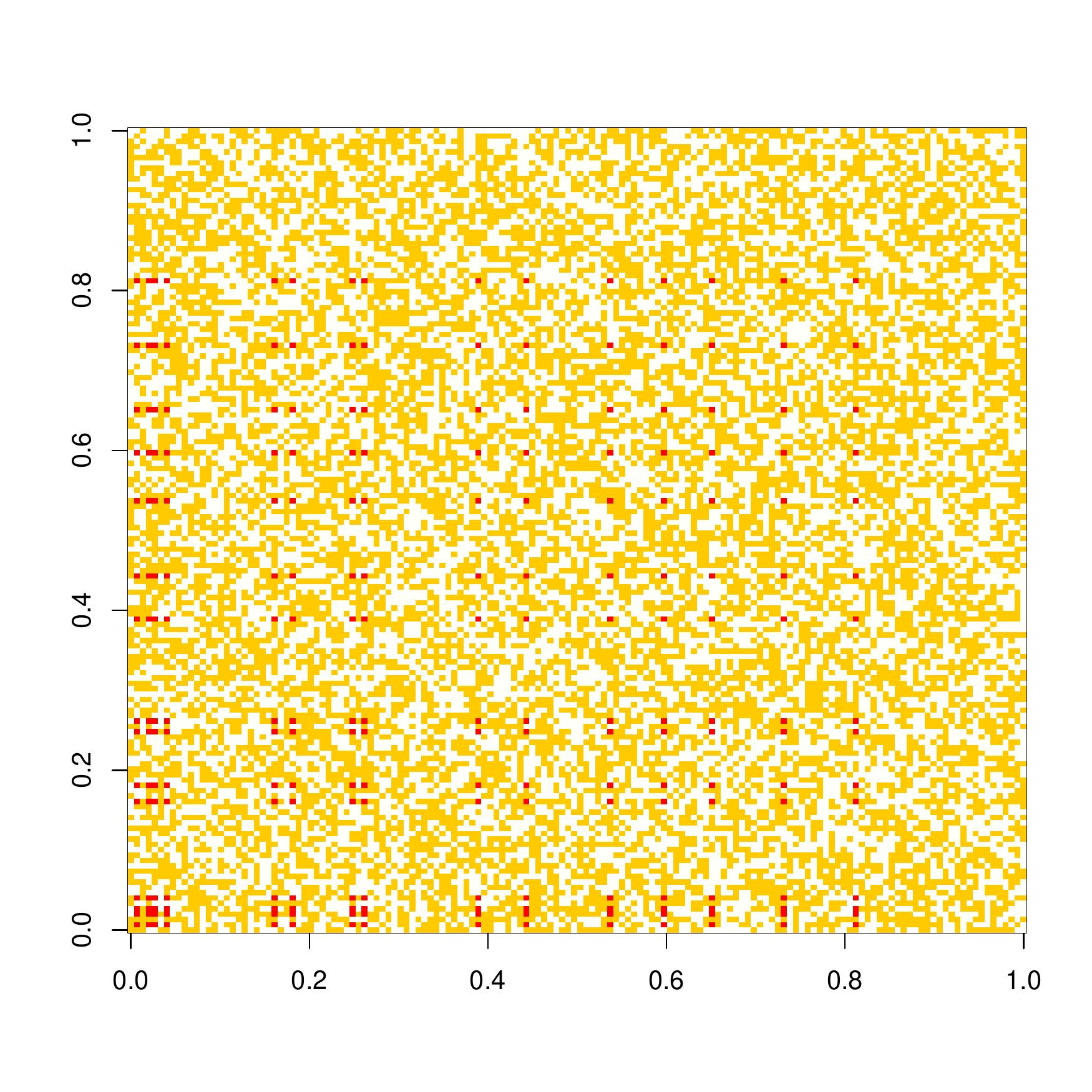}}
\caption{The same network adjacency matrix, is shown in (a) and (b), but 
		 the nodes permuted in (b). In (c), the hidden subgraph is revealed.
		 \label{fig:hidden_net}}
\end{figure}

\subsubsection{Example 3: Collaborative filtering}

Recommendation systems are ubiquitous in e-commerce and web
services. They aim at personalizing each user's experience through an
analysis of her past behavior, and --crucially-- the past behavior of
similar users. The algorithmic and statistical techniques that allow
to exploit this information are referred to as `collaborative
filtering.' Amazon, Netflix, YouTube all make intensive use of
collaborative filtering technologies.

In a idealized model for collaborative filtering, each user of a
e-commerce site is associated to a row of a matrix, and each product
to a column. Entry $\theta_{i,j}$  in this matrix corresponds to the evaluation that user $i$
gives of product $j$.  A small subset of  the entries is observed because of feedback
provided by the users (reviews, ratings, purchasing behavior). 
In this setting, collaborative filtering aims at estimating the whole
matrix, on the basis of noisy observations of relatively few of its entries.

While  this task is generally hopeless, it is observed empirically
that such data matrices are often well approximated by low-rank
matrices. This corresponds to the intuition that a small number of
factors (corresponding to the approximate rank) explain the opinions
of many users concerning many items.
The problem is then modeled as the one of estimating a low-rank matrix
from noisy observations of some of its entries. 
\begin{center}
\begin{tabular}{l l}
\multicolumn{2}{c}{Collaborative Filtering}\\
\hline
\hline
$y$ 				& small set of entries in a large matrix\\
$\theta$ 			& unknown entries of matrix \\
\textit{Hypothesis} & matrix has a low-rank representation \\
\hline
\end{tabular}
\end{center}

A toy example of this problem is demonstrated in 
Figures \ref{fig:mat_recovery_low}-\ref{fig:mat_recovery_high}. It 
can be observed that an accurate estimation of the original matrix 
is possible even when very few of its coefficients are known.

\begin{figure}[ht]
	\centering
	\begin{minipage}[b]{0.47\linewidth}
	\begin{mdframed}
	\centering
	\subfigure[low-rank matrix $M$]{%
		\includegraphics[width=\smallfigurewidth]{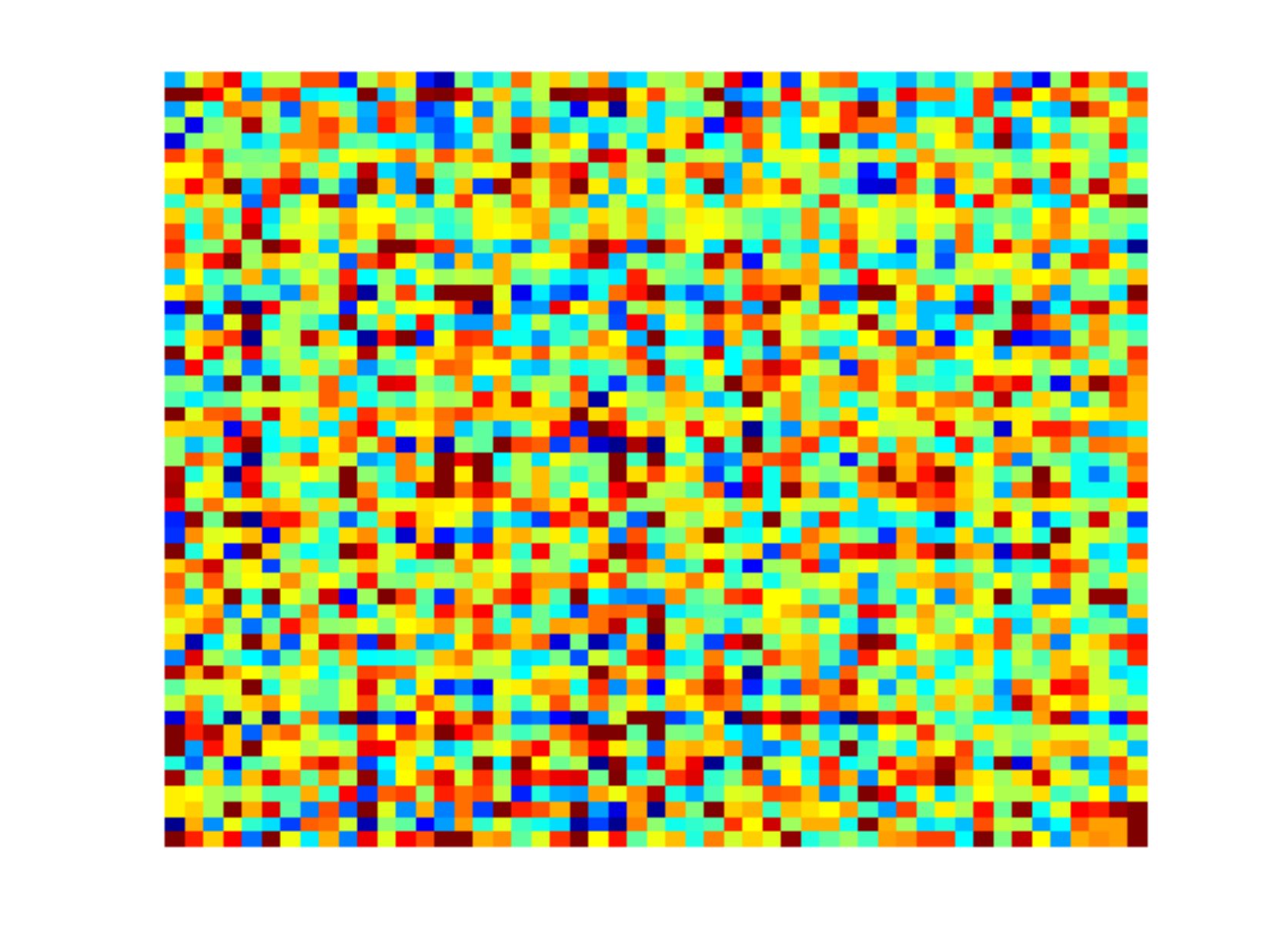}}
	\quad
	\subfigure[sampled matrix $M^\mathbb{E}$]{%
		\includegraphics[width=\smallfigurewidth]{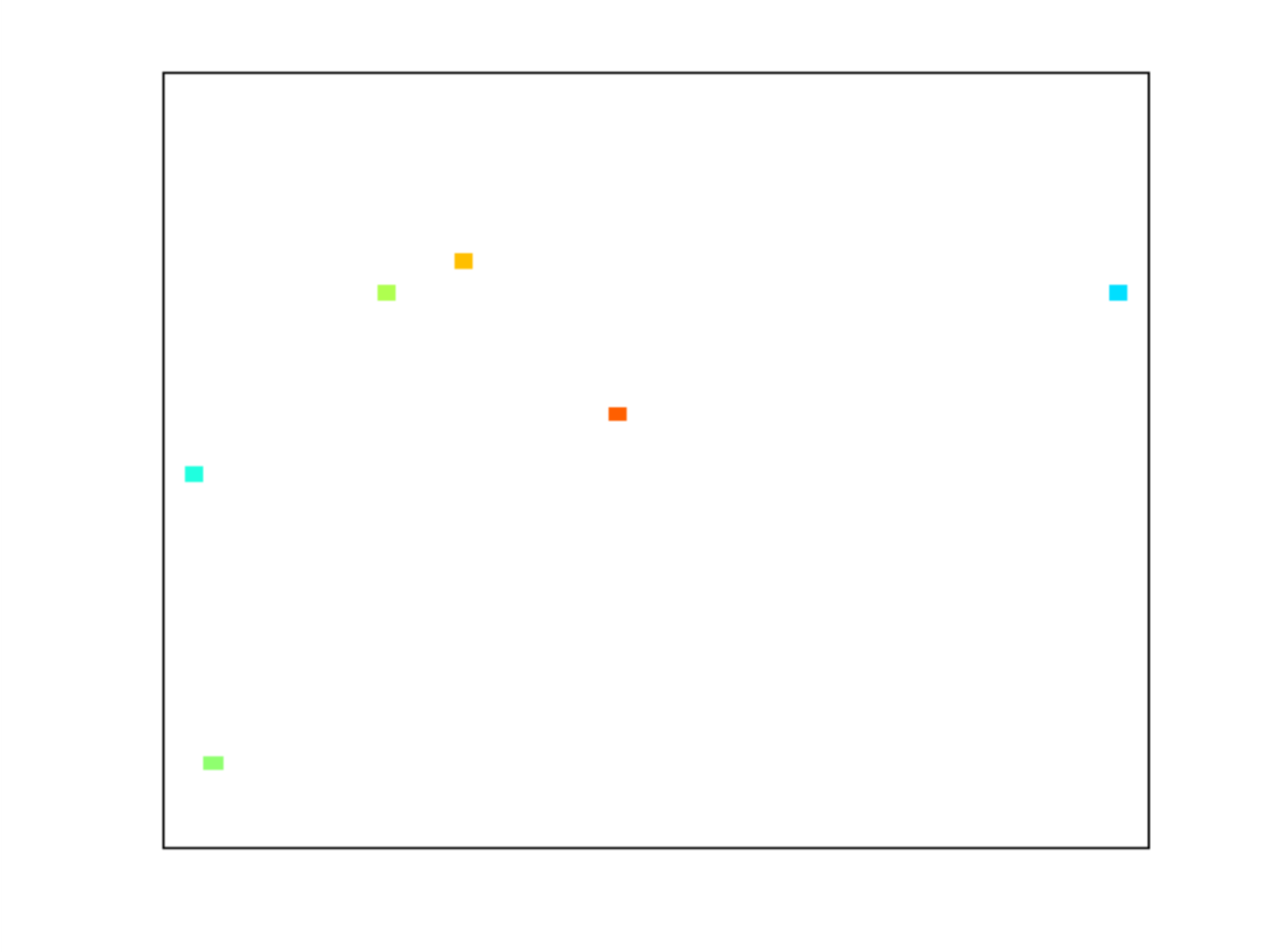}}

	\subfigure[output $\hat{M}$]{%
		\includegraphics[width=\smallfigurewidth]{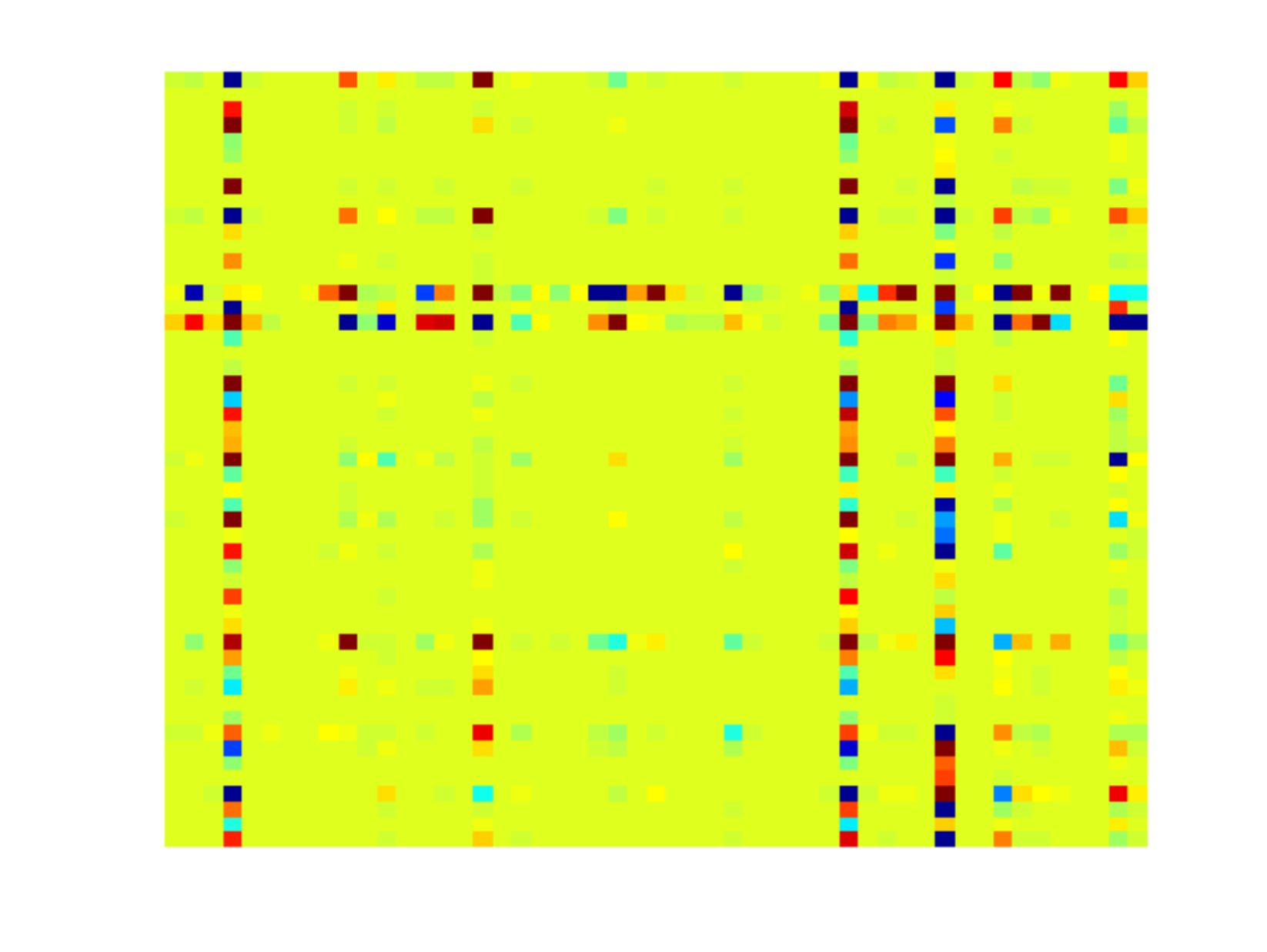}}
	\quad
	\subfigure[squared error $\left( M - \hat{M} \right)^2$]{%
		\includegraphics[width=\smallfigurewidth]{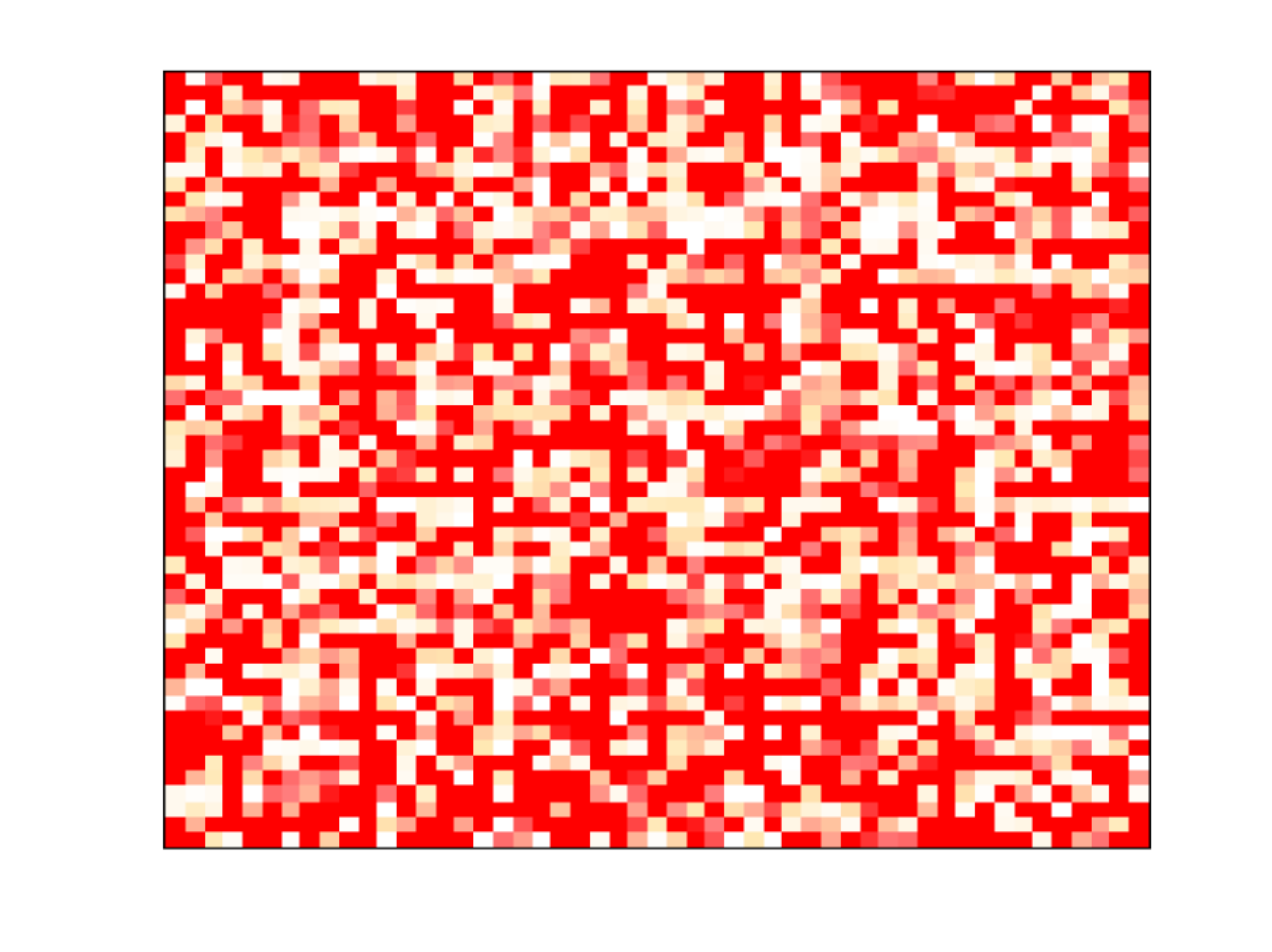}}
	\end{mdframed}	
	\caption{Recovering a $2000\times2000$ rank-8 matrix from $0.25\%$ of its values.
			\label{fig:mat_recovery_low}} 
	\end{minipage}
	\begin{minipage}[b]{0.47\linewidth}
	\begin{mdframed}
	\centering
	\subfigure[low-rank matrix $M$]{%
		\includegraphics[width=\smallfigurewidth]{ChapMontanari/figures/numerical_org.pdf}}
	\quad
	\subfigure[sampled matrix $M^\mathbb{E}$]{%
		\includegraphics[width=\smallfigurewidth]{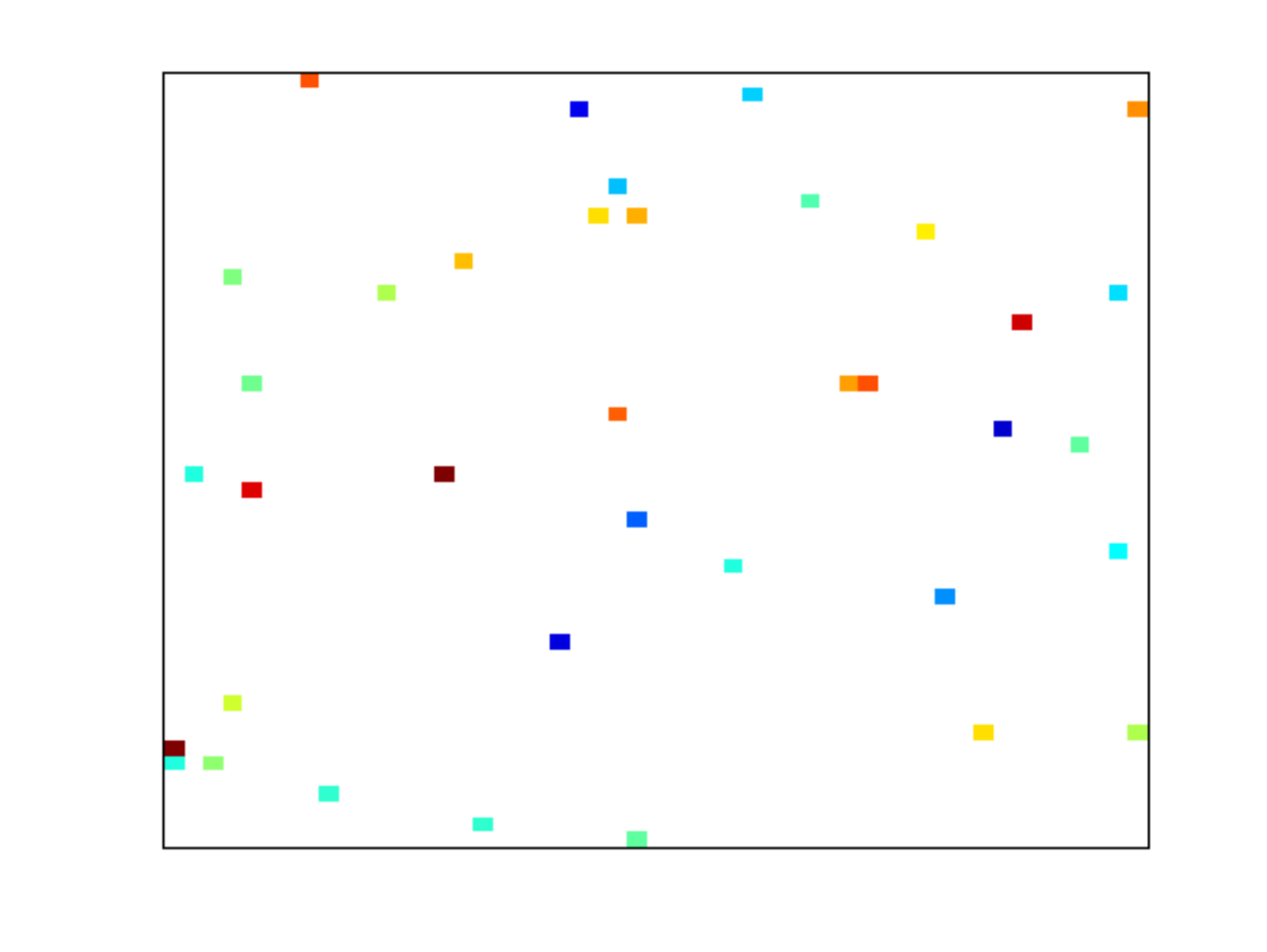}}

	\subfigure[output $\hat{M}$]{%
		\includegraphics[width=\smallfigurewidth]{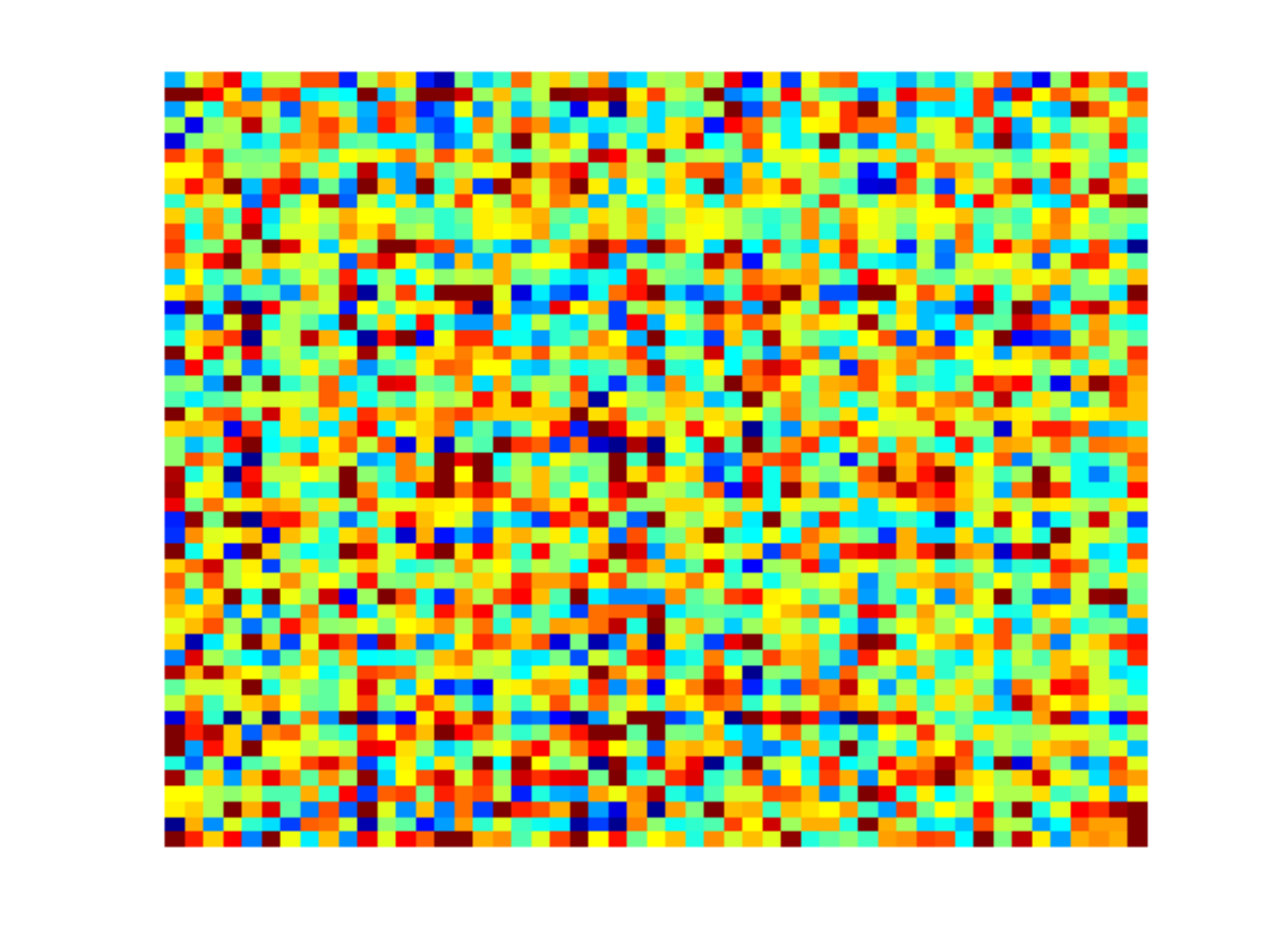}}
	\quad
	\subfigure[squared error $\left( M - \hat{M} \right)^2$]{%
		\includegraphics[width=\smallfigurewidth]{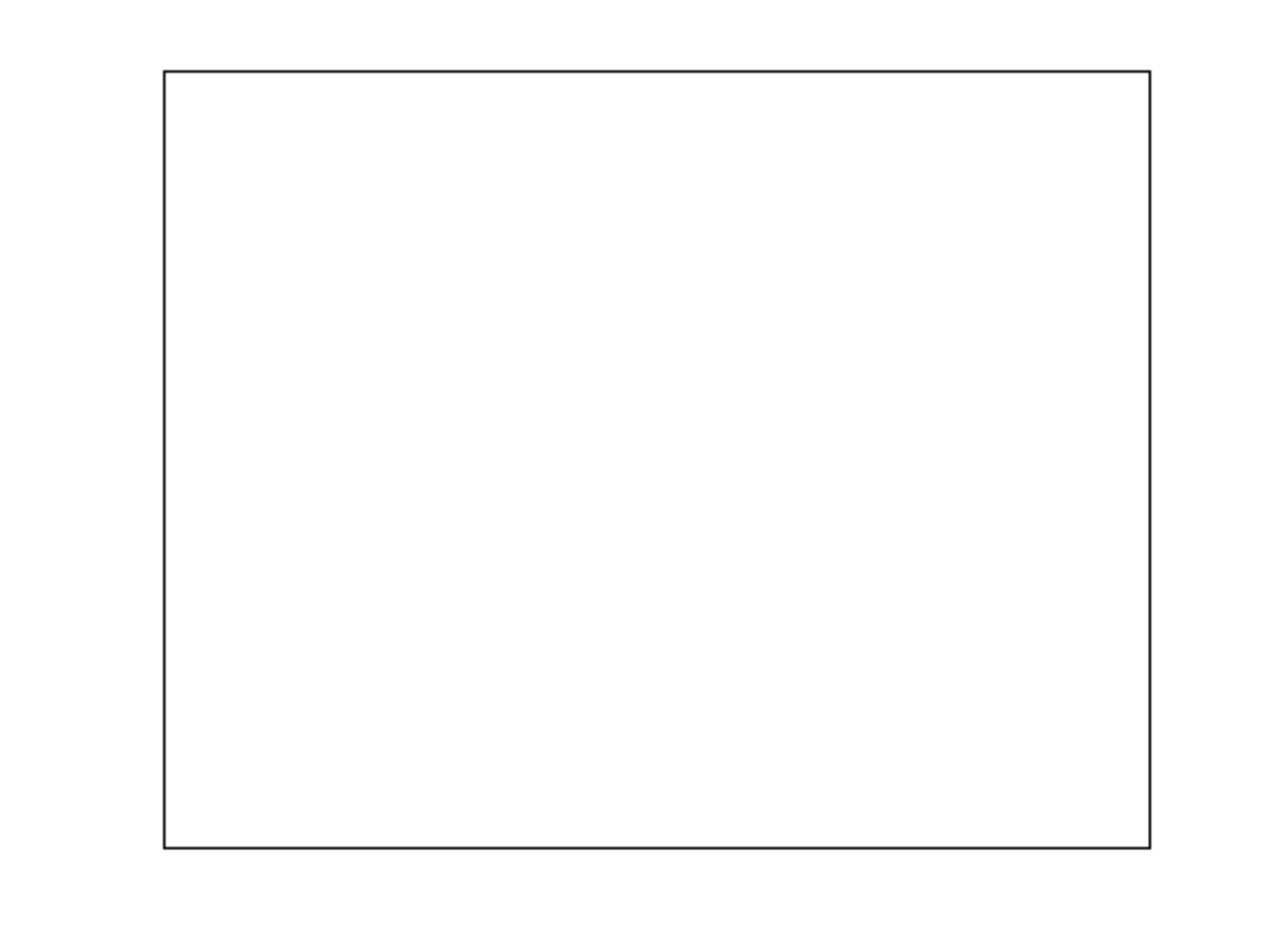}}
	\end{mdframed} 
	\caption{Recovering a $2000\times2000$ rank-8 matrix from $1.75\%$ of its values.
	\label{fig:mat_recovery_high}} 
	\end{minipage}
\end{figure}

\subsection{Denoising}

We will begin by considering in greater depth a specific statistical 
estimation problem, known as `denoising.' One the one hand, 
denoising is interesting, since it is a very common signal processing
task:  In essence, it seeks restore a signal which has been corrupted by some random
process, for instance additive white-noise. On the other, it will allow us to introduce some basic
concepts that will play an important role throughout these
lectures. Finally, recent research by \cite{donoho2013accurate} 
has unveiled a deep and somewhat
surprising connection between denoising and the rapidly developing
field of compressed sensing.
\index{Compressed sensing}

To formally define the problem, we assume that the signal to be
estimated is a function $t\mapsto f(t)$
Without loss of generality, we will
restrict the domain of $f(t)$, $f: \left[0,1 \right] \rightarrow \real$. 
We measure $n$ uniformly-spaced samples over the domain of $f$,
\begin{equation}
	y_i = f\left( i/n \right) + w_i~,
\end{equation}
where $i \in \left\{ 1, 2, \dots, n\right\}$ is the sample index,
and $w_i \sim \normal{0}{\sigma^2}$ is the additive noise term.
Each of $y_1,\dots y_n$ is a sample.

\begin{figure}[h]
\centering
	\includegraphics[width=0.5\figurewidth]{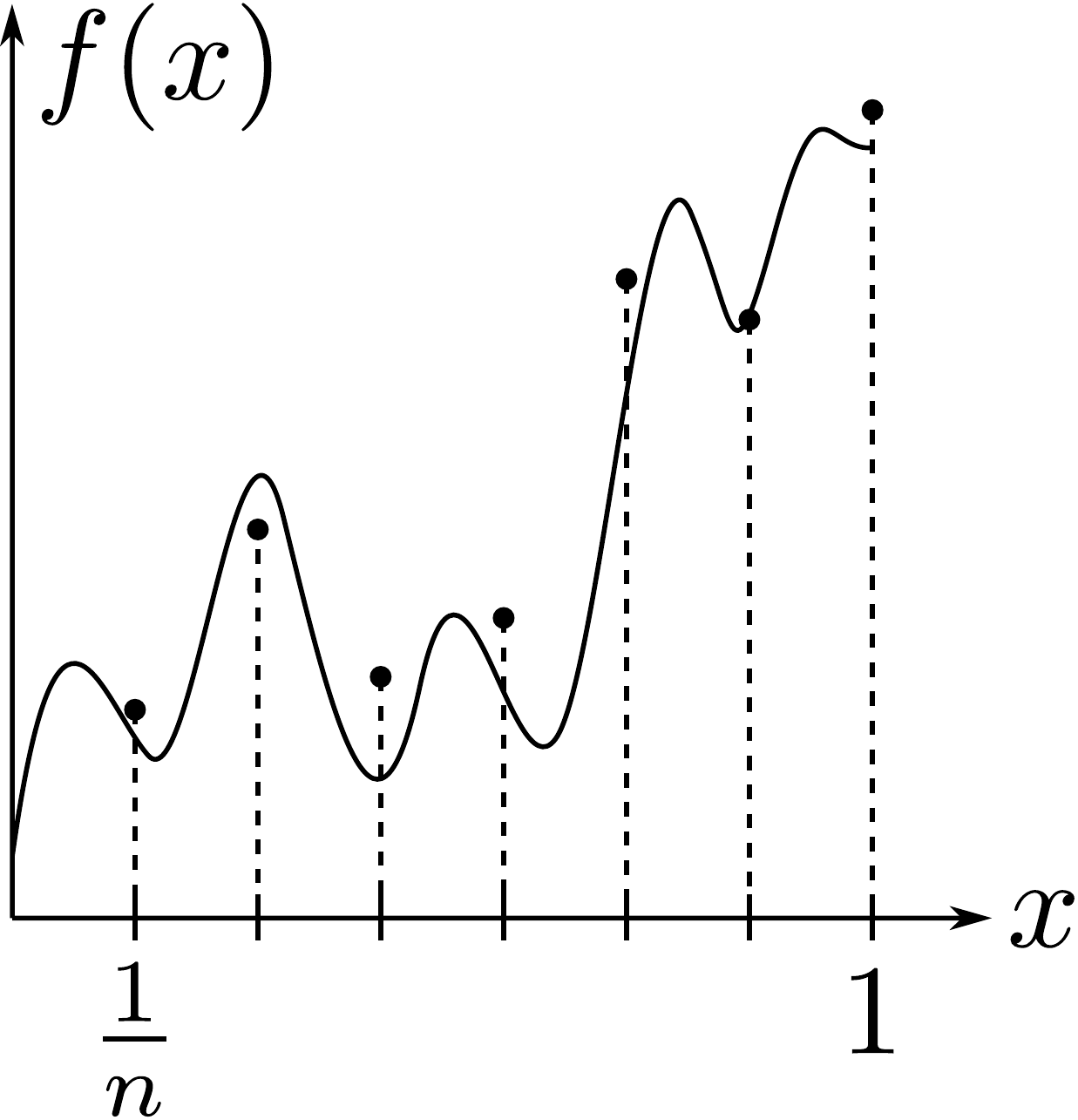}
\caption{Depiction of a discrete-time sampling of the continuous-time
function, $f(x)$. Notice that the additive noise term, $w_i$, prevents
exact knowledge of $f(i/n)$.}
\label{fig:noisy_sampling}
\end{figure}

For the denoising problem, we desire to calculate the original 
function from the noise-corrupted observables $y_i$. How might
we go about doing this? 

\subsection{Least Squares (LS) estimation}

The \textbf{Least Squares Method} dates back to Gauss and Legendre 
\citeyear{gauss1823theoria}.

A natural first idea is to parametrize the function $f$. For instance
we can tentatively assume that it is a degree $p-1$ polynomial
\begin{equation}
	f(t) = \lsum_{j=1}^{p} \theta_j~t^{j-1}.
\end{equation}
Each monomial is weighted according
to coefficient $\theta_j$ for $j \in \left\{ 1, 2, \dots, p \right\}$,
and we will collect these coefficients in a vector $\theta=(\theta_1,\theta_2,\dots,\theta_p)\in\real^p$.
Thus, the problem of recovering $f(t)$ boils down to the recovery of the
$p$ coefficients $\theta_j$ from the set of observables, $y_i$. We
therefore seek to find the set of coefficients which which generate
a function that most closely matches the observed samples.

It is natural to set this up as an optimization problem
(here RSS stands for `residual sum of squares')
\begin{align}
\htheta^{\LS}(y) 	&\equiv \argmin_{\theta}~~\text{RSS}(\theta),\label{eq:LSE1} \\
\text{RSS}(\theta)	&\equiv~\lsum_{i=1}^{n} 
		\left( y_i - \lsum_{j=1}^{p} \theta_j\left( \frac{i}{n} \right)^{j-1}
		\right)^2.
\end{align}
Fitting a low-degree polynomial to a dataset by least squares is a
very common practice, and the reader has probably tried this exercise
at least once. A moment reflection reveals that nothing is
special about the polynomials used in this procedure. In general, we can
consider a set  of functions $\{\varphi_1, \varphi_2, \dots, \varphi_p\}$,
where
\begin{equation}
\varphi_j: 
	[0,1] \rightarrow \real .
\end{equation}

Of course, the quality of our estimate depends on how well the
functions $\{\varphi_j\}$ capture the behavior of the signal $f$. 
Assuming that $f$ can be represented as a linear combination of these
functions, we can rewrite our model as
\begin{align}
y_i &=  \lsum_{j=1}^{p} \theta_{0j} \varphi_j\left(i/n\right)+ w_i.
\end{align}
Equivalently, if we define $\varphi:[0,1]\to\real^p$ by letting
$\varphi(x) = (\varphi_1(x),\varphi_2(x),\dots,\varphi_p(x))$, 
$\theta_0 = (\theta_{0,1},\theta_{0,2},\dots,\theta_{0,p})$,
and denoting by $\<a,b\> \equiv\sum_{i=1}^m a_ib_i$ the usual scalar
product in $\real^m$, we have
\begin{align}
y_i &= \< \theta_0, \varphi(i/n)\> + w_i. 
\end{align}

Before continuing further, it is convenient to pass to matrix
notation. Let us define a matrix 
$\bX \left( X_{ij}\right) \in \real^{n \times p}$
whose entry $i,j$ is given by
\begin{align}
X_{ij} &= \varphi_j \left( \frac{i}{n} \right),
\end{align}
Using this notation, and letting $y= (y_1,y_2,\dots,y_n)$,
$w = (w_1,w_2,\dots,w_n)$,  our model reads
\begin{align}
\label{eq:noisy_measurements}
y = X \theta_0 + w \, ,
\end{align}
$w \sim \normal{0}{\sigma^2 \id_n}$ (here and below $\id_n$ denotes
the identity matrix in $n$ dimensions: the subscript will be dropped if
clear from the context).

From \eqref{eq:noisy_measurements}, we see that vector of observations
$y$ is  approximated as a linear combination of the columns of
$\bX$, each columns corresponding to one of the functions $\varphi_1,
\varphi_2, \dots, \varphi_p $, evaluated on the sampling points.

This is a prototype of a very general idea in statistical learning,
data mining and signal processing. Each data point $x$ (or each point in a
complicated space, e.g. a space of images) is represented by a vector in $\real^p$. This
vector  is constructed by evaluating $p$ functions at $x$ hence
yielding the vector
$(\varphi_1(x), \varphi_2(x),\dots,\varphi_p(x))$. Of course, the
choice suitable functions $\{\varphi_j\}$ is very important and  domain-specific.

The functions $\{\varphi_j\}$ (or --correspondingly-- the columns of
the matrix $\bX$) have a variety of names. They  are known as
``covariates" and ``predictors''
in statistics, as ``features" in the context of machine learning
and pattern recognition. The set of features $\{\varphi_j\}$ 
is sometimes called a ``dictionary,'' and the matrix $\bX$ is also referred to as the
``design matrix.''
 Finding an appropriate
set of features, i.e. ``featurizing", is a problem of its own.
The observed $y_i$ are commonly 
referred to as the ``responses" or ``labels" within statistics and
machine-learning, respectively. The act of finding the true set of
coefficients $\theta_0$ is known as both ``regression" and 
``supervised learning".

So, how do we calculate the coefficients $\theta_0$ from $y$? Going
back to least squares estimation, we desire to find a set of coefficients,
$\htheta$ which best match our observations. Specifically, in matrix
notation \eqref{eq:LSE1} reads
\begin{equation}
	\htheta^{\LS}= \argmin_{\theta \in \real^p}~~\mathcal{L}(\theta),
\end{equation}
where 
\begin{align}
	\mathcal{L}(\theta) &= \frac{1}{2n} \| y - \bX \theta \|_2^2,\nonumber \\
	&= \frac{1}{2n} \lsum_{i=1}^{n}\left( y_i - \left<x_i, \theta \right>\right)^2,
\end{align}
with $x_i$ the $i$-th row of $\bX$. Here and below $\|a\|_2$ denotes
the $\ell_2$-norm of vector $a$: $\|a \|_2^2 =\lsum_{i} a_i^ 2$.\index{Norm $\ell_2$}
The minimizer can be found by noting that
\begin{align}
	\nabla\mathcal{L} (\theta) &= -\frac{1}{n} \bX^\sT (y - \bX \theta), \\
	\therefore~~~~\htheta^{\LS} &= \left( \bX^\sT \bX \right)^{-1} \bX^\sT y .  \label{eq:ls_estimator}	
\end{align}

Looking at \eqref{eq:ls_estimator}, we note that an important role is
played by the sample covariance matrix
\begin{equation}
	\hSigma = \frac{1}{n} \bX^T \bX .
\end{equation}
This is the matrix of correlations of the predictors
$\{\varphi_1,\dots,\varphi_p\}$.
The most immediate remark  is that, for $\htheta^{\LS}$ to be well defined,
$\hSigma$ needs to be invertible, which is equivalent to require
$\rank{\bX} = p$. This of course can only happen if the number of
parameter is no larger than the number of observations: $n\le p$.
Of course, if $\hSigma$ is invertible but is nearly-singular, then
$\htheta$ will be very unstable and hence a poor estimator.
A natural way to quantify the `goodness' of $\hSigma$ is through its
condition number $\kappa(\hSigma)$, that is the ratio of its largest
to its smallest  eigenvalue: $\kappa(\hSigma) = \lambda_{\rm
  max}(\hSigma)/\lambda_{\rm min}(\hSigma)$.  From this point of view,
an optimal design has minimal condition number $\kappa(\hSigma) = 1$,
which corresponds to $\bX$ to be proportional to an orthogonal matrix.
In this case $\bX$ is called an `orthogonal design' and we shall fix
normalizations by assuming $\hSigma = (\bX^{\sT}\bX/n) = \id_p$

In functional terms, we see that the LS estimator is calculated according to
the correlations between $y$ and the predictors,
\begin{align}
	\widehat{\Sigma}_{jl} &= \frac{1}{n} \lsum_{i=1}^{n} 
		\varphi_j \left(i/n\right) \varphi_l \left(i/n\right), \\
	\therefore~~~~ \htheta^{\LS}_l&= \lsum_{j=1}^p\left( \widehat{\Sigma}^{-1} \right)_{lj}
		\left( \frac{1}{n} \lsum_{i=1}^{n} \varphi_j (i/n ) y_i \right).
\end{align}

\subsection{Evaluating the estimator}

Now that we calculated the LS estimator for our problem, a natural 
question arises: is this indeed the best estimator we could use? In
order to answer
this question, we need a way of comparing one estimator to another. 

This is normally done by considering the \emph{risk function}
associated with the estimator. 
If the model depends on a set of parameters $\theta\in \real^p$,
the risk function is a function $R:\real^p\to\reals$, defined by
\begin{align}
	\label{eq:param_risk}
	R(\theta) &= \expval{\| \htheta(y) - \theta\|_2^2}, \nonumber \\
	&= \lsum_{j=1}^{p} \expval{(\hat{\theta}_j(y) - \theta_j)^2  }.
\end{align}
Here expectation is taken with respect to $y$, distributed according
to the model \eqref{eq:noisy_measurements} with $\theta_0=\theta$.
Note that the $\ell_2$-distance  is used to measure the estimation error.\index{Norm $\ell_2$}

Other measures (called `loss functions') could  be used as well, but
we will focus on this for the sake of concreteness. We can also calculate
risk over the function space and not just over the parameter space.
This is also  known as the `prediction error':
\begin{align}
	\label{eq:functional_risk}
	\Rp(\theta) &= \frac{1}{n} \lsum_{i=1}^{n} \expval{\left( 
						\widehat{f}\left(i/n\right) 
						- f\left(i/n \right) \right)^2}, \nonumber\\
				&= \frac{1}{n} \lsum_{i=1}^{n} \expval{\left[ 
						\lsum_{j=1}^{p} \bX_{ij} \left( \hat{\theta}_j - \theta_j \right)
						\right]^2}, \nonumber\\
				&= \frac{1}{n} \expval{\|\bX \left( \hat{\theta} - \theta \right)\|_2^2}.
\end{align}
In particular, for $\bX$ an orthogonal design, 
$\Rp(\theta) = c\, R(\theta)$. 

Let us apply this definition of risk to the LS estimator, 
$\hat{\theta}^{\LS}$. Returning to the signal sampling model,
\begin{align}	y &= \bX \theta_0 + w. \\
	\therefore~~~~\htheta^{\LS} &= \left( \bX^{\sT} \bX \right)^{-1} \bX^{\sT} y, \nonumber\\
	&= \theta_0 + (\bX^{\sT} \bX)^{-1} \bX^{\sT} w,
\end{align}
which shows that the LS estimator will return the true parameters, $\theta_0$,
perturbed by some amount due to noise. Now, we will calculate the risk
function
\begin{align}
	R(\theta) &= \expval{\| \htheta^{\LS}(y) - \theta_0 \|_2^2}, \nonumber\\
	&= \expval{\| \left( \bX^{\sT} \bX \right)^{-1} \bX^{\sT} w \|_2^2}, \nonumber\\
	&= \expval{w^{\sT} \bX \left( \bX^{\sT} \bX \right)^{-2} \bX^{\sT} w}, \nonumber\\
	&= \sigma^2 \trace{ \bX \left( \bX^{\sT} \bX \right)^{-2} \bX^{\sT} }, \nonumber\\
	&= \sigma^2 \trace{ \left( \bX^{\sT} \bX \right)^{-1}}, \nonumber\\
	&= \frac{\sigma^2 p}{n} \left[ \frac{ \trace{\hSigma^{-1}} }{p} \right],
\end{align}
where we add the $p$ term to the final result because we expect that 
$\frac{1}{p} \trace{\hSigma^{-1}}$ to be on the order
one, under the assumption of near-orthonormal predictors. 

Tho further illustrate this point, let us consider the case in which
the functions $\{\varphi_j\}$ are orthonormal (more precisely, they
are an orthonormal set in $L^2([0,1])$). This means that 
\begin{equation}
	\int_0^1 \varphi_i(x) \varphi_j(x) \, \de x = \delta_{ij}\, .
\end{equation}
where $\delta{ij}$ is $1$ when $i=j$ and $0$ for all $i\neq j$.
For $n$ large, this implies
\begin{equation}
	\hSigma_{jl} = \frac{1}{n} \lsum_{i=1}^{n} 
						\varphi_j \left( i/n \right)
						\varphi_l \left( i/n \right) 
						\approx \delta_{jl},	
\end{equation}
where we assumed that the sum can be approximated by an integral.
In other words, if the functions $\{\varphi_j\}$ are orthonormal, the
design is nearly orthogonal, and this approximation gets better as the
number of samples increases. Thus, in such good conditions,
$\trace{\widehat{\Sigma}^{-1}} \approx p$. Under these conditions, we can simplify
the risk function for the LS estimator in the case of orthonormal or near-orthonormal
predictors to be,
\begin{equation}
	\label{eq:ls_risk}
	R(\theta) \approx \frac{p \sigma^2}{n}.
\end{equation}

This result has several interesting properties:
\begin{itemize}
\item The risk is proportional to the noise variance. This makes
  sense: the larger is  noise, the worst we can estimate the function.
\item It is inversely proportional to the number of samples $n$: the
  larger is the number of observations, the better we can estimate
  $f$.
\item The risk is proportional to the number of parameters $p$.  
This fact can be interpreted as an over-fitting phenomenon
If we choose a large $p$, then our estimator will be more sensitive to
noise. Conversely, we are effectively searching the right parameters in a
higher-dimensional space, and a larger number of samples is required
to determine it  the same accuracy.
\item The risk $R(\theta)$ is independent of $\theta$. This is closely related to the
e linearity of the LS estimator.
\end{itemize}

At this point, two questions arise naturally:
\begin{enumerate}
\item[{\sf Q1.}] Is this the best that we can do? Is it possible to 
use a different estimator and decrease risk?
\item[{\sf Q2.}] What happens if the function to be estimated is function is not 
\textit{exactly} given by a linear combinations of the predictors?
Indeed in general, we cannot expect to have a perfect model for the
signal, and any set of predictors is only approximate:
\begin{equation}
	f(t) \neq \sum_{j=1}^p \theta_j \varphi_j(t).
\end{equation}
\end{enumerate}

To discuss these two issues, let us first change the notation for the
risk function, by making explicit the dependence on the estimator $\htheta$:
\begin{align}
	R(\theta;\htheta)\quad\text{where}\quad &\theta \in \real^{p},\\
	& \hat{\theta}: \real^n \rightarrow \real^p. \notag
\end{align}
Note that estimators, $\htheta$, are functions of $y\in\reals^n$.
For  two estimators, $\htheta^1$ and $\htheta^2$, we can compare 
$R(\theta;\htheta^1)$ and $R(\theta;\htheta^2)$. 
This leads to the next, crucial question: how do we compare these two
curves? For instance, in Figure \ref{fig:compare_risk} we sketch two
cartoon risk functions $R(\theta;\htheta^1)$ and $R(\theta;\htheta^2)$. 
Which one is the best one? The way this question is answered has
important consequences.

Note that naively, one could hope to find an estimator that is
\emph{simultaneously} the best at all points $\theta$.
Letting this ideal estimator be denoted by
$\htheta^{\text{opt}}$, we would get
\begin{align}
	R(\theta;\htheta^{\text{opt}}) \leq R(\theta;\htheta) 
	\quad \forall~\theta, \htheta.
\end{align}
However, assuming the existence of such an ideal estimator 
leads to a contradiction. To see this, we will let the predictors be,
for simplicity
\begin{equation}
	\bX = \sqrt{n}\id_{n} \, , \quad \quad p = n,
\end{equation}
which means our regression problem is now
\begin{equation}
	y = \theta + \frac{w}{\sqrt{n}}, \quad \quad w \sim \mathcal{N}(0,\sigma^2 \id_{n}).
\end{equation}
Note that in this case the LS estimator is simply $\htheta^{\LS}(y) =
y$. 

Next, fix $\xi\in\reals^p$, and consider the oblivious estimator that
always returns $\xi$:
\begin{equation}
	\htheta^{\xi}(y) = \xi\, .
\end{equation}
This has the  risk function
\begin{equation}
	R(\theta;\htheta^{\xi}) = \|\xi - \theta\|_2^2.
\end{equation}
If an `ideal' estimator $\htheta^{\text{opt}}$ as above existed, it
would beat $\htheta^{\xi}$, which implies in particular
\begin{align}
R(\xi;\htheta^{\text{opt}})  =0\, .
\end{align}
Since $\xi$ is arbitrary, this would imply that the ideal estimator
has risk everywhere equal to $0$, i..e. always reconstruct the true
signal perfectly, independently of the noise. This is of course impossible.

 One approach
would be to evaluate the Bayes risk, which would compute the expected value 
of each risk curve dependent upon the prior distribution of the parameters,
$\text{Pr}[\theta]$. However, it is not clear in every case how one might 
determine this prior, and its choice can completely skew the comparison between
$\hat{\theta}^1$ and $\hat{\theta}^2$.

\begin{figure}
\centering
	\includegraphics[width=\figurewidth]{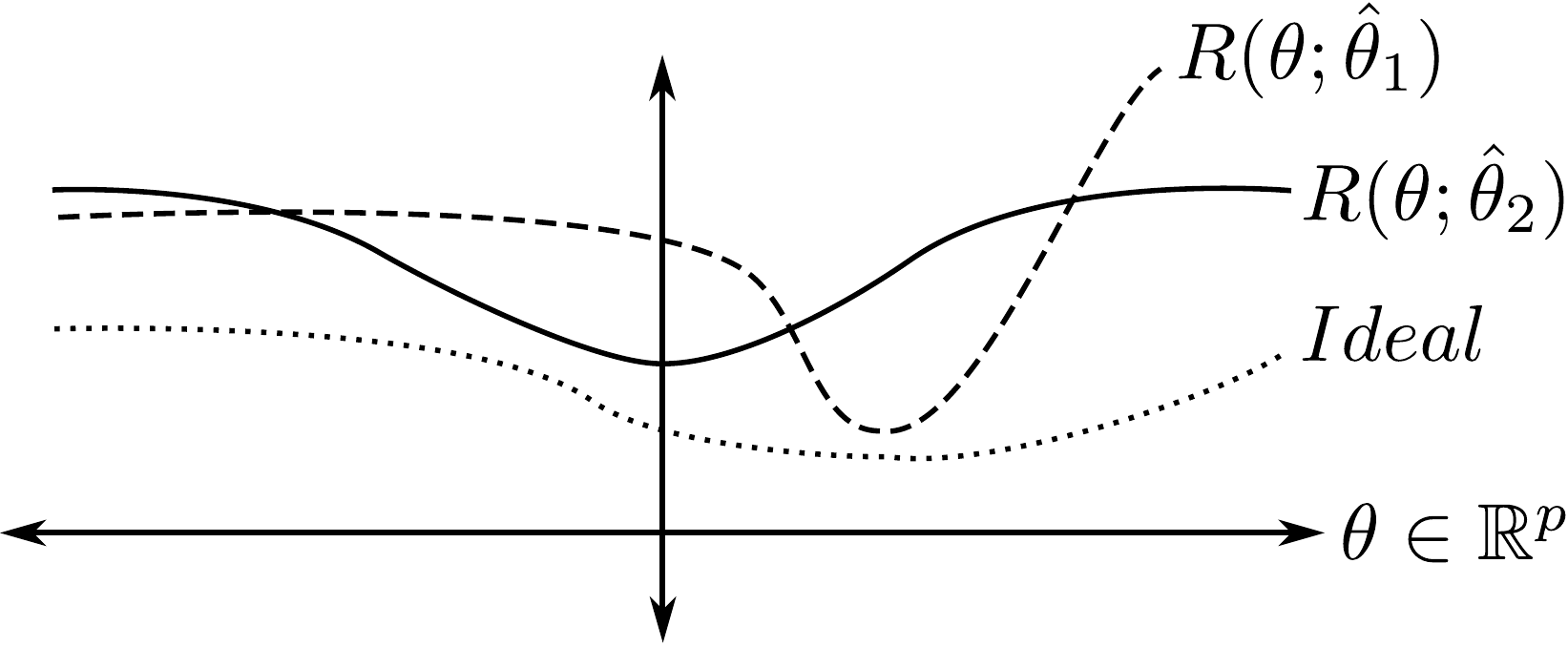}
\caption{Comparing the risk functions of two different estimators, $\hat{\theta}^1$,
$\hat{\theta}^2$ over the space of possible parameters, $\theta$. Also shown is a risk function for some estimator which is \textit{ideal} in the sense that it is below both 
both of the known estimators for all $\theta$.}
\label{fig:compare_risk}
\end{figure}

\begin{figure}
\centering
	\includegraphics[width=\figurewidth]{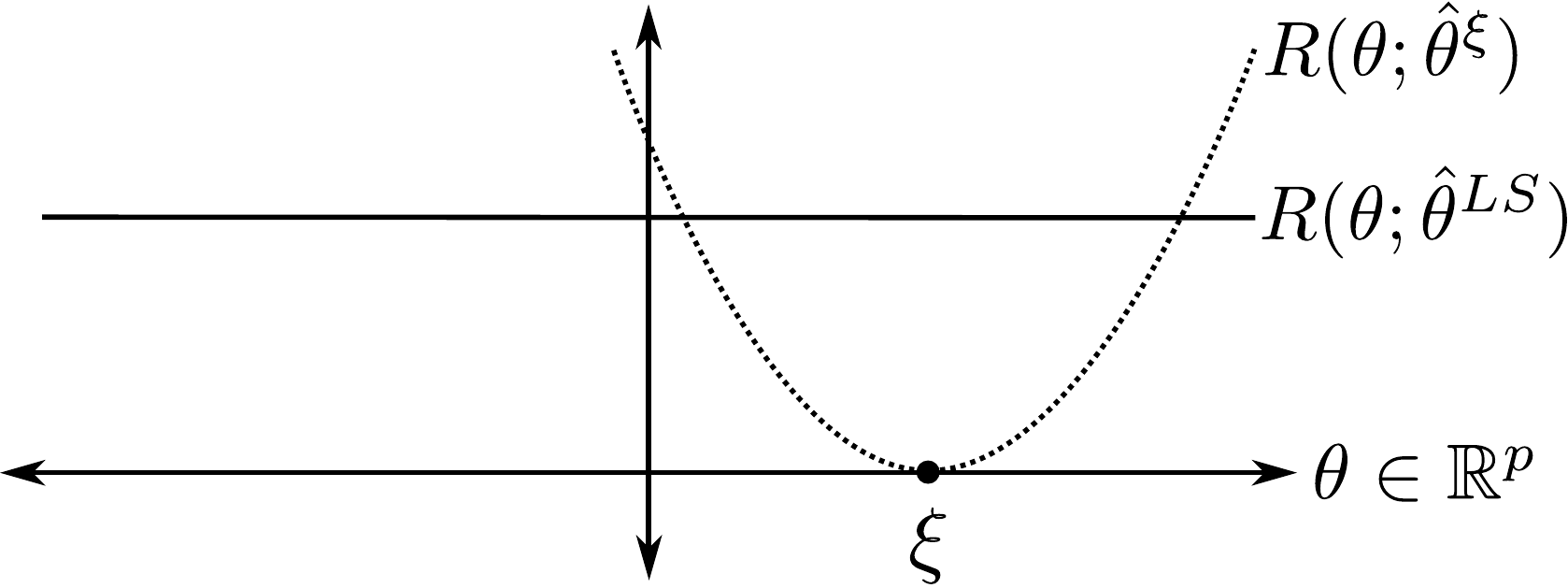}
\caption{Any arbitrary trivial estimator can minimize the point-wise
risk at point $\xi$.}
\label{fig:trivial_risk}
\end{figure}

 One approach to overcome this problem is to evaluate for each risk
 function the corresponding `Bayes risk.'  This amounts to averaging
 $R(\theta;\htheta)$ over $\theta$,  using a certain  prior distribution of the parameters
$\text{P}(\theta)$. Namely
\begin{align}
R_B(\P;\htheta) \equiv\int~R(\theta;\htheta) \; \P(\de\theta)  .
\end{align}
However, it is not clear in every case how one might 
determine this prior. Further, the choice of $\text{P}(\theta)$ can completely skew the comparison between
$\htheta^1$ and $\htheta^2$. If $\text{P}(\theta)$  is concentrated in
a region in which --say-- $\htheta^1$ is superior to $\htheta^2$, then $\htheta^1$
will obviously win the comparison, and viceversa.

In the next section we shall discuss the minimax  approach to
comparing estimators.

\section{Nonlinear denoising and sparsity}
\subsection{Minimax risk}

The previous lecture discussed estimating a set of parameters, $\theta$,
given the linear model
\begin{equation}
	y = \bX \theta + w.
\end{equation}
In this discussion, we stated that there exists no estimator which
dominates all other possible estimators in terms of risk,
$R(\theta;\htheta)$.
Still the question remains of  how to compare two different estimators
$\htheta^1$, $\htheta^2$.

\begin{figure}
\centering
	\includegraphics[width=\figurewidth]{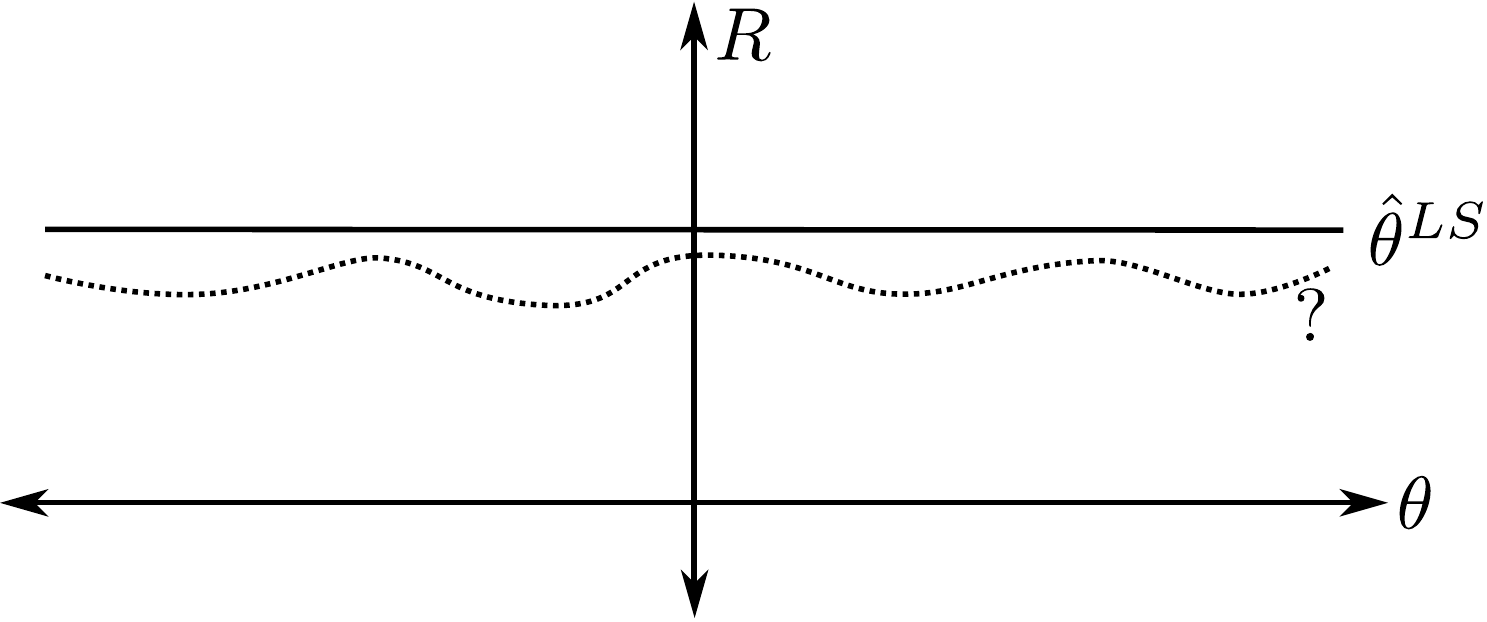}
\caption{The risk of an unknown estimator which dominates $\hat{\theta}^{LS}$ for
all $\theta$.}
\end{figure}

A fruitful approach to this question is to consider the worst case risk
over some region $\Omega \subseteq \real^{p}$.
Formally, we define the \emph{minimax risk} of $\htheta$ over $\Omega$ as
\begin{equation}
	R_*(\Omega;\htheta) = \sup_{\theta \in \Omega}~R(\theta;\htheta).
\end{equation}
Such a definition of risk is useful if we have some knowledge \textit{a priori}
about the region in which the true parameters live. The minimax risk allows us to
compare the maximal risk of a given estimator over the set $\Omega$ to find an estimator
with minimal worst case  risk. The minimax risk is also 
connected to the Bayes risk which we defined earlier,
\begin{align}
R_B(\P;\htheta)& \equiv\int~R(\theta;\htheta) \; \P(\de\theta)
\, ,\\
R_*(\Omega;\htheta) &= 
		\sup_{\text{supp}(\P) \subseteq \Omega} R_B (\P;\htheta).
\end{align}
With this definition of minimax risk, it is easy  do compute
the minimax risk of least squares
\begin{equation}
	R_*(\real^{p};\htheta^{\LS}) = \frac{p\sigma^2}{n} \left[
          \frac{\text{Tr}(\hat{\Sigma}^{-1})}{p} \right], .
\end{equation}

The least squares estimator is optimal in minimax sense.
\begin{theorem}
The least squares estimator is minimax optimal over $\real^p$.
Namely, any estimator $\htheta$ has minimax risk $R_*(\real^{p};\htheta) \ge R_*(\real^{p};\htheta^{\LS})$. 
\end{theorem}
\begin{proof}
The proof of this result relies on the connection with Bayes risk.
Consider for the sake of simplicity the case of orthogonal designs,
$\hSigma = \id$.
It is not hard to show that, if $\P$ is gaussian with mean $0$ and
covariance $c^2\id_p$, then
\begin{align}
\inf_{\htheta}R_B(\P;\htheta)& = \frac{c^2\sigma^2}{c^2+\sigma^2}\, .
\end{align}
Hence, for any estimator $\htheta$
\begin{align}
\sup_{\theta\in\reals^p}R(\theta;\htheta) \ge R_B(\P;\htheta) \ge
\frac{c^2\sigma^2}{c^2+\sigma^2}\, .
\end{align}
Since $c$ is arbitrary, we can let $c\to\infty$, whence
\begin{align}
R_*(\reals^p;\htheta) \ge \sigma^2\, .
\end{align}
A full treatment of a more general result can be found, for instance
in \cite[Chapter 7]{wasserman2006all}.
\end{proof}
A last caveat. One might suspect --on the grounds of the last
theorem-- that least squares estimation is optimal `everywhere' in
$\reals^p$. This was indeed common belief among statisticians until 
the surprising discovery of the `Stein phenomenon' in the early
sixties \cite{james1961estimation}.
In a nutshell, for $p\ge 3$ there exist estimators that have risk
$R(\theta;\htheta)<R(\theta;\htheta^{\LS})$
strictly for every $\theta\in\reals^p$! (The gap vanishes as
$\theta\to\infty$.)
We refer to \cite[Chapter 7]{wasserman2006all} for further background
on this.

\subsection{Approximation error and the bias-variance tradeoff}

Until now we have assumed that the unknown function$f(t)$, could be
exactly represented by the set of predictors, corresponding to columns
of $\bX$. How is our ability to estimate
the parameters set $\theta$, and thus $f(t)$, affected when this
assumption is violated?

In order to study this case, we assume that we are given an infinite
sequence of predictors $\{\varphi_j\}_{j\ge 1}$, and use only the
first $J$ to estimate $f$.
For any fixed $J$, $f(t)$ can be approximated as a linear combination
of the 
first $J$ predictors, plus an
error term which is dependent upon $J$
\begin{equation}
	f(t) = \lsum_{j=1}^{J} \theta_{j} \varphi_j(t) + \Delta_J(t).
\end{equation}
For a complete set $\{\varphi_j\}$, we can ensure
$\lim_{J\to\infty}\|\Delta_J\|=0$ in a suitable norm. 
This can be formalized by assuming $\{\varphi_j\}_{j\ge 1}$ to be a
orthonormal basis in the  Hilbert space $L^2([0,1])$ and the above to
be the orthonormal decomposition. 
In particular, the remainder will be orthogonal
to the expansion,
\begin{align}
	\int_0^1~\Delta_J(t) \varphi_j(t) \, \de t = 0 \quad \forall
        j\in\{1,\dots,J\}\, .\
\end{align}
Alternatively,  we can require orthogonality with respect to the
sampled points (the resulting expansions are very similar for $n$ large)
\begin{align}
	\frac{1}{n} \lsum_{i=1}^{n} \Delta_J\left( i/n\right)
        \varphi_j\left( i/n\right) = 0 \quad \forall j\in\{1,2,\dots,J\}
     .
\end{align}

With the remainder $\Delta_J$, our regression model becomes
\begin{equation}
	y = \bX \theta_0 + \Delta_J + w,
\end{equation}
where $\bX \in \real^{n \times J}$, $w \sim \normal{0}{\sigma^2 I_{n}}$, and $\bX^T \Delta_J = 0$. 
Recall, the LS estimator is given by
\begin{align}
	\htheta &= (\bX^T \bX)^{-1} \bX^T y,\nonumber\\
	 &= \theta_0 + (\bX^T \bX)^{-1} \bX^T w.
\end{align}
We can compute the prediction risk as follows
\begin{align}
	R_p(f) &=~\frac{1}{n}~\expval{\lsum_{i=1}^n \Big(\hat{f}\left(i/n\right) - f\left( i/n \right) \Big)^2}, \notag\\
	&=~\frac{1}{n}~\expval{ \| \bX \htheta - \bX \theta_0 - \Delta_J\|_2^2},\notag\\
	&=~\frac{1}{n}~\expval{ \| \bX(\theta - \htheta_0)\|_2^2} + \frac{1}{n} \|\Delta_J \|_2^2, \notag\\
	&=~\frac{1}{n}~\expval{ \| \bX (\bX^T \bX)^{-1} \bX^T w \|_2^2} + \frac{1}{n} \|\Delta_J \|_2^2, \notag\\
	&=~\frac{1}{n}~\| \Delta_J \|_2^2 + \frac{\sigma^2}{n}\trace{\bX (\bX^T \bX)^{-1} \bX^T}.
\end{align}
Finally, note that  $\bX (\bX^T \bX)^{-1} \bX^T\in\reals^{n\times n}$
is the orthogonal projector o the space spanned by the columns of $\bX$.
Hence its trace is always equal to $J$.This gives us the final form of the
estimation risk at $f(\,\cdot\,)$, as a function of $J$,
\begin{align}
	R_p(f) = \frac{\|\Delta_J\|_2^2}{n} + J \frac{\sigma^2}{n}.
	\label{eq:var_bias}
\end{align}
In other words, the estimation risk associated with $f(\,\cdot\, )$ is a sum of
two terms, 
both of which are dependent upon the choice of $J$:
\begin{itemize}
\item The first term is associated with the approximation error
  induced by the choice of the predictor $\{\varphi_j\}$. It is
  independent of the noise variance $\sigma^2$, decreases with $J$.

We interpret it therefore a \emph{bias} term. 
\item The second term depends on the noise level, and is related to the
  fluctuations that the noise induces in $\htheta$. It increases with the
  number of predictors $J$, as the fi becomes more unstable.

We interpret it therefore as a \emph{variance} term.
\end{itemize}
Therefore, the optimal number of predictors to use, $J^*$ is the one which minimizes the risk
by striking a balance between bias and variance.
In other words,  we want to
find the optimal point between under- and over-fitting the model. 
Note that how we choose the predictors themselves determines the rate at which the bias term
goes to zero as $J$ increases. 

\begin{figure}
\centering
	\includegraphics[width=\figurewidth]{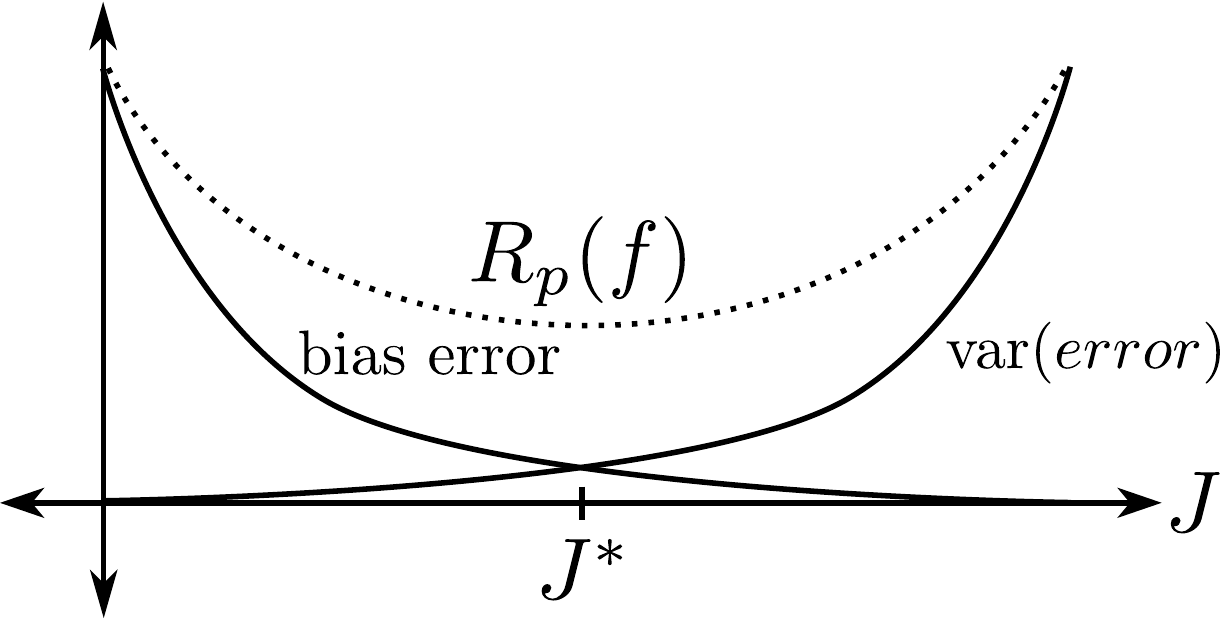}
\caption{The effect of bias and variance on estimation risk.}
\end{figure}

\subsubsection{Example: The Fourier basis}
In this example we select our set of predictors to be a Fourier basis 
\begin{equation}
	\varphi_j(t) = \sqrt{2} \cos((j-1)\pi t),
\end{equation}
for $t \in [0,1]$. If $f(\,\cdot\,)$ is square-integrable, then it can
be represented as an infinite series (converging in $L^2([0,1])$,
i.e. in mean square error)
\begin{equation}
 	f(t) = \lsum_{j=1}^{\infty} \theta_{0j} \varphi_j(t).
 \end{equation} 
 However, if only $J$ sinusoids are used, the remainder, $\Delta_J$ is
 \begin{equation}
 	\Delta_J(t) = f(t) - \lsum_{j=1}^{J} \theta_{0j} \varphi_j(t).
 \end{equation}
And, finally, the squared norm of $\Delta_J$ is, by orthogonality
\begin{align}
	\| \Delta_J \|^2_2 & = \sum_{i=1}^n \Delta_J(i/n)^2,\nonumber \\
& \approx n\int_{0}^2\Delta_J(t)^2 = n\int \lsum_{j=J+1}^{\infty} \theta_{0j}^2.
	\label{eq:fourier_remainder}
\end{align}
Here we replaced the sum by an integral, an approximation that is
accurate for $n$ large.
 
Note that the decay of the bias term mirrors the decay of the Fourier
coefficients of $f$, by  \eqref{eq:fourier_remainder}. In particular,
if  $f$ is smooth, its Fourier coefficients decay faster, and hence 
 the bias decays rapidly with $J$. In this case, the Fourier basis is
 a good set of features/predictors (a good dictionary) for our problem.

\begin{figure}
\centering
	\includegraphics[width=\figurewidth]{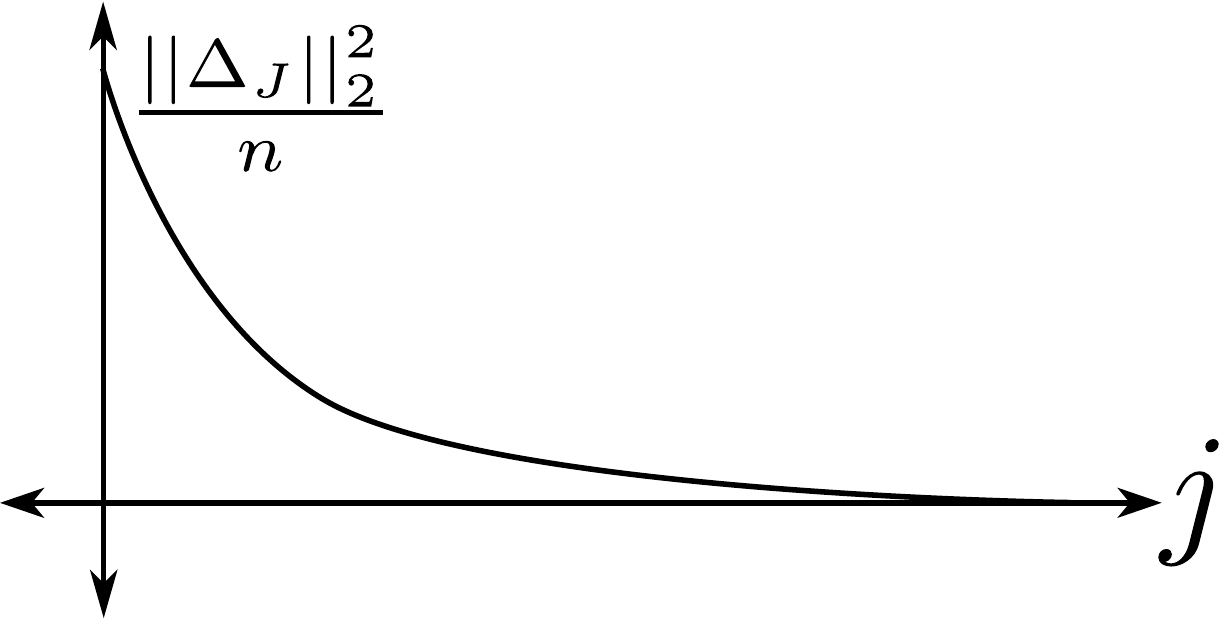}
\caption{Depiction of the rate of decay of the bias term as a function of $J$.}
\end{figure}

We now look at the case of Fourier predictors for a specific class od
smooth functions, namely functions whose second derivative is square
integrable.
Formally we define
\begin{align}
 W(C) \equiv	\left\{ \int_0^1 \left(  f^{\prime \prime} \left( t
     \right)\right)^2 \, \de t  \leq C^2 \right\} \, ,
\end{align}
and we will consider estimation over $\Omega = W(C)$. This space is
known in functional analysis as the 
`Sobolev ball of radius $C$ and order $2$.'

In terms of  Fourier coefficients,  this set of smooth functions can
be characterized as
\begin{align}
	\int_0^1 \left( \lsum_{j=1}^{\infty} \pi^2 (j-1)^2 \theta_{0j}
          \varphi_j(t) \right)^2 \de t \leq C
\end{align}
Or equivalently
\begin{align}
	\lsum_{j=1}^{\infty} \pi^4 (j-1)^4 \theta_{0j}^2 \leq C. \label{eq:bounded_second}
\end{align}
Now, in order for this to happen, we must have
\eqref{eq:bounded_second} is satisfied,
\begin{equation}
	\sum_{j=J+1}^{\infty}\theta_{0j}^2 \lesssim \frac{C^{\prime}}{J^4},
\end{equation}
and hence,
we can estimate a bound on the squared norm of the remainder term
\begin{equation}
	\| \Delta_J \|_2^2 \approx n \sum_{j = J+1}^{\infty} \theta_{0j}^2
	\lesssim \frac{nC'}{J^4}.
\end{equation}
Therefore, the prediction risk for the set of functions $\Omega = W(C)$ is upper 
bounded
\begin{equation}
	R_{p,*}(\Omega;\htheta)\lesssim \frac{C}{J^4} + \frac{\sigma^2 J}{n}.
\end{equation}
The optimum value of $J$ is achieved when the two terms are of the
same order, or by setting to $0$ the derivative with respect to $J$
\begin{align}
	&\frac{\partial\phantom{J} }{\partial J} \left\{ \frac{C}{J^4} + \frac{\sigma^2 J}{n}\right\} = 
	\frac{\sigma^2}{n} - \frac{4C}{J^5}, \\
	&~~~\therefore~~J_* \sim \left(  \frac{n}{\sigma^2} \right)^{1/5}.
\end{align}
Finally, with the optimal choice of $J$, we obtain the upper bound for the
prediction risk, in general,
\begin{equation}
	R_{p,*}(\Omega;\htheta) \lesssim \left( \frac{\sigma^2}{n}
        \right)^{4/5}\, .
\end{equation}
As in the standard parametric case, see  \eqref{eq:ls_risk}, the risk depends on the ratio of the noise
 variance to the number of samples. However the decay with the
        number of samples is slower: $n^{-4/5}$ instead of
        $n^{-1}$. This is the price paid for not knowing in advance
        the $p$-dimensional space to which $f$ belong. It can be
        proved that the exponent derived here is optimal. 

\subsection{Wavelet expansions}

As we emphasized several times, the quality of our function
estimation procedure is highly dependent on the choice of the features
$\{\varphi_j\}$.
More precisely, it depends on the ability to represent the signal of
interest with a few elements of this dictionary.
While the Fourier basis works well for smooth signals, 
it is not an adequate dictionary for many signals of interest.
For instance, the Fourier expansion does not work very well for
images.

 Why is this the case? Let us reconsider what are the LS estimates for
 the Fourier coefficients. Using orthonormality of the Fourier basis,
 we have
\begin{equation}
\htheta^{\LS}_j \approx \frac{1}{n} \lsum^{n}_{i=1}
        \varphi_j\left( i/n \right) y_i = \theta_{0,j}
        +\widehat{w}_j\, ,
\end{equation}
where $\widehat{w}_i =n^{-1} \lsum^{n}_{i=1}
        \varphi_j\left( i/n \right) y_i$. Figure
        \ref{fig:noise_spectrum} shows a cartoon of these coefficients.

\begin{figure}
\centering
	\includegraphics[width=\figurewidth]{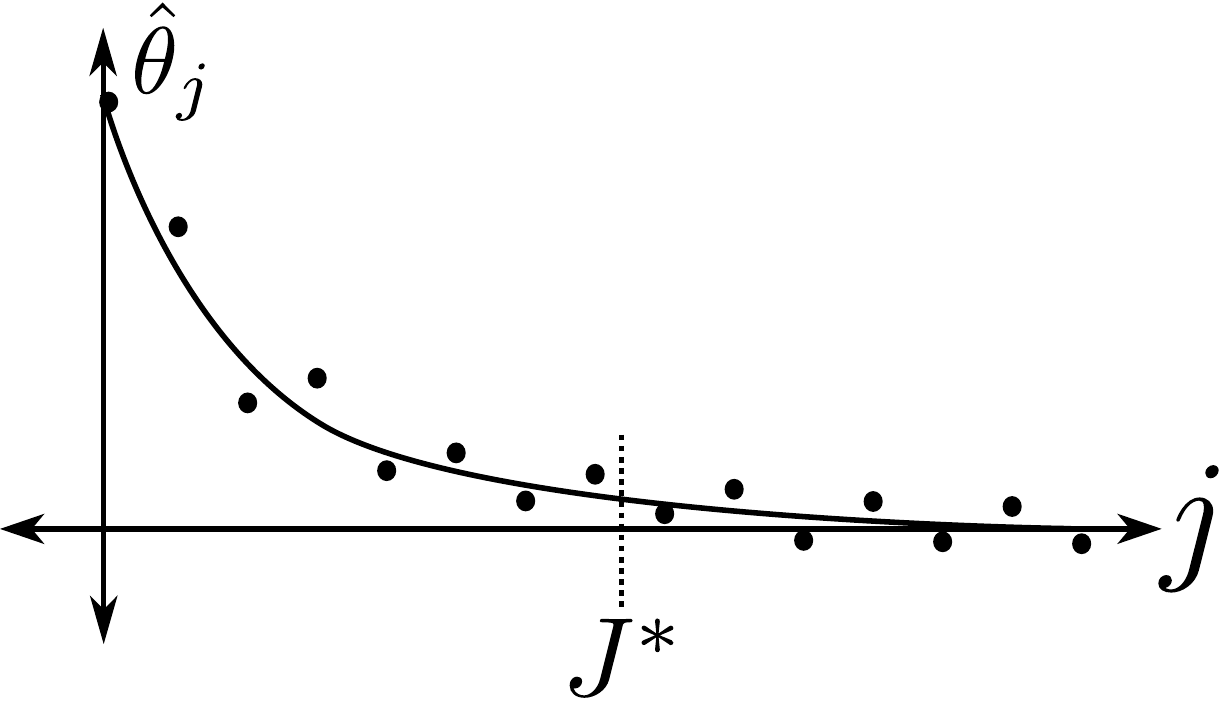}
\caption{Cutting off coefficients at $J^*$.}
\label{fig:noise_spectrum}
\end{figure}
In other words, each estimated coefficients is a sum of two
contributions:
the true Fourier coefficients $\theta_{0,j}$ and the noise
contribution $\widehat{w}_j$. Since the noise is white, its energy is 
equally spread across all Fourier modes. 
On the other hand, if the signal is smooth, its energy concentrates on
low-frequency modes.
By selecting a cut-off at $J^*$, we are sacrificing some true signal
information, in order to get rid of most of the noise.
For frequencies higher than $J^*$, the noise energy surpasses
any additional information these coefficients contain about the original signal
we wish to estimate.

In other words, by selecting $J=J_*$, we are filtering out high
frequencies in our measurements.
In `time' domain, this is essentially  equivalent to averaging the
observations over a sliding window of size of order $J_*^{-1}$.
Formally, this is done by convolving 
the observations $y$, with some smooth kernel $K(\,\cdot\, )$,
\begin{align}
	\hat{f}(t) = \frac{1}{n} \lsum_{i=1}^{n} K \left( \frac{i}{n} - t \right) y_i.\label{eq:Convolution}
\end{align}
The details of the kernel $K(\,\cdot\,)$ can be worked out in detail,
but what is important here is that it is $K(s)$ significantly
different from $0$ if and only if $|s|\le J_*^{-1}$. This point of
view gives a different perspective on the bias-variance tradeoff:
\begin{itemize}
\item For
small $J_*$, we are averaging over large window, and hence reducing the
variance of our estimates. On the other hand, we are introducing a
large bias in favor of smooth signals.
\item For large $J_*$ we average over a large window. The estimate is
  less biased, but has a lot of variance.
\end{itemize}

For truly smooth signals, this approach to denoising is adequate.
However, for many signals, the degree of smoothness changes
dramatically from one point to the other of the signal. For instance,
an image is mostly smooth, because of homogeneous surfaces
corresponding to the same object or degree of illumination.
However, they contain a lot of important discontinuities (e.g. edges)
as well. Missing  or smoothing out edges has a dramatic impact on the
quality of reconstruction.

Smoothing with a kernel with uniform width produces a very bad
reconstruction on
such signals. If the width is large the image 
becomes blurred across these edges. If the width is small, it will not
filter out noise efficiently in the smooth regions.
A different
predictor set must be used that adapts to different levels  of
smoothness in different point of the image.

Wavelets are one such basis. A wavelet expansion of a function allows for localization of
frequency terms, which means high-frequency coefficients can be localized to edges, while
smoother content of the image can be more concisely described with just a few low-frequency
coefficients. Wavelets, as in our previous example of the Fourier basis, are an orthonormal
basis of $[0,1]$. The expansion is formed via two functions, the \textit{father-wavelet}, or \textit{scaling}, function, $\varphi(\cdot)$,
and the \textit{mother-wavelet} function, $\psi_{jk}(\cdot)$. The mother-wavelet function is 
used to generate set of self-similar functions which are composed of scaled and shifted versions
of the mother-wavelet,
\begin{align}
	\psi_{jk}(t) = 2^{j/2} \psi(2^j t - k) \quad \text{where} &\quad j \in \{0,1,2,\dots\}, \\
												  &\quad k \in \{0,1,\dots,2^{j-1} \}. \notag
\end{align}
Hence $j$ is an index related to frequency, and $k$ is related to position.
The full wavelet expansion is then
\begin{equation}
	f(t) = \theta_0 \varphi(t) + \lsum_{j=0}^{\infty} \lsum_{k=0}^{2^{j-1}} \theta_{0jk} \psi_{jk}(t).
\end{equation}
There exist many families of wavelet functions, but the simplest among them
is the Haar wavelet family. For the Haar wavelet, the wavelet functions are defined
as
\[
	\varphi(t) = 
	\begin{cases}
	~~1 & \text{if}~~0 \leq t < 1, \\
	~~0 & \text{otherwise},
	\end{cases}
\]
and
\[
	\psi(t) = 
	\begin{cases}
	~~-1 & \text{if}~~0 \leq t < \frac{1}{2}, \\
	~~~~1 & \text{if}~~\frac{1}{2} \leq t < 1, \\
	~~~~0 & \text{otherwise}.
	\end{cases}
\]
\begin{figure}
\centering
	\includegraphics[height=\figurewidth]{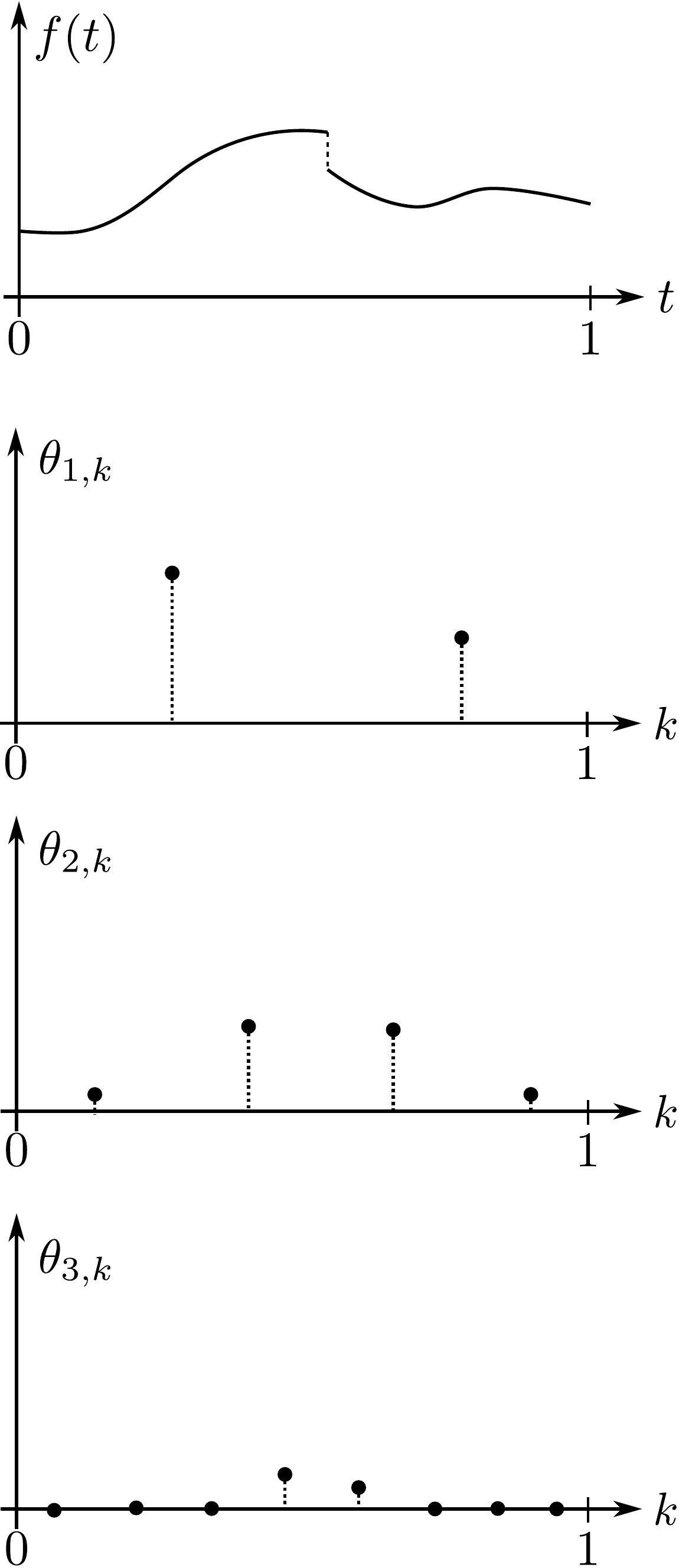}
\caption{Wavelet coefficients of a piecewise-continuous function for increasing 
scale levels.}
\label{fig:wavelet_coeffs}
\end{figure}
In Fig. \ref{fig:wavelet_coeffs} we see an example wavelet expansion of a piecewise-continuous
function. Larger magnitude wavelet coefficients will be located with the discontinuities in the
original function across all scales.

Two problems  arise naturally:
\begin{enumerate}
\item 
Unlike the Fourier basis,  wavelet coefficients have no 
natural ordering of ``importance'', since each wavelet coefficient
describes the function at a certain length scale, and in a certain
position. Hence, the simple idea of fitting all coefficients up to a
certain maximum index $J_*$ cannot be applied. If we select 
all coefficients corresponding to all positions up to a certain
maximum frequency, we will not exploit the spatial adaptivity property
of the wavelet basis.
\item Any linear estimation procedure, that is also translation
  invariant can be represented as a convolution
  cf. \eqref{eq:Convolution},
and thus incurs the problems outlined above. In order to treat
differently edges from smooth regions in an image, a nonlinear
procedure must be used.
\end{enumerate}

The simplest approach that overcomes these problems is the wavelet
denoising method that was developed in a sequence of  seminal papers
by David Donoho and Iain Johnstone
\citeyear{DJ94a,Donoho94,DoJo95,DJ98,JohnstoneBook}.
The basic idea is to truncate, not according to the wavelet index, 
but according according to the magnitude of the measured wavelet
coefficient.
 In the simplest implementation, we proceed in two steps. First we
 perform least squares estimation of each coefficient. In the case of
 orthogonal designs considered here, this yields
\begin{equation}
	\tilde{y} = \frac{1}{n} \bX^{\sT} y = \theta + \tilde{w},
\end{equation}
Here $ \tilde{w}=\bX^{\sT}w/n$ is again white noise $\tilde{w}\sim\normal{0}{(\sigma^2/n)\id_n}$.
After this, coefficients are  \textit{thresholded}, independently,
\[
	\htheta_i = 
	\begin{cases}
	~~\tilde{y_i},& \text{if}~~ | \tilde{y}_i | \geq \lambda, \\
	~~0,& \text{otherwise},.
	\end{cases}
\]
The overall effect of this thresholding is to preserve large magnitude
wavelet coefficients while zeroing those that are `below the noise level.'
Since larger coefficients corresponds to edged in the image, this
approach seek to estimate higher frequencies near edges, only
retaining low frequencies in smooth regions. This allows
 denoising  without blurring across edges.

\section{Denoising with thresholding}
In the last section we briefly described a denoising method, wavelet
thresholding, that is can adapt to a degree of smoothness
that varies across a signal (e.g. an image).  In this section, we
work out some basic properties of this method, under a simple signal model. Apart from being
interesting per se, this analysis provides key insights for
generalizing the same method to high-dimensional statistical
estimation problems beyond denoising. For an in-depth treatment we
refer, for instance, to \cite{JohnstoneBook,DJ94a}.

To recall the our set-up, we are considering the model
\begin{align}
  y = \bX \theta + w,
\end{align}
where $y \in \real^{n}$, $\bX \in \real^{n \times p}$ are observed,
and we want to estimate the vector of coefficients $\theta \in
\reals^p$. The vector $w$ is noise $w\sim\normal{0}{\sigma^2\id_n}$.
We are focusing on orthogonal designs, i.e.  on the case $n=p$ with
$\bX^{\sT} \bX = n , \id_{n \times n}$.

There is no loss of generality in carrying out least squares as a
first step, which in this case reduces to
\begin{align}
  \tilde{y} = \frac{1}{n} \bX^{\sT} y = \theta + \tilde{w}, \quad \tw \sim \mathcal{N}(0,\frac{\sigma^2}{n} \id_{n \times n}).
\end{align}
In other words, in the case of orthogonal designs we can equivalently
assume that the unknown object $\theta$ has been observed directly,
with additive Gaussian noise.

Since we expect $\theta$ to be sparse, it is natural to return a
sparse estimate $\htheta$. In particular, if $\ty_i$ is of the same
order as the noise standard deviation $\sigma$, it is natural to guess
that $\theta_i$ is actually very small or vanishing, and hence set
$\htheta_i=0$.
Two simple ways to implement this idea are `hard thresholding' and
`soft thresholding.'

Under \emph{hard thresholding}, the estimate $\htheta=(\htheta_1,\cdots,\htheta_p)$ of $\theta$ is given by
\begin{align}
  \htheta_i =
  \begin{cases}
    \ty_i & \text{if $|\ty_i| \geq \lambda$}, \\
    0 & \text{else.}
  \end{cases}
\end{align}
Under \emph{soft-thresholding}, the estimate $\htheta$ is given by
\begin{align}
  \htheta_i =
  \begin{cases}
    \ty_i - \lambda & \text{if $\ty_i \geq \lambda$}, \\
    0 & \text{if $|\ty_i| \leq \lambda$}, \\
    \ty_i + \lambda & \text{if $\ty_i \leq - \lambda$} .
  \end{cases}
\end{align}
These hard thresholding and soft thresholding functions are plotted in
Fig. \ref{fig:SoftHard}. While the two approaches  have comparable
properties (in particular, similar risk over sparse vectors), we shall
focus here on soft thresholding since it is most easily generalizable
to other estimation problems.
\begin{figure}[!ht]
  \centering
  \begin{tikzpicture}[xscale=2,yscale=1.5]
    \draw[very thin,color=gray,->] (-2.1,0) -- (2.1,0) node[right,black] {$\tilde{y_i}$}; 
    \draw[very thin,color=gray] (0,-1.1) -- (0,1.1); 

    \draw[color=red,smooth,samples=100,domain=1:2,very thick] plot (\x,\x-1);
    \draw[color=red,smooth,samples=100,domain=-2:-1,very thick] plot (\x,\x+1);
   \draw[color=red,smooth,samples=100,domain=-1:1,very thick] plot(\x,0);

    \draw[color=blue,smooth,samples=100,domain=1:2,very thick] plot (\x,\x);
    \draw[color=blue,smooth,samples=100,domain=-2:-1,very thick] plot (\x,\x);
\draw[color=blue,smooth,samples=100,domain=-1:1,very thick]
plot(\x,0.02);

    \draw (1,-0.1) node[below] {$\lambda$} -- (1,0.1);
    \draw (-1,-0.1) node[left,below] {$\lambda$}  -- (-1,0.1);
  \end{tikzpicture}
  \caption{Soft thresholding (red) and hard thresholding (blue).\label{fig:SoftHard}}
\end{figure}
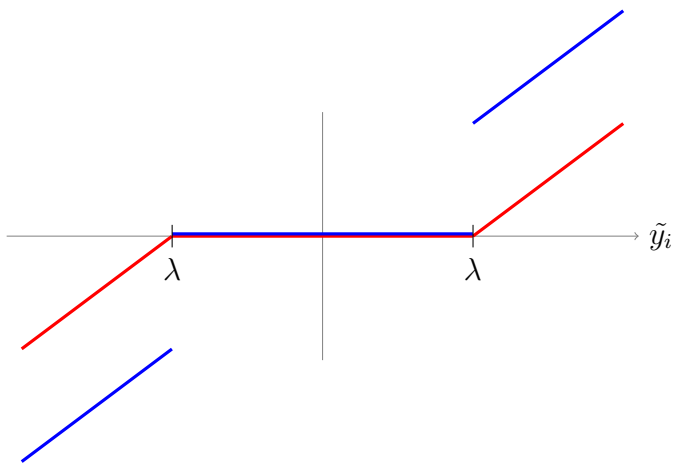

Note that both hard and soft thresholding  depend on a threshold parameter that
we denoted by $\lambda$. Entries below $\lambda$ are set
to zero: to achieve minimal risk, it is of course crucial to select an appropriate $\lambda$.
Ideally, the threshold  should cut-off the coefficients
resulting from the noise,
and hence we expect $\lambda$ to be proportional to the noise standard
deviation $\sigma$.
In order to determine the  optimal choice of $\lambda$, let us
first consider the case $\theta=0$.
Note that, when $\theta=0$, $\ty \sim
\mathcal{N}(0,\frac{\sigma^2}{n}\,\id_{n})$ is a vector with
i.i.d. Gaussian entries.
We claim that, in this case
\begin{align}
  \max_{i \in 1,\cdots,p} |\tilde{y}_i| \approx \sigma \sqrt{\frac{2
      \log p}{n}} \, ,
\end{align}
with probability very close to one\footnote{In this derivation we will be by
choice somewhat imprecise, so as to increase readability. The reader
is welcome to fill in the details, or to consult, for instance, 
\cite{JohnstoneBook,DJ94a}.}

\index{Y}
To see why this is the case, let $N(z)=\E \# \{ i\in [p] : |\ty_i| \ge
|z|\}$ be the expected number of coordinates in the vector $\ty$ that are
above level $|z|$, or below $-|z|$. By linearity of expectation, we
have
\begin{align}
N(z) = 2\, p\, \Phi\Big(-\frac{n|z|}{\sigma}\Big)\, ,
\end{align}
where $\Phi(x) = \int_{-\infty}^xe^{-t^2/2}\de t/\sqrt{2\pi}$ is the
Gaussian distribution function.
Using the inequality $\Phi(-x) \le e^{-x^2/2}/2$, valid for $x\ge 0$,
we obtain
$N(z) \le p\, \exp(-nz^2/2\sigma^2)$. In particular, for any $\delta>0$
\begin{align}
\prob\Big\{\max_{i \in 1,\cdots,p} |\tilde{y}_i| \ge \sigma \sqrt{\frac{2(1+\delta)
      \log p}{n}} \Big\}\le
N\Big(\sigma \sqrt{\frac{2(1+\delta)
      \log p}{n}}\Big)\le p^{-\delta}\, ,
\end{align}
which vanishes as $p\to\infty$. Roughly speaking, this proves that 
$\max_{i \in 1,\cdots,p} |\tilde{y}_i| \lesssim \sigma \sqrt{\frac{2
      \log p}{n}}$ with high probability. A matching lower bound can
  be proved by a second moment argument and we leave it to the reader
  (or refer to the literature). 

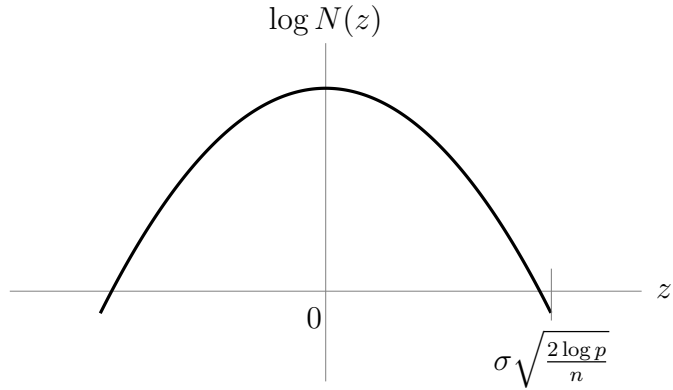
\begin{figure}[!ht]
  \centering
  \begin{tikzpicture}[xscale=3,yscale=3][domain=-1.4:1.4]
    \draw[very thin,color=gray] (-1.4,0.1) -- (1.4,0.1);
    \draw[very thin,color=gray] (0,-0.3) -- (0,1.2); 

    \draw[color=black,smooth,samples=100,domain=-0.999:0.999,very thick] plot (\x,1-\x*\x);

    \draw[color=gray] (1,-0.03) -- (1,0.2);
    \node at (1,-0.2) {$\sigma \sqrt{\frac{2 \log p}{n}}$};
    \node at (1.5,0.1) {$z$};
    \node at (-0.05,-0.02) {$0$};
    \node at (0,1.3) {$\log N(z)$};
  \end{tikzpicture}
  \caption{Sketch $\log N(z)$ (logarithm of the number of coordinates
    with noise level $z$).}\label{fig:logN}
\end{figure}
Figure \ref{fig:logN} reproduces the behavior of $\log N(z)$. The
reader with a background in statistical physics has probably noticed
the similarity between the present analysis and Derrida's treatment of
the `random-energy model' \cite{derrida1981random}. In fact the two
models are identical and there is a close relationship between the
problem addressed within statistical physics and estimation theory.

\begin{figure}[!ht]
  \centering
  \begin{tikzpicture}[xscale=3,yscale=3,
    infonode/.style={circle, inner sep = 0pt, minimum size = 0.8mm, draw=black, fill=black},
    ]
    \draw[->,thin] (-0.1,0) -- (1.5,0);
    \draw[->,thin] (0,-0.7) -- (0,0.7);
    \draw (-0.1,0.5) -- (1.3,0.5) node[right] {$\frac{\sigma}{\sqrt{n}}\sqrt{2 \log p}$};
    \draw (-0.1,-0.5) -- (1.3,-0.5) node[right] {$-\frac{\sigma}{\sqrt{n}}\sqrt{2 \log p}$};
    
    \node[infonode] at (0.1,0.33) {};
    \node[infonode] at (0.2,-0.47) {};
    \node[infonode] at (0.3,-0.29) {};

    \node[gray] at (0.6,0.05) {$\cdots$};
    \node[infonode] at (0.8,0.45) {};
    \node[infonode] at (0.9,-0.06) {};
    \node[infonode] at (1.0,0.26) {};

    \draw[very thin,gray] (0.1,-0.05) node [below] {\tiny{$1$}} -- (0.1,0.05) ;
    \draw[very thin,gray] (0.2,-0.05) node [below] {\tiny{$2$}} -- (0.2,0.05);  
    \draw[very thin,gray] (0.3,-0.05) node [below] {\tiny{$3$}} -- (0.3,0.05);

    \draw[very thin,gray] (1.0,-0.05) node [below] {\tiny{$p$}} -- (1.0,0.05);

  \end{tikzpicture}
  \caption{Cartoon of the vector of observations $\ty$  when $\theta=0$.}\label{fig:PureNoise}
\end{figure}
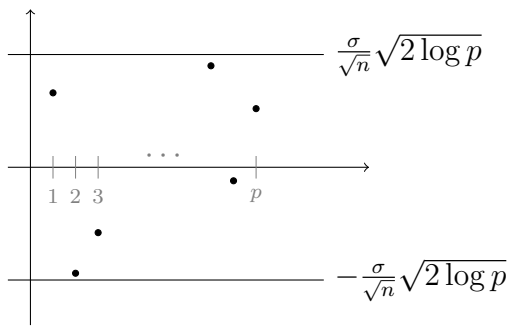

Figure \ref{fig:PureNoise} is a carton of the vector of observations
$\ty$ in the case in which the signal vanishes: $\theta=0$. All the
coordinates of $\ty$ lie between $-\sigma \sqrt{\frac{2 \log p}{n}}$
and $+\sigma \sqrt{\frac{2 \log p}{n}}$. This suggests to set the
threshold $\lambda$ as to zero all the entries that are pure noise.
This leads to the so-called following thresholding rule, proposed in \cite{donoho1994ideal}
\begin{align}
\lambda = \sigma \sqrt{\frac{2 \log p}{n}}\, .
\end{align}
We now turn to evaluating the  risk for such an estimator, when
$\theta \neq 0$ is a sparse signal: 
\begin{align}
  R(\theta ; \htheta) = \E\big\{\|\theta - \htheta \|^2\big\} =
  \sum_{i=1}^p \E\big\{ (\theta_i-\htheta_i)^2 \big\}\, .
\end{align}
We can decompose this risk as
\begin{align*}
  R = R_0 + R_{\neq 0} ,
\end{align*}
where $R_0$ (respectively, $R_{\neq 0}$) is risk from entries
$\theta_i$ that are zero (respectively, non-zero). 
The two contributions depend differently on $\lambda$: the contribution
of zeros decreases with $\lambda$ since for large $\lambda$ more
entries are set to $0$. The contribution of non-zero entries instead
increases with $\lambda$ since large $\lambda$ produces a larger
bias, see Fig.~\ref{fig:Risk0vsNon0} for a cartoon.
\begin{figure}[!ht]
  \centering
  \begin{tikzpicture}[xscale=8,yscale=2][domain=0:1]
    \draw[very thin,color=gray] (-0.1,0) -- (1.1,0) ; 

    \draw[color=red,smooth,samples=100,domain=0:1,very thick] plot (\x,{exp(1-\x) -1)});
    \node at (1,0.3) {$R_0$};
    \draw[color=blue,smooth,samples=100,domain=0:1,very thick] plot (\x,{exp(\x) - 1});
    \node at (1,1.4) {$R_{\neq 0}$};
    \draw[color=gray] (1,-0.2) -- (1,0.2);

    \node at (1,-0.4) {$\sigma \sqrt{\frac{2 \log p}{n}}$};
  \end{tikzpicture}
  \caption{Risk $R_{\neq 0}$ in comparison to $R_0$.}\label{fig:Risk0vsNon0}
\end{figure}
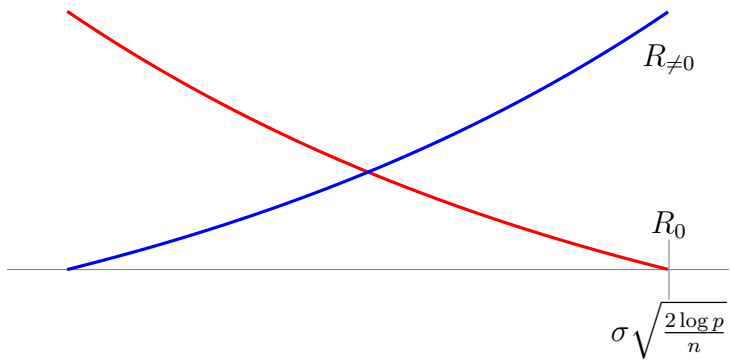

\begin{figure}[!ht]
  \centering
  \begin{tikzpicture}[xscale=8,yscale=2,
    infonode/.style={circle, inner sep = 0pt, minimum size = 0.8mm, draw=black, fill=black},
    domain=0:1,]
    \draw[very thin,color=gray] (-0.1,0) -- (1.5,0);
    \draw[very thin,color=gray] (0,-0.2) -- (0,1);

    \node at (-0.05,1) {$\tilde{y}$};
    \draw[gray] (0.1,0) -- (0.1,0.7);  \node[infonode] at (0.1,0.79) {};
    \draw[gray] (0.2,0) -- (0.2,0.8); \node[infonode] at (0.2,0.88) {};
    \draw[gray] (0.3,0) -- (0.3,0.9); \node[infonode] at (0.3,0.82) {};  
    \draw[very thin,color=gray] (0,0.55) -- (1.5,0.55); \node at (-0.05,0.55) {$\lambda$};
    \draw[gray] (0.4,0) -- (0.4,0.5); \node[infonode] at (0.4,0.53) {}; 
    \draw[gray] (0.5,0) -- (0.5,0.4); \node[infonode] at (0.5,0.32) {};

    \node at (0.9,0.1) {$\cdots$};

    \node[infonode] at (1.2,0) {};
    \node[infonode] at (1.3,0) {};
    \node[infonode] at (1.4,0) {};
  \end{tikzpicture}
  \caption{This picture illustrates the universal thresholding.
    The solid dots represent $\tilde{\theta}$ and the bars represent $\tilde{y}$.
  }
\end{figure}
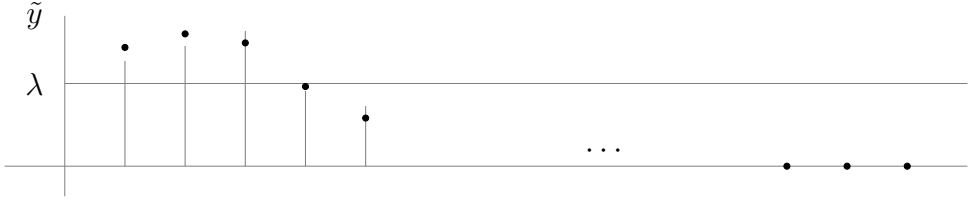

Under universal thresholding, since $\max_{i :\theta_i=0}
|\tilde{y}_i| \lesssim \sigma \sqrt{\frac{2 \log p}{n}}=\lambda$, we have $R_0
\approx 0$.
In order to evaluate the contribution of non-zero entries, we
assume that $\theta$ is $s_0$ sparse, i.e., letting $\supp(\theta)
\equiv \{ i\in [p] : \theta_i \neq 0 \}$, we have $|\supp(\theta)|\leq s_0$.
Note that soft thresholding introduces a bias of size $\lambda$ on
these entries, as soon as they are sufficiently than $\lambda$. This
gives an error per coordinate proportional to $\lambda^2$ (the
variance contribution is negligible on these entries).
This gives
\begin{align}
  R(\theta ; \htheta) \approx R_{\neq 0} \approx s_0 \lambda^2  =
  \frac{s_0\sigma^2}{n} (2 \log p)\,  .\label{eq:SparseRisk}
\end{align}
We can now step back and compare this result with the risk of least
square
estimation (\ref{eq:ls_risk}). Neglecting the factor $(2\log p)$ which
is small even for very high dimension, our formula for sparse vectors
(\ref{eq:SparseRisk}) is the same as for least squares, except that
the dimension $p$ is replaced by the number of non-zero entries $s_0$.
In other words, we basically achieve the same risk \emph{as if} we
knew a priori $\supp(\theta)$ and run least squares on that support!
The extra factor $(2\log p)$ is the price we pay for not knowing where
the support is.
For sparse vectors, we achieve an impressive improvement over least
squares.

Notice that this improvement is achieved simultaneously over all
possible sparsity levels $s_0$, and the estimator does not need to know a
priori $s_0$.

\subsection{An equivalent analysis: Estimating a random scalar}
\label{sec:RandomScalar}

There is a different, and essentially equivalent, way to analyze soft
thresholding denoising. We will quickly sketch this approach because it provides
an alternative point of view and, most importantly, because we will
use some of its results in the next sections. We will omit spelling
out the correspondence with the analysis in the last section.

We state this analysis in terms of a different --but essentially
equivalent-- problem. A source of information produces a random
variable $\Theta$ in $\reals$ with distribution $p_{\Theta}$, and we
observe it corrupted by Gaussian noise. Namely, we observe $Y$ given by
\begin{align}
  Y = \Theta + \tau\, Z\, , 
\end{align}
where $Z\sim\normal{0}{1}$ independent of $\Theta$, and $\tau$ is the
noise standard deviation. We  want to estimate $\Theta$ from $Y$. A block diagram of this proces is shown below.
\begin{figure}[!ht]
  \centering
  \begin{tikzpicture}[
    xscale=2,yscale=3,
    squarenode/.style={rectangle, inner sep = 0pt, minimum size = 10mm, draw=black},
    ]
    \node (n1) at (0,0) {\Large{$\Theta$}};
    \node[squarenode] (n2) at (1,0) {\Large{$+$}};
    \node (n3) at (2,0) {\Large{$Y$}};
    \node[squarenode] (n4) at (3,0) {\Large{$\hTheta$}};
    \node (n5) at (4,0) {\Large{$\hTheta(Y)$}};
    \node (n6) at (1,-0.5) {\Large{$\tau\, Z$}};

    \draw[->,thick] (n1) -- (n2);
    \draw[->,thick] (n2) -- (n3);
    \draw[->,thick] (n3) -- (n4);
    \draw[->,thick] (n4) -- (n5);
    \draw[->,thick] (n6) -- (n2);
  \end{tikzpicture}
\end{figure}
(A hint: the correspondence with the problem in the previous section
is obtained by setting $\tau= \sigma/\sqrt{n}$ 
and $p_{\Theta} = \frac{1}{p} \sum_{i=1}^p \delta_{\theta_i}$.)

We saw in the previous sections that sparse vectors can be used to
model natural signals (e.g. images in wavelet domain). In the present
framework, this can be modeled by the set of probability distributions
that attribute mass at least $1-\eps$ to $0$:
\begin{align}
  \mathcal{F}_{\eps} = \big\{  p_{\Theta} \in \cP\;\big|
  \;p_{\Theta}(\{0\}) \ge 1-\eps\, \big\},
\end{align}
where $\cP$ is the set of all probability distributions over the real
line $\reals$.
Equivalently, $\cF_{\eps}$ is the class of probability distributions
that can be written as $p_{\Theta} = (1-\eps)\delta_0+\eps\, Q$
where $\delta_0$ is the Dirac measure at $0$ and $Q$ is an arbitrary probability distribution.

The Bayes risk of an estimator $\hTheta$ is given by
\begin{align}
  R_{B}(p_{\Theta} ; \hTheta) = \E\big\{ [\hTheta - \Theta ]^2
  \big\}\, .
\end{align}
In view of the interesting properties of soft thresholding, unveiled
in the previous section, we will assume that $\Theta$ is obtained by
soft thresholding $Y$. It is convenient at this point to introduce
some notation for soft thresholding:
\begin{align}
  \eta(z;\lambda) =
  \begin{cases}
    z - \lambda & \text{if $z \geq \lambda$}, \\
    0 & \text{if $|z| \leq \lambda$}, \\
    z + \lambda & \text{if $z \leq - \lambda$} .
  \end{cases}
\end{align}
With an abuse of notation, we write $R_B(p_{\Theta};\lambda) =
R_{B}(p_{\Theta} ; \eta(\,\cdot\,;\lambda))$ for the Bayes risk of 
soft thresholding with threshold $\lambda$.
Explicitly 
\begin{align}
  R_{B}(p_{\Theta} ; \lambda)= \E\big\{ [\eta(Y;\lambda) - \Theta ]^2
  \big\}\, .
\end{align}

We are interested in bounding the risk  $R_{B}(p_{\Theta} ; \lambda)$ for all
sparse signals, i.e., in the present framework, for all the
probability distributions $p_{\Theta}\in\cF_{\eps}$. We then consider
the minimax risk:
\begin{align}
  R_{*}(\eps;\tau^2) = \inf_{\lambda} \sup_{p_{\Theta} \in
    \cF_{\eps}} R_{B}(p_{\Theta};\lambda)\, .
\end{align}
First note that the class $\cF_{\eps}$ is scale invariant. If
$p_{\Theta}\in\cF_{\eps}$, also the probability distribution that is
obtained by `stretching' $p_{\Theta}$ by any positive factor $s$ is in
$\cF_{\eps}$.
Hence the only scale in the problem is the noise variance $\tau^2$. It follows that
\begin{align}
  R_{*}(\eps;\tau^2) = M(\eps)\, \tau^2 .
\end{align}
for some function $M(\eps)$.
Explicit formulae for $M(\eps)$ can be found --for instance-- in
\cite[Supplementary Material]{donoho2009message} or
\cite{donoho2011noise}. A sketch is shown in Fig.~\ref{fig:Meps}: in
particular 
$M(\eps) \approx 2 \eps \log \frac{1}{\eps}$ as $\eps\to 0$.
\begin{figure}[!ht]
  \centering
  \begin{tikzpicture}[xscale=6,yscale=3][domain=0:1]
    \draw[very thin,color=gray,->] (-0.1,0) -- (0.6,0) node[right,black] {$\eps$}; 
    \draw[very thin,color=gray,->] (0,-0.1) -- (0,1.1) node[above,left,black] {$M(\eps)$}; 

    \draw[color=blue,smooth,samples=100,domain=0.00001:0.5,very thick] plot (\x,{-\x*log2(\x)-(1-\x)*log2(1-\x)});

    \node at (-0.1,-0.1) {$0$};
    \draw (0.5,-0.1) node[below] {$1$} -- (0.5,0.1);
    \draw (-0.03,1) node[left,below] {$1$}  -- (0.03,1);

    \node[left] (n1) at (-0.2,0.3) {$ \approx 2 \eps \log \frac{1}{\eps}$ when $\eps \approx 0$};
    \draw[<-] (n1) -- (0,0) ;
  \end{tikzpicture}
  \caption{Sketch of the minimax risk of soft thresholding $M(\eps)$.}
\label{fig:Meps}
\end{figure}
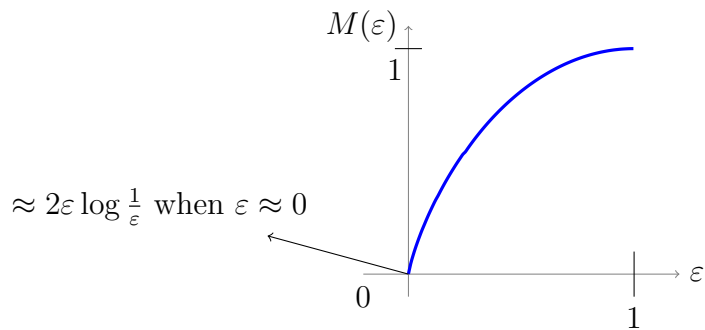
By the same scaling argument as above, the  optimal threshold $\lambda$ takes the form 
\begin{align}
  \lambda^* = \tau \ell(\eps), \quad 
\end{align}
where the function $\ell(\eps)$ can be computed as well and behaves as
$\ell(\eps)\approx \sqrt{2 \log \frac{1}{\eps}}$ for small $\eps$.
Finally, the worst case signal distribution is
\begin{align}
  p_{\Theta}^* = (1-\eps) \delta_0 + \eps \delta_{\infty} .
\end{align}

Note that the small $\eps$ behavior matches --as expected-- the very
sparse limit for vector denoising derived in the previous section.
The correspondence is obtained by substituting $\epsilon=s_0/p$ for
the fraction of non-zero entries and noting that the vector risk is
\begin{align}
  R = p R_*(\eps;\tau) = p M(\eps) \tau^2 .
\end{align}
When $s_0\ll p$, we have $\epsilon \approx 0$ and
\begin{align}
  R = p M(\eps) \tau^2 &\approx p\, 2 \eps\log \frac{1}{\eps} \,\tau^2 \\
  &= p\, 2 \frac{s_0}{p} \log \frac{p}{s_0} \,\frac{\sigma^2}{n} \\
  &= \frac{2s_0\sigma^2}{n} \log \frac{p}{s_0},
\end{align}
which matches the behavior derived earlier.

\label{sec:Denoising}

\section{Sparse regression}
Up to now we have focused on estimating $\theta\in\reals^p$ from observations of the form
\begin{align}
  y = \bX \theta + w,\label{eq:LinearModel}
\end{align}
where $y \in \reals^{n}$, $\bX \in \reals^{n \times p}$ are known 
and $w\in\reals^n$ is an unknown noise vector. We focused in the
previous case on orthogonal designs $n\ge p$ and $\bX^{\sT}\bX = n \, \id_{p \times p}$.

Over the last decade, there has been a lot of interest in the 
underdetermined case $n \ll p$, as general a $\bX$ as possible, which naturally emerges in many
applications.
It turns out that goos estimation is possible provided $\theta$ is
highly structured, and in particular when it is very sparse.
Throughout, we let  $\|\theta\|_0$ denote the `$\ell_0$ norm' of\index{Norm $\ell_0$}
$\theta$, i.e.  the number of non-zero entries in $\theta$ (note that
this is not really a norm).
The main outcome of the work in this area is that the number of
measurements
needs to scale with with the number of non-zeros $\|\theta\|_0$
instead of the ambient dimension $p\gg\|\theta\|_0$. 
This setup is the so-called \emph{sparse regression}, or
\emph{high-dimensional regression} problem.

\subsection{Motivation}

It is useful to overview a few scenarios where the above framework
applies, and in particular the high-dimensional regime $n\ll p$ plays
a crucial role.
\begin{description}
\item[Signal processing] 
  An image can be modeled, for instance, by a function
  $f:[0,1]\times[0,1]\to\reals$ if it is gray-scale. Color images
  requires three scalars at each point, three dimensional imaging
  requires to use a domain $[0,1]\times [0,1]\times [0,1]$, and so on.
Many imaging devices can be modeled (to a first order) as linear
operators $A$ collecting a vector $y\reals^n$ corrupted by noise  $w$ denotes the noise.
  \begin{align}
    y = A f + w,
  \end{align}
 To keep a useful example in mind,  $A$ can be the partial Fourier matrix,
i.e. the operator computing a subset of Fourier coefficients.
  As emphasized in the previous sections, the image $f$ is often sparse in some domain, say wavelet transform.
  That is $f = T \theta$, where $\theta$ is sparse and $T$ is the
  wavelet transform, or whatever sparsifying transform.
  This gives rise to the model
  \begin{align}
    y = A f + w =  (A T) \theta + w = \bX \theta + w, 
  \end{align}
where $\bX=AT$. Here $n$ corresponds to the number of measurements,
while $p$ scales with the number of wavelet coefficient, and hence
with the resolution that we want to achieve. The high-dimensional
regime $n\ll p$ is therefore very useful as it corresponds to simpler
measurements and higher resolution. 
\item[Machine learning] In web services, we often want to predict
  an unknown property of a user, on the basis of a large amount on
  known data about her.
  For instance, an online social network as Facebook, might want to
  estimate the income of its users, in order to display targeted advertisement.
  For each user $i$, we can construct a feature vector $x_i \in \reals^p$, where e.g.,
  \begin{align}
    \text{$x_i$ = (age, location, number of friends, number of posts,} \nonumber\\
      \text{time of first post in a day, $\cdots$)}.\nonumber
  \end{align}
  In a linear model, we assume
  \begin{align}
    \underbrace{y_i}_{\text{income}} = \<x_i,\theta\> + w_i ,\label{eq:LinearRelation}
  \end{align}
  Combining all users, we have
  \begin{align}
    y= \bX\theta + w ,
  \end{align}
where $y=(y_1,y_2,\dots,y_n)$ is a vector comprising all response
variables (e.g. the customers' income), and $\bX$ is a matrix whose
$i$-th row is the feature vector $x_i$ of the $i$-th customer.
  Typically, one constructs feature vectors with tens of thousands of
  attributes, hence giving rise to $p= 10^4$ to $10^5$. On the other
  hand, in order to fit such a model, the response variable (income) $y_i$ needs to be known for a
  set of users and this is often possible only for $n\ll p$ users.

  Luckily, only a small subset of features is actually relevant to
  predict income, and hence we are led to use sparse estimation techniques.
\end{description}

\subsection{The LASSO}
\label{sec:LASSO}

The LASSO (Least Absolute Shrinkage and Selection Operator) presented in
\cite{Tibs96}, also known as Basis Pursuit DeNoising (BPDN) \cite{BP95,chen1998atomic} is arguably the
most successful method for sparse regression.
The LASSO estimator is defined in terms of an optimization problem
\begin{align}\tag*{(LASSO)}
  \htheta = \argmin_{\theta \in \mathbb{R}^p } \Big\{
  \underbrace{\frac{1}{2n} \| y - \bX \theta \|_2^2}_{\text{Residual
      sum of squares}} + \underbrace{\lambda \| \theta
    \|_1}_{\text{Regularizer}} \Big\} \label{eq:LASSO}
\end{align}
The term $\mathcal{L}(\theta)=\frac{1}{2n}\|y-\bX\theta\|_2^2$ is the
ordinary least squares cost function, and the regularizer $\lambda
\|\theta\|_1$ promotes sparse vectors by
penalizing coefficients different from $0$.
Note that the  optimization problem  is convex and hence it can be
solved efficiently: we wil discuss a simple algorithm in the following.

\index{Y}
To gain insight as to why the LASSO is well-suited for sparse
regression, let us start by revisiting  the case of  orthogonal
designs, namely $n\ge p$ and
\begin{align}
 \bX^{\sT}\bX = n \, \id_{n \times n} .
\end{align}
Rewriting $\mathcal{L}(\theta)$:
\begin{align}
  \mathcal{L}(\theta) &= \frac{1}{2n} \left< y-\bX\theta, y-\bX\theta \right>,\nonumber \\
  &= \frac{1}{2n} \left< y-\bX\theta,\frac{1}{n} \bX\bX^{\sT} (y-\bX\theta) \right>, \nonumber\\
  &= \frac{1}{2n^2} \left< \bX^{\sT}y - n \theta , \bX^{\sT} y - n \theta  \right>, \nonumber\\
  &= \frac{1}{2} \left\| \theta - \frac{1}{n}\bX^{\sT}y \right\|^2, \nonumber\\
  &= \frac{1}{2} \left\| \theta - \tilde{y}   \right\|^2, \qquad
  \text{where $\tilde{y}=\frac{1}{n}\bX^{\sT}y$}\, .
\end{align}
Thus, in this case, the LASSO problem is equivalent to
\begin{align}
  {\rm minimize}\;\;\;\sum_{i=1}^{p} \left\{ \frac{1}{2}
    \big|\tilde{y}_i-\theta_i\big|^2 + \lambda |\theta_i| \right\} .
\end{align}
This is a `separable' cost function, and we can minimize each coordinate separately.
Let $F(\theta_i)=\frac{1}{2} (\tilde{y}-\theta_i)^2 + \lambda |\theta_i|$.
Now,
\begin{align}
  \frac{\partial F}{\partial \theta_i} &=\theta_i - \tilde{y}_i + \lambda\, \text{sign}(\theta_i),
\end{align}
where $\text{sign}(\cdot)$ denotes the sign function shown in Figure \ref{figure:sign_function}.
\begin{figure}[!ht]
  \centering
  \begin{tikzpicture}[xscale=3,
    opennode/.style={circle, inner sep = 0pt, minimum size = 1mm, draw=black, fill=white},
    fillnode/.style={circle, inner sep = 0pt, minimum size = 1mm, draw=black, fill=blue},
    ]
    \draw[very thin,color=gray,->] (-1.1,0) -- (1.1,0) node[right,black] {$\theta_i$}; 
    \draw[very thin,color=gray] (0,-1.1) -- (0,1.1) node[right,above,black] {\small{$\textrm{sign}(\theta_i)$}}; 

   \draw[thick,blue] (0,-1) -- (0,1) node[fillnode] {};
    \draw[thick,blue] (-1,-1) -- (0,-1) node[fillnode] {};
    \draw[thick,blue] (0,1) node[fillnode] {} -- (1,1);

    \node[fillnode] at (0,0) {};

    \node[left] at (0,1) {$1$};
    \node[right] at (0,-1) {$-1$};
  \end{tikzpicture}
  \caption{The sign function.}
  \label{figure:sign_function}
\end{figure}
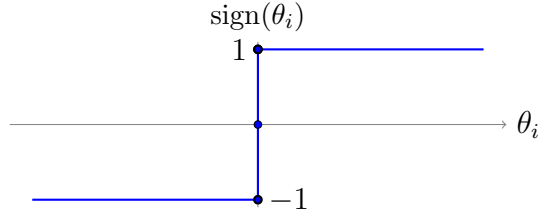
Note that $|\theta_i|$ is non-differentiable at $\theta_i=0$. How
should we interpret its derivative $\sign(\theta_i)$ in this case? For
convex functions (which is the case here) the derivative can be safely
replaced by the `subdifferential,' i.e. the set of all
possible slopes of  tangent lines at $\theta_i$ that stay below the graph of the
function to be differentiated. The subdifferential coincides with the
usual derivative where the function is differentiable. For the function $\theta_i\to
|\theta_i|$, it is an easy exercise to check that the subdifferential
at $\theta_i = 0$ is given by the interval $[-1,1]$. In other words,
we can think of Figure \ref{figure:sign_function} as the correct graph
of the subdifferential of $|\theta_i|$ if we interpret its value at
$0$ as given by the whole interval $[-1,1]$.

The minimizer of $F(\theta_i)$ must satisfy
\begin{align}
\tilde{y}_i =\theta_i+ \lambda\, \text{sign}(\theta_i) \, .
\end{align}
Hence we can obtain the minimizer  as a function of $\ty_i$ by adding
$\theta_i$ to the graph i figure (\ref{figure:sign_function}) and
flipping the axis. The result is plotted in the next figure.
\begin{figure}[!ht]
  \centering
  \begin{tikzpicture}[xscale=2,yscale=2]
    \draw[very thin,color=gray,->] (-2.1,0) -- (2.1,0) node[right,black] {$\tilde{y_i}$}; 
    \draw[very thin,color=gray] (0,-1.1) -- (0,1.1) node[right,black] {$\eta(\tilde{y}_i;\lambda)$}; 

    \draw[color=blue,smooth,samples=100,domain=1:2,very thick] plot (\x,\x-1);
    \draw[color=blue,smooth,samples=100,domain=-2:-1,very thick] plot
    (\x,\x+1);
\draw[color=blue,smooth,samples=100,domain=-1:1,very thick] plot (\x,0);

    \draw (1,-0.1) node[below] {$\lambda$} -- (1,0.1);
    \draw (-1,-0.1) node[left,below] {$\lambda$}  -- (-1,0.1);
  \end{tikzpicture}
\end{figure}
The reader will recognize the soft thresholding function
$\eta(\,\cdot\,;\lambda)$ already encountered in the previous section.
Summarizing, in the case of orthogonal designs, the LASSO estimator
admits the explicit representation
\begin{align}
\htheta = \eta\Big(\frac{1}{n}\bX^{\sT}y;\lambda\Big)\, ,
\end{align}
where it is implicitly understood that the soft thresholding function
is applied component-wise to the vector $(1/n)\bX^{\sT}y$. As we saw in
the previous section, component-wise soft-thresholding has nearly
optimal performances on this problem, and hence the same holds for the LASSO.

In the high-dimensional setting $p \gg n$ and $\bX$ is obviously not
orthogonal,
and the LASSO estimator is non-explicit. Nevertheless it can be
computed efficiently, and we will discuss next a  simple algorithm
that is guaranteed to converge. It is an example of a generic method
for convex optimization known as a `subgradient' or `projected
gradient' approach. The important advantage of these algorithms (and
more generally of  `first order methods')  is that their complexity
per iteration scales only linearly in the dimensions $p$ of the
problem, and are hence well suited for high-dimensional applications \cite{juditsky2011first}.
They are not as fast to converge as --for instance-- Newton's method,
but this is often not crucial. For statistical problems a `low
precision' solution is often as good as a `high precision' one since
in any case there is an unavoidable statistical error to deal with.

We want to minimize the cost function
\begin{align}
 F(\theta) = \frac{1}{2n} \| y - \bX \theta \|^2 + \lambda \| \theta
 \|_1\, .
\end{align}
At each iteration, the algorithm constructs an approximation $\theta^{(t)}$
of the minimizer $\htheta$.
In order to update this state, the idea is to  construct an upper bound to $F(\theta)$ that is easy
to minimize and is a good approximation of $F(\theta)$ close to $\theta^{(t)}$.
Rewriting $\mathcal{L}(\theta)$:
\begin{align}
  \mathcal{L}(\theta)&=\frac{1}{2n} \|y-\bX\theta\|^2_2,\nonumber \\
  &= \frac{1}{2n} \|y-\bX\theta^{(t)} - \bX(\theta-\theta^{(t)})\|^2_2,\nonumber \\
  &= \frac{1}{2n} \|y-\bX\theta^{(t)}\|_2^2 - \frac{1}{n} \left\<
    \bX(\theta-\theta^{(t)}) , y-\bX \theta^{(t)} \right\> +
  \frac{1}{2n} \|\bX(\theta-\theta^{(t)})\|_2^2 \, .
\end{align}
Note that the first two terms are `simple' in that they are linear in
$\theta$. The last term is `small' for $\theta$ close to
$\theta^{(t)})$ (quadratic in $(\theta-\theta^{(t)})$). We will upper
bound the last term.
Suppose the largest eigenvalue of $\frac{1}{n} \bX^\sT\bX$ is bounded by $L$,
\begin{align}
  \lambda_{\max}\left( \frac{1}{n} \bX^{\sT}\bX \right) \leq L,
\end{align}
and let $v=\frac{1}{n} \bX^{\sT}(y-\bX\theta^{(0)})$.
Then
\begin{align}
  \mathcal{L}(\theta)
  &= \frac{1}{2n} \|y-\bX\theta^{(t)}\|_2^2 -
  \left\<v,\theta-\theta^{(t)}\right\> + 
\frac{1}{2} \left\< \theta-\theta^{(t)},\frac{1}{n}\bX^\sT\bX(\theta-\theta^{(t)})\right\>,\nonumber \\
  &\leq \frac{1}{2n} \|y-\bX\theta^{(t)}\|_2^2 - \left\< v , \theta - \theta^{(t)} \right\> + \frac{L}{2} \|\theta - \theta^{(t)}\|_2^2,\nonumber \\
  &= \underbrace{\frac{1}{2n} \|y-\bX\theta^{(t)}\|_2^2 - \frac{1}{2L}
    \|v\|_2^2}_{\triangleq C} + \frac{L}{2} \|\theta - \theta^{(t)} -
  \frac{1}{L}v\|_2^2 \, .
\end{align}
We therefore obtain the following upper bound, whereby $C$ is a
constant independent of $\theta$, 
\begin{align}
  F(\theta) \leq C + \lambda \|\theta\|_1 + \frac{L}{2} \Big\|\theta -
  \theta^{(0)} - \frac{v}{L}\Big\|_2^2 .
\end{align}
We compute the next iterate $\theta^{(t+1)}$ by minimizing the above upper bound:
\begin{align}
  \text{minimize} \quad \quad&\frac{\lambda}{L} \|\theta\|_1 \!+\!
  \frac{1}{2} \|\theta-\tilde{\theta}^{(0)}\|^2, \,\, \nonumber\\
&\tilde{\theta}^{(0)}
=\!\theta^{(0)}+\frac{1}{nL}\bX^{\sT}(y\!-\!\bX\theta^{(0)}) \, .
\end{align}
We already solved this problem when discussing  the case of orthogonal designs.
The solution is given by the soft thresholding operator:
\begin{align}
  \theta^{(t+1)} = \eta \left(  \theta^{(t)} + \frac{1}{nL}
    \bX^{\sT}(y-\bX\theta^{(t)}) ; \frac{\lambda}{L}  \right) .
\end{align}
This yields an iterative procedure known as Iterative Soft
Thresholding that can be initialized
arbitrarily, e.g. with $\theta^{(0)} = 0$. 
This algorithm is guaranteed to always converge, as shown in \cite{daubechies2004iterative,beck2009fast},
in the sense that
\begin{align}
  F(\theta^{(t)}) - F(\htheta) \leq
  \frac{\text{constant}}{t} .
\end{align}
Note that this is much slower that the rate achieved by Newton's
method. However it can be proved that no first order method (i.e. no
method using only gradient information) can achieve global convergence
rate faster than $1/t^2$ for any problem in the class of the LASSO. We
refer to \cite{juditsky2011first} for a recent  introduction to first order methods.
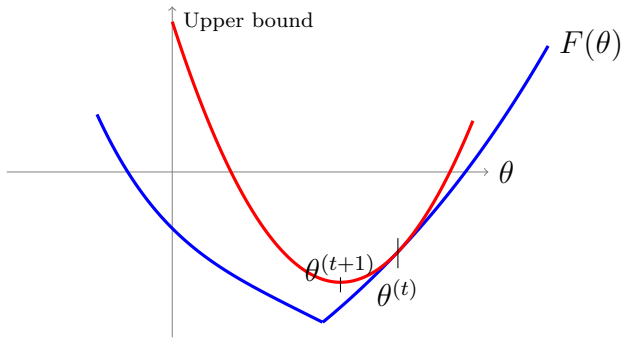
\begin{figure}[!ht]
  \centering
  \begin{tikzpicture}[xscale=2,yscale=2]
    \draw[very thin,color=gray,->] (-1.1,0) -- (2.1,0) node[right,black] {$\theta$}; 
    \draw[very thin,color=gray,->] (0,-1.1) -- (0,1.1); 

    \draw[color=blue,smooth,samples=100,domain=-0.5:1,very thick] plot (\x,{((\x-1)^4)/8+(1-\x)/2-1});
    \draw[color=blue,smooth,samples=100,domain=1:2.5,very thick] plot (\x,{exp(1)^(\x/2)-exp(1/2)-1}); 

    \draw[color=red,smooth,samples=100,domain=0:2,very thick] plot (\x,{1.3864*\x*\x - 3.1008*\x + 1}); 

    \node[right] at (2.5,0.8416) {$F(\theta)$};
    \node[right] at (0,1) {\tiny{Upper bound}};

    \draw (1.5,-0.64) node[below] {$\theta^{(t)}$} -- (1.5,-0.44);
    \draw (1.1183,-0.8) node[above] {$\theta^{(t+1)}$} -- (1.1183,-0.7);
  \end{tikzpicture}
  \caption{Subgradient approach: construct a convenient upper bound to $F(\theta)$.}\label{fig:Subgradient}
\end{figure}

To conclude this section, it is instructive to quickly consider two
special cases: $\lambda \to \infty$ and $\lambda \to 0$.
We rewrite the minimization problem as
\begin{align}
  \htheta &= \argmin_{\theta \in \mathbb{R}^p} F(\theta) \, ,\nonumber\\
  \frac{1}{\lambda} F(\theta) &= \frac{1}{2n\lambda}
  \|y-\bX\theta\|_2^2 + \|\theta\|_1 .
\end{align}
When $\lambda \to \infty$, the first term vanishes and $\htheta \to \argmin_{\theta} \|\theta\|_1 =0$.
In fact, $\htheta=0$ for all $\lambda \geq \lambda_*$ for some critical $\lambda_*$.
When $\lambda \to 0$, the weight in front of the first term goes to
infinity, and hence the equality $y=\bX\theta$ is enforced strictly.
In the high-dimensional regime $p>n$, this linear system is
underdetermined and has multiple solution. The most relevant is
selected by minimizing the $\ell_1$ norm.\index{Norm $\ell_1$}
In other words, as $\lambda\to 0$, the LASSO estimator $\htheta$
converges
to the solution of the following problem (known as `basis pursuit')
\begin{align}
  \text{minimize} &\qquad \|\theta\|_1 \, ,\nonumber\\
  \text{subject to} &\quad y=\bX\theta .
\end{align}

\subsection{Behavior of the LASSO under Restricted Isometry Property}

A significant amount of theory has been developed to understand and
generalize the remarkable properties of the LASSO estimator and its empirical
success. The theory establishes certain optimality properties under
suitable assumptions on the design matrix $\bX$. The most popular of
these assumptions goes under the name of restricted isometry
property (RIP) and was introduced in the groundbreaking work of
Cand\'es, Tao and collaborators
\citeyear{CandesTao,candes2007dantzig}. Several refinements of this
condition were developed in recent years (the restricted eigenvalue
condition of \cite{BickelEtAl}, the compatibility condition of \cite{buhlmann2011statistics} and so on).

In order to motivate the RIP, we notice that 
the LASSO estimator performs well when the columns of the matrix $\bX$
are orthogonal, $\bX^\sT\bX=n\, \id_{n \times n}$. Indeed this is the
case of orthogonal designs explored above.
The orthogonality condition $\bX^\sT\bX=n\, \id_{n \times n}$ is equivalent to
\begin{align}
  \|\bX v\|^2 = n \|v\|^2 \quad \text{for all $v \in \mathbb{R}^p$}
\end{align}
This is of course impossible in the high-dimensional regime $p > n$ 
(indeed the null space of $\bX$ has dimension at least $p-n$).
The idea is to relax this condition, by requiring  that $\bX$ is
``almost orthogonal'' instead of orthogonal, and only when it acts on  sparse vectors.
Explicitly, we say that $\bX$ satisfies the condition RIP$(k,\delta)$
for some integer $k$ and $\delta\in (0,1)$ if 
\begin{align}
  (1\!-\!\delta) \|v\|^2 \leq \frac{1}{n}\|\bX v\|^2 \leq (1\!+\!\delta)\|v\|^2 \quad \text{for all $v \in \mathbb{R}^p$ with $\|v\|_0 \!\leq\! k$} \quad \text{(RIP property)}
\end{align}

It is possible to show that this definition is non-empty and indeed
--in a certain sense--  most matrices satisfy it.
For instance if $\bX$ has iid entries $\bX_{ij} \sim
\text{Unif}\{+1,-1\}$ or $\bX_{ij} \sim \normal{0}{1}$, then with high
probability, $\bX$ satisfies RIP$(k,\delta)$ for a fixed $\delta$ and
$n \geq C k \log \frac{p}{k}$.
The RIP property has been established for a large number of matrix
constructions. For instance partial Fourier matrices\footnote{That is,
  matrices obtained by subsampling randomly the rows of the $p\times
  p$ discrete Fourier transform.} satisfy RIP  with high probability for
$n \geq C_1 \, k (\log p)^{4}$, as shown in \cite{rudelson2008sparse}.

The following theorem illustrates the utility of RIP matrices for
sparse estimation. It is a simplified version of stronger results
established in \cite{BickelEtAl} (without any attempt at reproducing
optimal constants, or the explicit dependence of all the quantities). Results of the same nature were proved earlier in
\cite{candes2007dantzig} for a closely related estimator, known as the
`Dantzig selector.'
\begin{theorem}[Candes, Tao 2006 and Bickel, Ritov, Tsybakov 2009]
  If $\theta$ is $s_0$-sparse and $\bX$ satisfies RIP$(10\, s_0,0.1)$,
  then, by choosing $\lambda=\sigma \sqrt{\frac{5 \log p}{n}}$, we
  have, with high probability for a suitable constant $C>0$,
\begin{align}
    \|\htheta-\theta\|_2^2 \leq \frac{C\, s_0 \sigma^2}{n} \log
    p\, .
\end{align}
\end{theorem}
A few observations are in order.
It is --once again-- instructive to compare this bound with the risk
of ordinary least squares, cf. Eq.~(\ref{eq:ls_risk}). 
Apart from the $\log p$ factor, the error scales as if $\theta_0$ was
$s_0$-dimensional. As in the case of orthogonal designs discussed
above, we obtain roughly the same scaling as if the support of
$\theta$ was known.
Also the  choice of $\lambda$ scales as in the case where $\bX$ is
orthogonal.

Finally, as  $\sigma \rightarrow 0$, we have $\htheta\to \theta$
provided the RIP condition is satisfied. As mentioned above, this
happens for random design matrices if
$n \geq C s_0 \log p$. In other words, we can reconstruct exactly an $s_0$-sparse vector
from about $s_0\log p$ random linear observations.
\subsubsection{Modeling the design matrix $\bX$}

The restricted isometry property and its refinements/generalizations allow to build a
develop a powerful theory of high-dimensional statistical estimation
(both in the context of linear regression and beyond).
This approach has a number of strengths:
\begin{enumerate}
\item[$(a)$] Given a matrix $\bX$, we can characterize it in terms of its
  RIP constant, and hence obtain a bound on the resulting estimation
  error. The bound holds uniformly over all signals $\theta$.
\item[$(b)$] The resulting bound is often nearly optimal.
\item[$(c)$] Many  class of random matrices of interest have been proved to
  possess RIP.
\item[$(d)$] RIP allows to decouple the analysis of the statistical error,
  e.g. the risk of the LASSO estimator $\htheta$, (which is the main
  object interest of statisticians) from the development
  of algorithms to compute $\htheta$ (which is the focus within the
  optimization community). 
\end{enumerate}

The RIP theory has also some weaknesses. It is useful to understand
them since this exercise leads to several interesting research directions that
are --to a large extent-- still open: 
\begin{enumerate}
\item[$(a)$]  In practice, given a matrix $\bX$ it is  NP-hard to whether it
  has  RIP. Hence, one has often to rely on the intuition provided by
  random matrix constructions.
\item[$(b)$] The resulting bounds typically optimal  \emph{within a
    constant}, that can be quite large.
This makes it difficult to compare different estimators for the same
problem. If  estimator  $\htheta^{(1)}$ has risk that is --say-- twice
as large as the one of $\htheta^{(2)}$, this is often not captured by
this theory.
\item[$(c)$] As a special case of the last point, RIP theory provides little
  guidance for the practically important problem of selecting the
  right amount of regularization $\lambda$. It is observed in practice
  that changing $\lambda$ by a  modest amount has important effects on the
  quality of estimation, but this is hardly captured by RIP theory.
\item[$(d)$] Since RIP theory aims at bounding the risk uniformly over all
  (sparse) vectors $\theta$, it is typically driven by the `worst
  case'  vectors, and is overly conservative for most $\theta$'s. 
\end{enumerate}
Complementary information on the LASSO, and other high-dimensional
estimation methods, can be gathered by studying simple random models
fr the design matrix $\bX$. This will be the object of the next lecture.

\label{lecture_4}

\section{Random designs and Approximate Message Passing}
In this lecture we will revisit the linear model
(\ref{eq:LinearModel}) and the LASSO estimator, while assuming a
very simple probabilistic model for the design matrix $\bX$. Before
proceeding, we should therefore ask: Is there any application for
which probabilistic design matrices are well suited?
Two type of examples come to mind
\begin{itemize}
\item In statistics and machine learning, we are given pairs
  $($response variable, covariate vector$)$, 
  $(y_1,x_1)$, \dots, $(y_n,x_n)$ and postulate a relationship as for
  instance in Eq.~(\ref{eq:LinearRelation}).
  These pairs can often be thought as samples from a larger
  `population,' e.g. customers of a e-commerce site are samples of
  a population of potential customers. 

One way to model this, is to assume that the covariate vectors
$x_i$'s, i.e. the rows of $\bX$ are i.i.d. samples from a
distribution. 
\item In compressed sensing, the matrix $\bX$ models a sensing or
\index{Compressed sensing}
  sampling device, that is designed within some physical constraints. 
Probabilistic constructions have been proposed and implemented by
several authors, see e.g. \cite{tropp2010beyond} for an example. 
A cartoon example of these constructions is obtained by sampling
i.i.d. random rows from the discrete $p\times p$ Fourier transform.
\end{itemize}

In other words, random design matrices $\bX$ with i.i.d. rows can be used to
model several applications.
Most of the work has however focused on the special case in which the
rows are i.i.d. with distribution $\normal{0}{\id_{p\times
    p}}$. Equivalently, the matrix $\bX$ has i.i.d. entries
$\bX_{i,j}\sim\normal{0}{1}$.
Despite its simplicity, this model has been an important playground for
the development of many ideas in compressed sensing, starting with
the pioneering work of  Donoho \citeyear{donoho2006high}, and Donoho
and Tanner \citeyear{donoho2005neighborliness,donoho2005sparse}.
Recent years have witnessed an explosion of contributions also thanks
to the convergence of powerful ideas from high-dimensional convex
geometry and Gaussian processes, see e.g.
 \cite{chandrasekaran2012convex,candes2013simple,stojnic2013framework,oymak2013squared,amelunxen2013living}.
Non-rigorous ideas from statistical physics were also used in \cite{kabashima2009typical,rangan2009asymptotic,guo2009single,krzakala2012statistical}.

\index{Compressed sensing}

Here we follow a rigorous approach that builds upon ideas from statistical
physics, information theory  and graphical models, and is based on the analysis of an
\index{Graphical model}
highly efficient reconstruction  algorithm. We will sketch the main
ideas referring to
\cite{donoho2009message} for the original idea,
to \cite{donoho2011noise,bayati2011dynamics,bayati2012lasso} for the
analysis of the LASSO, and to
\cite{donoho2013accurate,javanmard2013state,donoho2013information} for extensions.
This approach was also used in \cite{bayati2014universality} to
establish universality of the compressed sensing phase transition for
non-Gaussian i.i.d. entries $\bX_{i,j}$.\index{Phase Transition}

\subsection{Message Passing algorithms}

The plan of our analysis is as follows:
\begin{enumerate}
\index{Approximate Message Passing}
\item We define an approximate message passing (AMP)  algorithm to
  solve the LASSO optimization problem. 
The derivation presented here
  starts from the subgradient method described in Section
  \ref{sec:LASSO} and obtain a slight --but crucial-- modification of
  the same algorithm.  Also in this case the algorithm is iterative
  and computes a sequence of iterates $\{\theta^{(t)}\}$.

An alternative approach (susceptible of
  generalizations --for instance-- to Bayesian estimation) is
  presented in \cite{donoho2010message}.
\item Derive an exact asymptotic characterization of the same
  algorithm as  $n,p \to\infty$, \emph{for $t$ fixed}. The characterization
  is given in terms of the so-called state evolution method developed 
  rigorously in \cite{bayati2011dynamics} (with generalizations in \cite{javanmard2013state,bayati2014universality}).
\item Prove that AMP converges fast to the optimized $\htheta$, namely
  with high probability as $n,p\to\infty$ we have
  $\|\theta^{(t)}-\htheta\|_2^2/p\le c_1,\, e^{-c_2 t}$, with
  $c_1,c_2$ two dimension-independent constants. A full proof of this
  step can be found in \cite{bayati2012lasso}.
\item Select $t$ a large enough constant and use the last two result
  to deduce properties of the optimizer $\htheta$.
\end{enumerate}

We next provide a sketch of the above steps. We start by considering iterative soft thresholding with $L  = 1$:
\begin{align}
\begin{cases}
	&\theta^{(t + 1)} = \eta (\theta^{(t)} + \frac{1}{n} \bX^\sT
        r^{(t)}; \gamma_t)\, ,\\
	&r^{(t)} = y - \bX \theta^{(t)}\, ,
\end{cases}\label{eq:IST_Repeat}
\end{align}
where we introduced the additional freedom of an iteration-dependent
threshold $\gamma_t$ (instead of $\lambda$). 
Component-wise, the iteration takes  the form
\begin{align}
\begin{cases}
	&\theta^{(t + 1)}_i = \eta (\theta^{(t)}_i + \frac{1}{n}
        \sum_{a = 1}^n \bX_{ai} r^{(t)}_a; \gamma_t)\, ,\\
	&r^{(t)}_a = y_a - \sum_{i = 1}^p \bX_{ai} \theta^{(t)}_i\, .
\end{cases}\label{eq:IST_coordinate}
\end{align}
We next derive a message passing version of this iteration\footnote{We
  use the expression `message passing' in the same sense attributed in
  information theory and graphical models.} (we refer for instance to
\cite{richardson2008modern,mezard2009information}
for background). The
motivation for this modification is that message passing algorithms
have appealing statistical properties. For instance, they admit an
exact asymptotic analysis on locally tree-like graphs. While
--in the present case-- the underlying graph structure is not locally
tree-like, the conclusion (exact asymptotic characterization)
continues to hold.

In order to define the message-passing version, we need to associate a
factor graph to the LASSO cost function:
\begin{align}
F(\theta) = \frac{1}{2n} \sum_{a  = 1}^n \big(y_a - \<x_a, \theta\>
\big)^2 + \lambda \sum_{i = 1}^p |\theta_i|\, .
\end{align}
Following a general prescription from \cite{mezard2009information}, we
associate a factor node to each term $(y_a - \<x_a, \theta\>)^2/(2n)$ in the cost function
indexed by $a\in \{1,2,\dots,n\}$ (we do not need to represent the singletons $|\theta_i|$ by factor
nodes), and we associate a variable node to each variable, indexed by
$i\in \{1,2,\dots, p\}$. We connect factor node $a$ and variable node
$i$ by an edge $(a,i)$   if and only if term $a$ depends on variable
$\theta_i$, i.e. if $\bX_{ai} \neq 0$. Note for Gaussian design
matrices, all the entries $\bX_{ai}$  are non-zero with probability
one. Hence, the resulting factor graph is a complete bipartite graph
with $n$ factor nodes and $p$ variable nodes. 
 
The message-passing version of the iteration (\ref{eq:IST_coordinate})
has iteration variables (messages) associated to directed edges of the factor graph. Namely,
for each edge $(a,i)$ we introduce a message $r^{(t)}_{a\to i}$ and a
message $\theta^{(t)}_{i\to a}$. We replace the update rule
(\ref{eq:IST_coordinate}) by the following 
\begin{align}
\begin{cases}
	&\theta^{(t + 1)}_{i\to a} = \eta \Big(\frac{1}{n} \sum_{b \in
          [n]\setminus a} \bX_{bi} r^{(t)}_{b\to i}; \gamma_t\Big)\, ,\\
	&r^{(t)}_{a\to i} = y_a - \sum_{j \in [p]\setminus i}
        \bX_{aj} \theta^{(t)}_{j\to a}\, .
\end{cases}\label{eq:MP}
\end{align}
The key property of this iteration is that an outgoing message from
node $\alpha$ is updated by evaluating a function
of all messages incoming in the same node $\alpha$, except the one
along the same edge.
An alternative derivation of this iteration follows by considering the
standard belief propagation algorithm  (in its sum-product or min-sum
forms), and using a second order approximation of the messages as in \cite{donoho2010message}.

Note that, with respect to standard iterative soft thresholding,
cf. Eq.~(\ref{eq:IST_Repeat}), the algorithm (\ref{eq:MP}) has higher
complexity, since it requires to keep track of $2np$ messages, as
opposed to the $n+p$ variables in Eq.~(\ref{eq:IST_Repeat}). Also,
there is obvious interpretation to the fixed points of the iteration
(\ref{eq:MP}).

It turns out that a simpler algorithm can be defined, whose state as
dimension $n+p$ as for iterative soft thresholding, but tracks
closely the iteration (\ref{eq:MP}). 
This builds on the remark that the messages $\theta^{(t)}_{i\to a}$
issued from a node $i$ do not differ to much, since their definition
in Eq.~(\ref{eq:MP}) only differ in one out of $n$ terms.
A similar argument applies to the messages $r^{(t)}_{a\to i}$ issued
by node $a$/
We then write $\theta^{(t)}_{i\to a} = \theta_i^{(t)} +
\delta\theta_{i\to a}^{(t)}$, $r^{(t)}_{a\to i}= r_a^{(t)} +
\delta r_{a\to i}^{(t)}$ and linearize the iteration (\ref{eq:MP}) in $\{\delta\theta_{i\to a}^{(t)}\}$, $\{\delta r_{a\to i}^{(t)}\}$.
After eliminating these quantities \cite{donoho2010message}, the resulting iteration takes the
form, known as approximate message passing (AMP)
\begin{align}\tag*{(AMP)}
\begin{cases}
	&\theta^{(t + 1)} = \eta \Big(\theta^{(t)} + \frac{1}{n} \bX^\sT
        r^{(t)}; \gamma_t\Big)\, ,\\
	&r^{(t)} = y - \bX \theta^{(t)} + \ons_t r^{(t-1)},
\end{cases}\label{eq:AMP}
\end{align}
where $\ons_t \equiv \|\theta^{(t)}\|_0/n$ is a scalar.
In other words we recovered iterative soft thresholding except for
the memory term  $\ons_t r^{(t-1)}$ that is straightforward to
evaluate. In the context of statistical physics, a similar correction
is known as
the Onsager term. Remarkably, this memory term changes the statistical
behavior of the algorithm.
\index{Onsager}

It is an instructive exercise (left to the reader) to prove that fixed
points of the AMP algorithm (with $\gamma_t = \gamma_*$ fixed) 
\index{Approximate Message Passing}
are minimizers of the LASSO. In particular, for Gaussian sensing
matrices, such minimizer is unique with probability one.

We notice in passing that there is nothing special about the least
squares objective, or the $\ell_1$ regularization in our
derivation. Indeed similar ideas were developed and applied to a large
number of problems, see 
\cite{som2012compressive,som2012compressive,RanganGAMP,donoho2013accurate,donoho2013high,metzler2014denoising,tan2014compressive,barbier2014replica}
for a a very incomplete list of examples.

\subsection{Analysis of AMP and the LASSO}
\index{Approximate Message Passing}

We next carry out a heuristic analysis of AMP, referring to
\cite{bayati2011dynamics} for a  rigorous treatment that uses ideas
developed by Bolthausen in the context of mean-field spin glasses
\cite{bolthausen2014iterative}. 

We use the message passing version of the algorithm, cf. Eq.~(\ref{eq:MP})
and we will use the  assumption that the pairs $\{(r_{a\to i}^{(t)},
\bX_{ai})\}_{a\in [n]}$ are ``as if'' independent, and likewise for
$\{(\theta_{a\to i}^{(t)}, \bX_{ai})\}_{i\in [p]}$. This assumption is
only approximately correct, but leads to the right asymptotic conclusions.

Consider the first equation in (\ref{eq:MP}), and further assume 
(this assumption will be verified inductively)
\begin{align}
\E (r_{a\to i}^{(t)}) = \bX_{ai} \theta_{i}\, , \;\;\;\;
\Var(r^{(t)}_{a\to i}) =\tau_t^2\, .
\end{align}
Letting $\tr^{(t)}_{a\to i} \equiv r^{(t)}_{a\to i}-\E (r_{a\to i}^{(t)}) $,
the argument of $\eta(\,\cdot\,;\gamma_t)$ in Eq.~(\ref{eq:MP}) can be
written as 
\begin{align}
\frac{1}{n} \sum_{b \in [n]\setminus a} \bX_{bi} r_{b\to i}^{(t)} =
\frac{1}{n} \sum_{b \in [n]\setminus a} \bX_{bi}^2 \theta_{i} +
\frac{1}{n} \sum_{b \in [n]\setminus a} \bX_{bi} \tilde r_{b\to
  i}^{(t)} \approx  \theta_{i} + \frac{\tau_t}{\sqrt n}\, Z_{i\to a}^{(t)},
\end{align}
where, by central limit theorem,  $Z_{i\to a}^{(t)}$ is approximately distributed as $\normal{0}{1}$.

Rewriting the first  equation in  (\ref{eq:MP}) , we obtain
\begin{align}
\theta^{(t + 1)}_{i\to a} = \eta\Big(\theta_{i} + \frac{\tau_t}{\sqrt
  n}Z_{i\to a}^{(t)}; \gamma_t)\, .
\end{align}
In the second message equation, we substitute $y_a = w_a + \sum_{j =
  1}^p \bX_{aj} \theta_{j}$, thus obtaining
\begin{align}
r^{(t+1)}_{a\to i} = w_a + \bX_{ai} \theta_{i} - \sum_{j \in
  [p]\setminus i} \bX_{aj} (\theta^{(t+1)}_{j\to a} - \theta_{j})\, .
\end{align}
The first and the last terms have $0$ mean thus confirming the
induction hypothesis
$\E (r_{a\to i}^{(t+1)}) = \bX_{ai} \theta_{i}$.
The variance of $r^{t+1}_{a\to i}$ is given by (neglecting sublinear terms)
\begin{align}
\tau_{t+1}^2 = \sigma^2 + \sum_{j = 1}^p \left[\eta(\theta_{j} + \frac{\tau_t}{\sqrt n} Z_j ; \gamma_t) - \theta_{j}\right]^2.
\end{align}
It is more convenient to work with the rescaled quantities $\tilde
\theta_{i} = \theta_{i} \sqrt n$ and $\tilde \gamma_t = \gamma_t \sqrt
n$ (this allows us to focus on the most interesting regime, whereby
$\theta_i$ is  of the same order as
the noise level $\tau_t/\sqrt{n}$). Using the scaling property of the
thresholding function $\eta(ax, a\gamma) = a\eta(x,\gamma)$,
the last equation becomes
\begin{align}
\tau_{t+1}^2 = \sigma^2 + \frac{1}{n}\sum_{j = 1}^p \left[\eta(\ttheta_{j} + \tau_t Z_j ; \tgamma_t) - \ttheta_{j}\right]^2.
\end{align}
We now define the probability measure $p_{\Theta}$ as the asymptotic 
empirical distribution of
$\ttheta$, $p^{-1}\sum_{j =1}^p \delta_{\tilde\theta_{j}}$
(formally,  we assume that $p^{-1}\sum_{j =1}^p \delta_{\tilde\theta_{j}}$ converges weakly to  
$p_{\Theta}$, and that low order moments converge as well).
We also let  $\delta = \lim_{n\to\infty}(n/p)$ be the asymptotic
aspect ratio of $\bX$. We then obtain
\begin{align}
\tau_{t+1}^2 = \sigma^2 + \frac{1}{\delta}\E\left\{\left(\eta(\Theta + \tau_t Z; \tilde\gamma_t) - \Theta\right)^2\right\},
\end{align}
where expectation is with respect to $\Theta\sim p_{\Theta}$
independent of $Z\sim\normal{0}{1}$. 
The last equation is known as \emph{state evolution}: despite the many
unjustified assumptions in our derivation, it can be proved to
correctly describe the $n,p\to\infty$ asymptotics of the message
passing algorithm (\ref{eq:MP}) as well as of the AMP algorithm.
\index{Approximate Message Passing}

Reconsidering the above derivation, we can derive asymptotically exact
expressions for the risk at
$\theta$ for of the AMP
estimator $\theta^{(t+1)}$. Namely, we define the asymptotic risk
\begin{align}
R_{\infty}(\theta; \theta^{(t+1)}) = \lim_{n,p\to\infty}\E\big\{\|
\theta - \theta^{(t)} \|^2\big\}\, ,
\end{align}
the limit being taken along sequences of vectors $\theta$ with
converging empirical distribution. Then we claim that the limit exists
and is given by
\begin{align}\label{risk-form-1}
R_{\infty}(\theta; \theta^{(t+1)}) = \frac{1}{\delta} \ee
\left\{\left(\eta(\Theta+ \tau_t Z ; \tilde\gamma_t) -
    \Theta\right)^2\right\}\, ,
\end{align}
or, equivalently,
\begin{align}\label{risk-form-2}
R_{\infty}(\theta; \theta^{(t+1)}) = \tau_{t+1}^2 - \sigma^2.
\end{align}
 Thus, apart from an additive constant, $\tau_t^2$ coincides with
 risk and the latter can be tracked using state evolution.

In \cite{bayati2012lasso}, it is proved that the AMP iterates
$\theta^{(t)}$ converge rapidly to the LASSO estimator $\htheta$. We
are therefore led to consider the large $t$ behavior of $\tau_t$,
which yields the risk of the LASSO, or --equivalently-- the risk of
AMP after a sufficiently large (constant in $n,p$) number of iterations. 
Before addressing this question, we need to  set the values of
$\tgamma_t$. A reasonable choices to fix 
$\tgamma_t = \kappa \tau_t$, for some constant $\kappa$, since
$\tau_t$ can be thought as the 
``effective noise level'' at iteration $t$. There is a one-to-one
correspondence between
$\kappa$ and the regularization parameter $\lambda$ in the LASSO \cite{donoho2011noise}.
We thus define the function
\begin{align}
G(\tau^2;\sigma^2) \equiv \sigma^2 + \frac{1}{\delta}\E\left\{\left(\eta(\Theta
    + \tau\, Z; \kappa\tau) - \Theta\right)^2\right\}\, ,
\end{align}
which of course depends implicitly on $p_{\Theta}$, $\kappa$,
$\delta$. State evolution is then the one-dimensional
recursion $\tau_{t+1}^2 = G(\tau_t^2;\sigma^2)$. 
For the sequence $\tau_t$ to stay bounded we assume
$\lim_{\tau^2\to\infty}G(\tau^2;\sigma^2)/\tau^2<1$ which can always be ensured
by taking $\kappa$ sufficiently large. 

Let us first consider the noiseless case $\sigma=0$. Since $G(0;0) = 0$,
we know that $\tau = 0$ is always a fixed point. 
It is not hard to  shown \cite{donoho2009message} that indeed 
$\lim_{t\to\infty}\tau_t^2=0$ if and only if this is the unique
non-negative fixed point, see figure below.  If this condition is satisfied, AMP
reconstructs exactly the signal $\theta$, and due to the
correspondence with the LASSO, also basis pursuit (the LASSO with
$\lambda\to 0$) reconstructs exactly $\theta$. 
\index{Approximate Message Passing}
\vspace{0.5cm}
\begin{center}
   \def\svgwidth{0.5\textwidth}
\begingroup%
  \makeatletter%
  \providecommand\color[2][]{%
    \errmessage{(Inkscape) Color is used for the text in Inkscape, but the package 'color.sty' is not loaded}%
    \renewcommand\color[2][]{}%
  }%
  \providecommand\transparent[1]{%
    \errmessage{(Inkscape) Transparency is used (non-zero) for the text in Inkscape, but the package 'transparent.sty' is not loaded}%
    \renewcommand\transparent[1]{}%
  }%
  \providecommand\rotatebox[2]{#2}%
  \ifx\svgwidth\undefined%
    \setlength{\unitlength}{218.58122559bp}%
    \ifx\svgscale\undefined%
      \relax%
    \else%
      \setlength{\unitlength}{\unitlength * \real{\svgscale}}%
    \fi%
  \else%
    \setlength{\unitlength}{\svgwidth}%
  \fi%
  \global\let\svgwidth\undefined%
  \global\let\svgscale\undefined%
  \makeatother%
  \begin{picture}(1,0.99139347)%
    \put(0,0){\includegraphics[width=\unitlength]{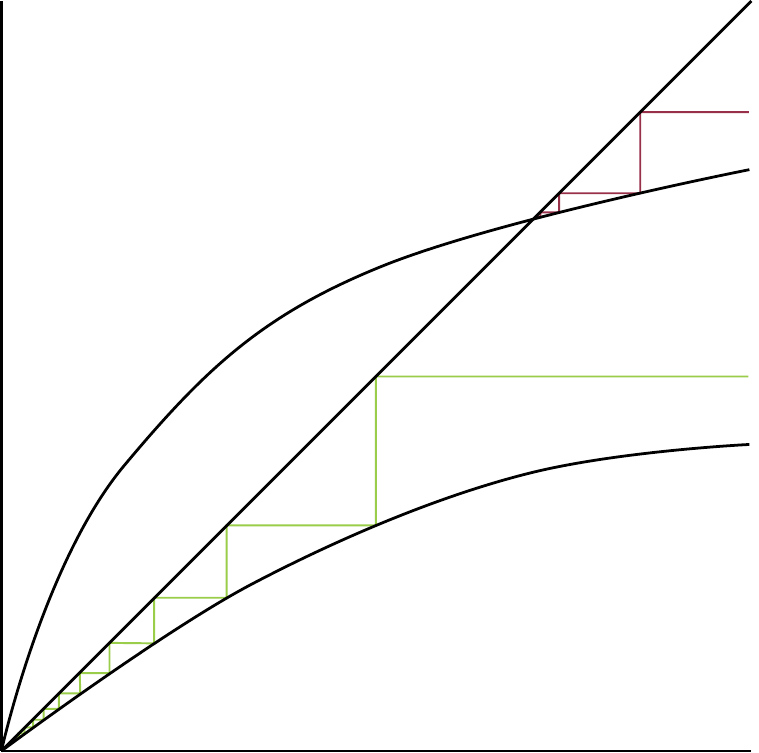}}%
    \put(0.89803919,0.02703224){\color[rgb]{0,0,0}\makebox(0,0)[lb]{\smash{$\tau^2$}}}%
    \put(0.35641031,0.72948582){\color[rgb]{0,0,0}\makebox(0,0)[lb]{\smash{bad/unstable}}}%
    \put(0.60226914,0.28028545){\color[rgb]{0,0,0}\makebox(0,0)[lb]{\smash{good/stable}}}%
    \put(0.01627489,0.9291307){\color[rgb]{0,0,0}\makebox(0,0)[lb]{\smash{$F(\tau^2)$}}}%
  \end{picture}%
\endgroup%

   \end{center}
\vspace{0.5cm}
Notice that this condition is sharp: If it is not satisfied, then AMP
and the LASSO fail to reconstruct $\theta$, despite vanishing noise. 
In order to derive the phase transition location, remember that by the
definition of minimax risk of soft thresholding, cf. Section
\ref{sec:RandomScalar}, we have, assuming $\kappa = \ell(\eps)$ to be
set in the optimal way\index{Phase Transition}
\begin{align}
G(\tau^2;0) = \frac{1}{\delta}\E\left\{\left(\eta(\Theta
    + \, Z; \kappa\tau) - \Theta\right)^2\right\}\le \frac{M(\eps)}{\delta} \tau^2\, .\label{eq:AMP_MMAX}
\end{align}
Hence $\tau_{t+1}^2\le (M(\eps)/\delta)\tau_t^2$ and, if 
\begin{align}
\delta>M(\eps)\, ,
\end{align}
  then $\tau_t^2\to 0$ and AMP (LASSO) reconstructs
$\theta$ with vanishing error. This bound is in fact tight: For
$\delta < M(\eps)$, any probability distribution $p_{\Theta}$ with
$p_{\Theta}(\{0\}) = 1-\eps$, and any threshold parameter $\kappa$,
the mean square error remains bounded away from zero.

Recalling the definition of $\delta = n/p$, the condition
$\delta>M(\eps)$ corresponds to requiring a sufficient number of
samples, as compared to the sparsity.
It is interesting to recover the very sparse regime from this point of view.
Recall from previous lectures that $M(\eps) \approx 2\eps
\log(1/\eps)$ for small $\epsilon$. The condition $\delta > M(\eps)$ then translates to
$\delta \gtrsim 2\eps \log (1/\eps)$ or, in other words, $(n/p)
\gtrsim 2  (s_0/p)\log (p/s_0)$. Thus, we obtain the condition
--already discussed before-- that the number of samples must be as
large as the number of non-zero coefficients, times a logarithmic
factor.
Reconstruction is possible if and only if
\begin{align}
n \gtrsim 2 s_0 \log \frac{p}{s_0},
\end{align}
a condition that we have seen in previous lectures.

In the noisy case, we cannot hope to achieve perfect
reconstruction. In this case,
we say that estimation is stable if  there is a constant $C$ such
that, for any $\theta\in\reals^p$, $R(\theta;\htheta ) \leq C
\sigma^2$.
This setting is sketched in the figure below. Exact reconstruction at
$\sigma = 0$ translate into a fixed point $\tau_*^2 = O(\sigma^2)$ and
hence stability. Inexact reconstruction corresponds to a fixed point
of order $1$ and hence lack of stability.
\vspace{0.5cm}
\begin{center}
   \def\svgwidth{0.5\textwidth}
\begingroup%
  \makeatletter%
  \providecommand\color[2][]{%
    \errmessage{(Inkscape) Color is used for the text in Inkscape, but the package 'color.sty' is not loaded}%
    \renewcommand\color[2][]{}%
  }%
  \providecommand\transparent[1]{%
    \errmessage{(Inkscape) Transparency is used (non-zero) for the text in Inkscape, but the package 'transparent.sty' is not loaded}%
    \renewcommand\transparent[1]{}%
  }%
  \providecommand\rotatebox[2]{#2}%
  \ifx\svgwidth\undefined%
    \setlength{\unitlength}{227.65734863bp}%
    \ifx\svgscale\undefined%
      \relax%
    \else%
      \setlength{\unitlength}{\unitlength * \real{\svgscale}}%
    \fi%
  \else%
    \setlength{\unitlength}{\svgwidth}%
  \fi%
  \global\let\svgwidth\undefined%
  \global\let\svgscale\undefined%
  \makeatother%
  \begin{picture}(1,0.95186912)%
    \put(0,0){\includegraphics[width=\unitlength]{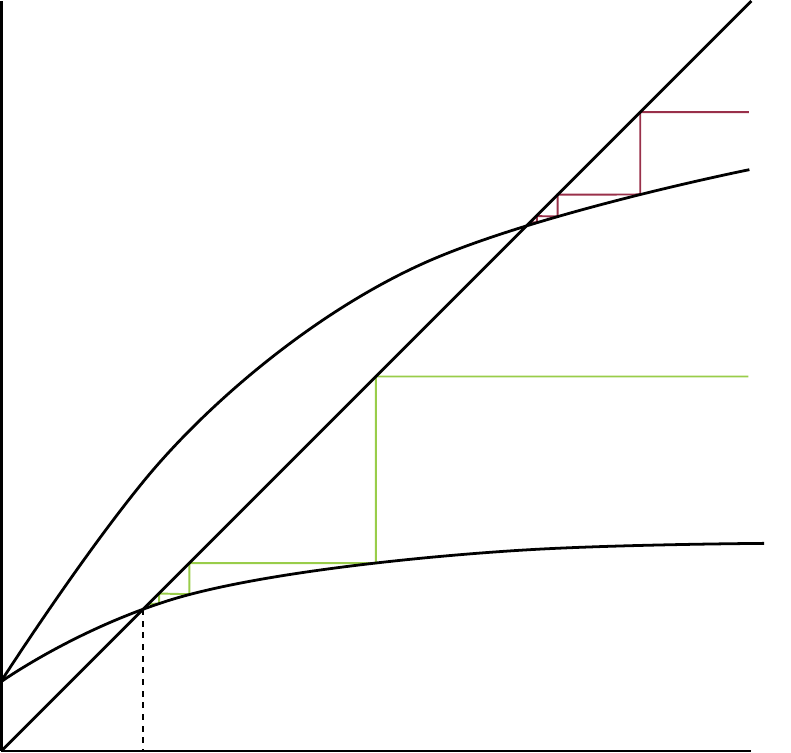}}%
    \put(0.86223664,0.02595453){\color[rgb]{0,0,0}\makebox(0,0)[lb]{\smash{$\tau^2$}}}%
    \put(0.29451042,0.71358076){\color[rgb]{0,0,0}\makebox(0,0)[lb]{\smash{bad/unstable}}}%
    \put(0.59833852,0.20636032){\color[rgb]{0,0,0}\makebox(0,0)[lb]{\smash{good/stable}}}%
    \put(0.01562605,0.89208861){\color[rgb]{0,0,0}\makebox(0,0)[lb]{\smash{$F(\tau^2)$}}}%
    \put(0.00760704,0.05282662){\color[rgb]{0,0,0}\makebox(0,0)[lb]{\smash{$\sigma^2$}}}%
    \put(0.19147691,0.13113226){\color[rgb]{0,0,0}\makebox(0,0)[lb]{\smash{$\tau^2 = O(\sigma^2)$}}}%
    \put(0.65950748,0.60116174){\color[rgb]{0,0,0}\makebox(0,0)[lb]{\smash{$\tau^2 \gg \sigma^2)$}}}%
  \end{picture}%
\endgroup%

\end{center}
\vspace{0.5cm}

Again by choosing a suitable threshold value $\kappa$, we can ensure
that the minimax bound (\ref{eq:AMP_MMAX}) is valid and hence
\begin{align}
\tau_{t+1}^2 \leq \sigma^2 + \frac{M(\eps)}{\delta} \tau_t^2\, .
\end{align}
Taking the limit $t\to\infty$, in the case that $\delta > M(\eps)$ we have that 
\begin{align}
\tau^2_* \leq \frac{\sigma^2}{1 - (M(\eps)/\delta)}.
\end{align}
This establishes that the following is an upper bound on the
asymptotic mean square error of AMP, and hence (by the equivalence
discussed above) of the LASSO,
\begin{align}
R_\infty(\theta;\htheta) = \begin{cases}
\frac{\displaystyle M(\epsilon)}{\displaystyle \delta - M(\eps)}\, \sigma^2, &\text{if } M(\eps) < \delta,\\
\infty &\text{otherwise.}
\end{cases} 
\end{align}
As proven in \cite{donoho2011noise,bayati2012lasso}, this result holds
indeed with equality (these papers have slightly different
normalizations of the noise variance $\sigma^2$).

A qualitative sketch of resulting phase diagram in $\eps$ and $\delta$ is in
the figure below.
As anticipated above, if $\delta > M(\eps)$, i.e. in the regime in
which exact reconstruction is feasible through basis pursuit in zero
noise, reconstruction is also stable with respect to noise.
\vspace{0.5cm}
\begin{center}
   \def\svgwidth{0.7\textwidth}
\begingroup%
  \makeatletter%
  \providecommand\color[2][]{%
    \errmessage{(Inkscape) Color is used for the text in Inkscape, but the package 'color.sty' is not loaded}%
    \renewcommand\color[2][]{}%
  }%
  \providecommand\transparent[1]{%
    \errmessage{(Inkscape) Transparency is used (non-zero) for the text in Inkscape, but the package 'transparent.sty' is not loaded}%
    \renewcommand\transparent[1]{}%
  }%
  \providecommand\rotatebox[2]{#2}%
  \ifx\svgwidth\undefined%
    \setlength{\unitlength}{226.34187012bp}%
    \ifx\svgscale\undefined%
      \relax%
    \else%
      \setlength{\unitlength}{\unitlength * \real{\svgscale}}%
    \fi%
  \else%
    \setlength{\unitlength}{\svgwidth}%
  \fi%
  \global\let\svgwidth\undefined%
  \global\let\svgscale\undefined%
  \makeatother%
  \begin{picture}(1,1.0100502)%
    \put(0,0){\includegraphics[width=\unitlength]{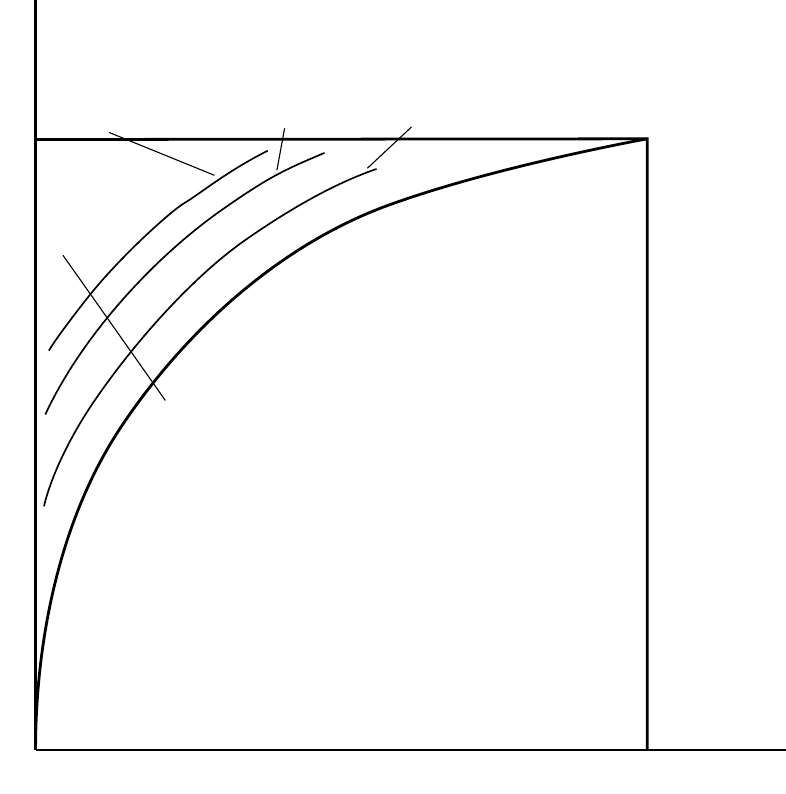}}%
    \put(0.87537519,0.09079078){\color[rgb]{0,0,0}\makebox(0,0)[lb]{\smash{$\epsilon 
= \frac{s_0}{p}$}}}%
    \put(0.06490334,0.9405366){\color[rgb]{0,0,0}\makebox(0,0)[lb]{\smash{$\delta = \frac{n}{p}$}}}%
    \put(0.43086533,0.65580027){\color[rgb]{0,0,0}\makebox(0,0)[lb]{\smash{$\delta = M(\epsilon)$}}}%
    \put(0.00554759,0.00490219){\color[rgb]{0,0,0}\makebox(0,0)[lb]{\smash{$0$}}}%
    \put(0.80887875,0.00311683){\color[rgb]{0,0,0}\makebox(0,0)[lb]{\smash{$1$}}}%
    \put(-0.00176033,0.81264676){\color[rgb]{0,0,0}\makebox(0,0)[lb]{\smash{$1$}}}%
    \put(0.29340647,0.22289421){\color[rgb]{0,0,0}\makebox(0,0)[lb]{\smash{No exact reconstruction}}}%
    \put(0.52905022,0.86466638){\color[rgb]{0,0,0}\makebox(0,0)[lb]{\smash{$\frac{R_\ast}{\sigma} = 3$}}}%
    \put(0.31750639,0.86555884){\color[rgb]{0,0,0}\makebox(0,0)[lb]{\smash{$\frac{R_\ast}{\sigma} = 2$}}}%
    \put(0.1131033,0.86645152){\color[rgb]{0,0,0}\makebox(0,0)[lb]{\smash{$\frac{R_\ast}{\sigma} = 1$}}}%
    \put(0.23717331,0.488886){\color[rgb]{0,0,0}\makebox(0,0)[lb]{\smash{Exact reconstruction
}}}%
    \put(0.23806588,0.43086743){\color[rgb]{0,0,0}\makebox(0,0)[lb]{\smash{ by $\ell_1$ minimization}}}%
  \end{picture}%
\endgroup%

\end{center}
\vspace{0.5cm}

Again, let us consider the sparse regime $\eps\to 0$. Assuming
$M(\eps)\ll \delta$, and substituting $M(\eps) \approx
2\eps\log(1/\eps)$ together with the definitions of $\eps$ and
$\delta$ we get 
\begin{align}
R_{\infty}(\theta;\htheta) =
\sigma^2\frac{M(\eps)}{\delta - M(\eps)} \approx
\frac{\sigma^2}{\delta} \, 2\eps \log \frac{1}{\eps} =
 \frac{s_0\sigma^2}{n} 2\log \frac{p}{s_0}\, .
\end{align}
We therefore rederived the same behavior already established in the
previous section under the RIP assumption. Apart from the factor
$2\log (p/s_0)$ the risk is the same `as if' we knew the support of
$\theta$. 

\label{lecture_5}

\section{The hidden clique problem}

One of the most surprising facts about sparse regression is that we
can achieve ideal estimation error, using a low complexity algorithm,
namely by solving a convex optimization problem such as the
LASSO. Indeed --at first sight-- one might have suspected it necessary
to search over possible supports of size $s_0$, a task that requires
at least $\binom{p}{s_0}$ operations, and is therefore non-polynomial.
Unfortunately, this is not always the case. There are problems in
which a huge gap exists between the statistical limits of estimation
(i.e. the minimax risk achieved by an arbitrary estimator) ant the
computational limits 
(i.e. the minimax risk achieved by any estimator computable in
polynomial time).
The hidden clique (or hidden submatrix) problem is a prototypical
example of this class of computationally hard estimation
problems. Recently, reductions to this problem were used to prove that
other estimation problems are hard as well \cite{berthet2013complexity}.

We next define the problem. Let $Q_0$ and $Q_1$ be two given
probability distributions on $\reals$. 
For a set $S \subseteq \{1, 2, \ldots, n\}$ we let $\bW \in
\reals^{n\times n}$ be  a symmetric
random matrix  with  entries
$(\bW_{ij})_{i \le j}$ independent, with distribution:
\begin{align}
\bW_{ij} &\sim Q_1, \;\;\;\mbox{if } i,j \in S\\
\bW_{ij} & \sim Q_0, \;\;\;\mbox{otherwise}\, .
\end{align}
The problem is to find the set $S$ given one realization of $\bW$.

\begin{description}
\item[Example 1] Suppose $Q_0 = \normal{0}{1}$ and $Q_1 =
  \normal{\mu}{1}$ to be two Gaussian distributions with different
  means and same known variance (which we set --without loss of
  generality-- equal to one). The model is then equivalent to the
  following
\begin{align}
\bW = \mu\, u_Su_S^{\sT} + \bZ\, ,
\end{align}
where $u_S$ is the indicator vector of the set $S$, namely $(u_S)_i =
1$ if $i\in S$, and $(u_S)_i = 0$ otherwise.
\item[Example 2] This is the original setting of the hidden clique
  problem from \cite{jerrum1992large}. Both $Q_0$ and $Q_1$ are Bernoulli
  distributions:
\begin{align}
Q_0 & = \frac{1}{2} \delta_{-1} + \frac{1}{2}\delta_{+1}\, ,\\
Q_1 & = \delta_{+1}.
\end{align}
There is a straightforward way to interpret this as a graph
problem. Let $G$ be the random graph on $n$ vertices $\{1, \ldots,
n\}$ whereby two vertices $i, j$ are joined by an edge if and only if
$\bW_{ij} = +1$. Then $G$ is an Erd\"os-Renyi random graph (with edge
density $1/2$) to which a clique has been added with support on $S$.
\end{description}
For simplicity of exposition, we will focus for the rest of this
lecture on the Bernoulli case, i.e. on the last example above.
We will use interchangeably the language of random graphs and the one
of random matrices.
All of our results can in fact be generalized to arbitrary probability
distributions $Q_0$, $Q_1$ under suitable tail conditions, as shown in \cite{deshpande2013finding}.

\begin{figure}[ht]
\centering
\subfigure[A random graph with a planted clique.]{%
	\includegraphics[width=0.5\figurewidth]{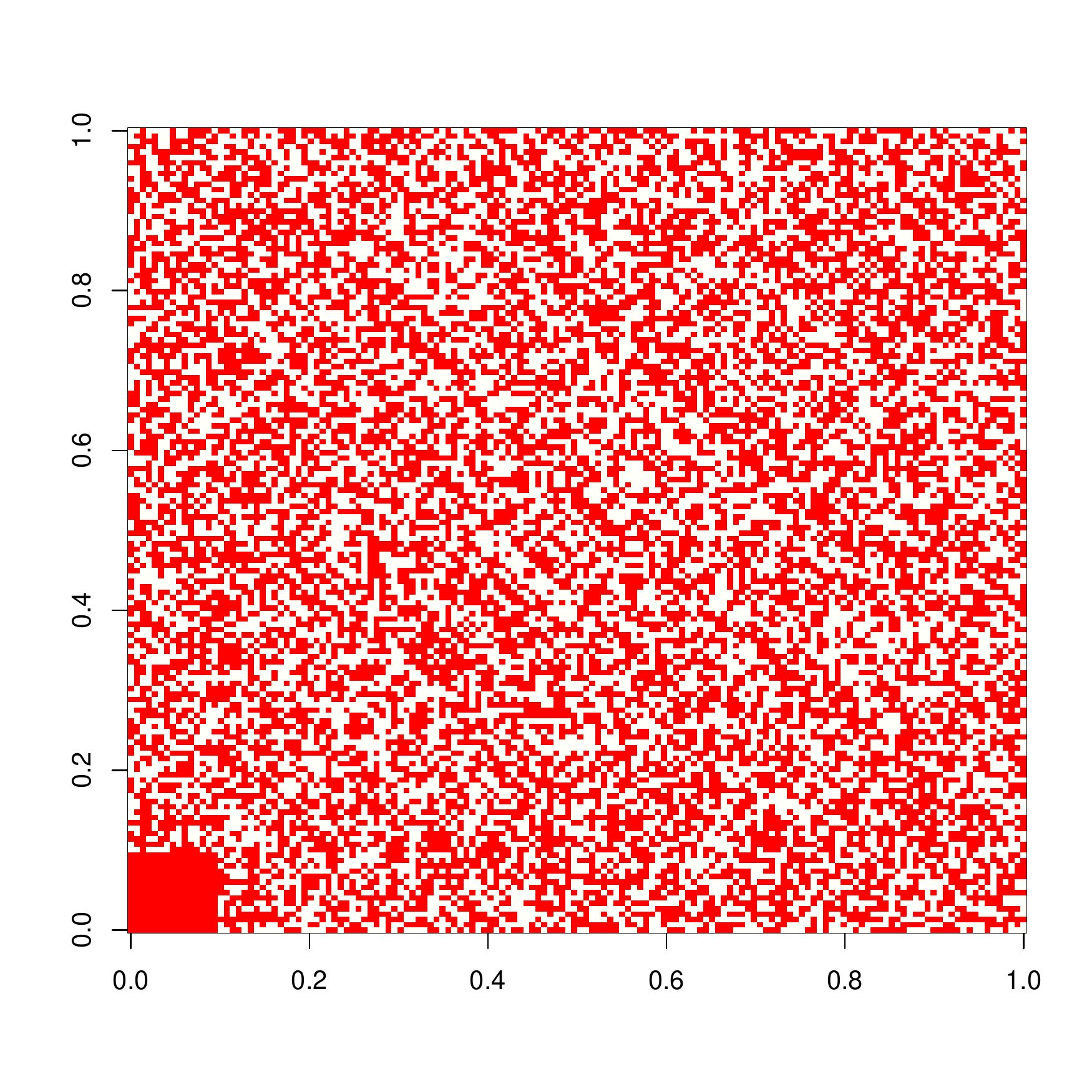}}
\quad
\subfigure[The same graph, but with the vertices shuffled.]{%
	\includegraphics[width=0.5\figurewidth]{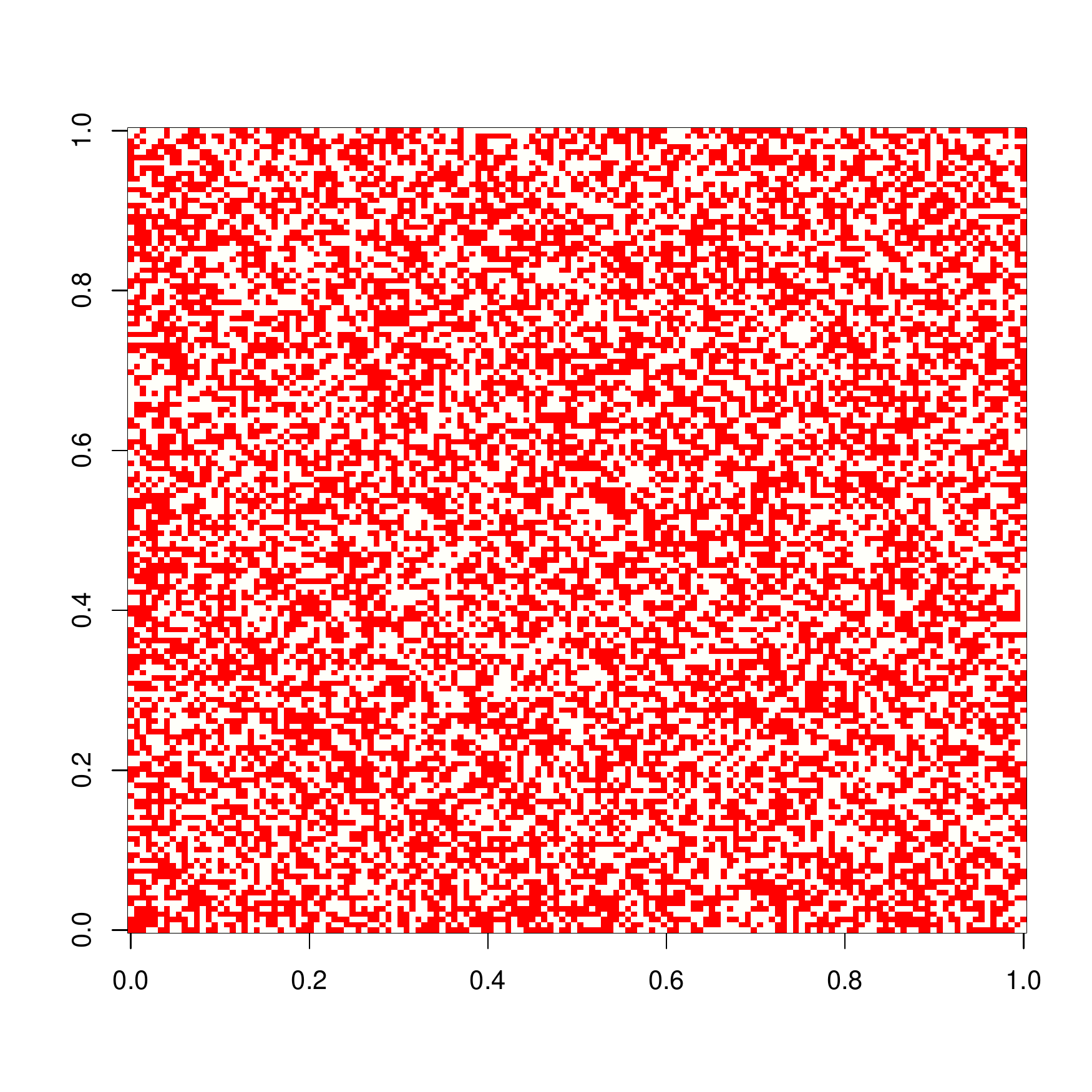}}
\quad
\subfigure[Retrieving the clique in the shuffled graph]{%
	\includegraphics[width=0.5\figurewidth]{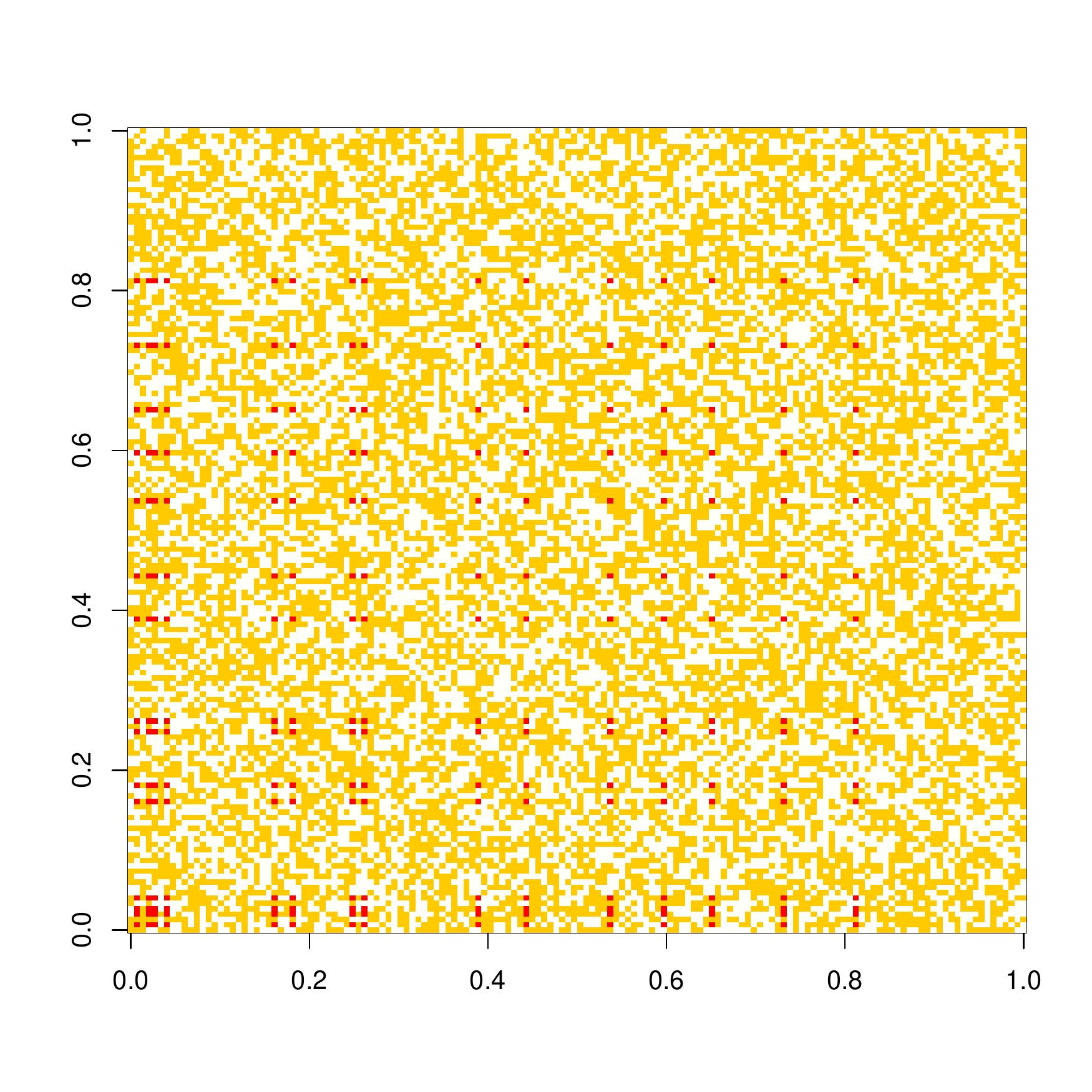}}
\end{figure}

We will denote by $k= |S|$ the size of the hidden set. It is not hard
to see that the  problem is easy for $k$ large (both from the
statistical and the computational point of view), and hard  for $k$
small (both computationally and a statistically).
Indeed, for $k$ sufficiently large, a simple degree based heuristics
is successful. This is based on the remark that vertices in the clique
have a slightly higher degree than others. Hence sorting the vertices
by degree, the first $k$ vertices should provide a good estimate of $S$.
\begin{proposition}\label{lec6-ubound}
Let $\hS$ be the set  of $k$ vertices with larges degree in $G$.
If $k \ge\sqrt{(2+\eps)n \log n}$, then with high probability $\hS = S$.
\end{proposition}
\begin{proof}
Let $D_i$ denote the degree of vertex $i$. If $i \not\in S$, then $D_i
\sim \Binom(n-1,1/2)$. In particular, standard concentration bounds on
independent random variables yield $\prob\{D_i\ge \E D_i + t\}\le
\exp(-2t^2/n)$.
By a union bound (the same already used to analyze denoising in
Section \ref{sec:Denoising}), and using $\E(D_i) = (n-1)/2$, we have,
for any $\eps'>0$, with probability converging to
one as $n\to\infty$, 
\begin{align}
\max_{i\not\in S} D_i\le \frac{n}{2}+\sqrt{(1+\eps')\frac{n\log n}{2}}\, . 
\end{align}
On the other hand, if $i \in S$, then 
$D_i\sim k-1+\Binom(n-k,1/2)$. Hence, by a similar union bound
\begin{align}
\min_{i\in S} D_i\ge \frac{n+k}{2}-\sqrt{(1+\eps')\frac{n\log k}{2}}\, .
\end{align}
The claim follows by using together the above, and selecting a
suitable value $\eps'$.
\end{proof}

For $k$ too small, the problem becomes statistically intractable
because the planted clique is not the unique clique of size $k$.
Hence no estimator can distinguish between the set $S$ and another
set $S'$ that supports a different (purely random)  clique. The next
theorem characterizes this statistical threshold.
\begin{proposition}
Let $\eps>0$ be fixed. Then, for $k <  2(1-\eps) \log_2 n$
any estimator $\hS$ is such that $\hS\neq S$ with probability
converging to one as $n\to\infty$. 

Viceversa, for $k <  2(1-\eps) \log_2 n$ there exists an estimator
$\hS$ such that $\hS = S$ with probability converging to one as $n\to\infty$.
\end{proposition}
\begin{proof}
We will not present a complete proof but only sketch the fundamental
reason for a threshold $k\approx 2\log_2 n$ and leave to the reader
the task of filling the details.

The basic observation is that the largest `purely random' clique is of
size approximately $2\log_2 n$. As a consequence, for $k$ larger than
this threshold, searching for a clique of size $k$ returns the planted clique.

More precisely, let $\cG(n,1/2)$ be an Erd\"os-Renyi random graph with
edge density $1/2$ (i.e. a random graph where each edge is present
independently with probability $1/2$). We will show that the largest clique in
$\cG(n,1/2)$ is with high probability of size between 
$2(1-\eps)\log_2n$ and $2(1+\eps)\log_2 n$.

This claim can be proved by a moment calculation. In particular, for
proving that the largest clique cannot be much larger than $2\log_2
n$, it is sufficient to compute the expected number of cliques of size $\ell$.
Letting $N(\ell;n)$ denote the number of cliques of size $\ell$ in
$\cG(n,1/2)$, we have
\begin{align}
 \E\, N (\ell;n) = \binom{n}{\ell} 2^{-\binom{\ell}{2}} \approx n^\ell 2^{-\ell^2/2} =2^{\ell\log_2 n - \ell^2/2}.
\end{align}
For $\ell>2(1+\eps)\log_2 n$ the exponent is negative and the
expectation vanishes as $n\to\infty$. In fact $\sum_{\ell\ge
  2(1+\eps)\log_2 n}\E N (\ell;n)$ vanishes as well. By Markov
inequality, it follows that --with high probability-- no clique has size larger than $2(1+\eps)\log_2 n$.
\end{proof}
The catch with the last proposition is that the estimator needs not to
be computable in polynomial time. Indeed the estimator implicitly
assumed in the proof requires searching over all subsets of $k$
vertices, which takes time at least $\binom{n}{k}\approx n^k$. For $k$
above the threshold, this is $\exp\{c(\log n)^2\}$, that is super-polynomial.

To summarize, with unlimited computational resources we can find
planted  cliques as soon as their size is larger  than $c\log_2
  n$ for any $c>2$. This is the fundamental statistical barrier
  towards estimating the set $S$/ On the other hand, the naive degree-based
  heuristic described above, correctly identifies the clique if $k \ge
  \sqrt{c\, n \log n}$. 
There is a huge gap between the fundamental statistical limit, and
what is achieved by a simple polynomial-time algorithm. This begs the
question as
to whether this gap can be filled by more advanced algorithmic ideas.

A key observation, due to Alon,  Krivelevich and Sudakov
\citeyear{alon1998finding}
is that the matrix $\bW$ --in expectation--  a rank-one matrix. Namely
\begin{align}
\E\{\bW\} = u_S\, u_S^{sT}\, ,
\end{align}
and therefore $S$ can be reconstructed from the eigenvalue
decomposition of $\E\{\bW\}$. Of course $\E\{\bW\}$ is not available, 
but one can hope  the random part of $\bW$ not to perturb too much
the leading eigenvector. In other words, one can compute the principal
eigenvector $v_1(\bW)$, i.e. the
eigenvector of $\bW$ with largest eigenvalue, and use its largest
entries to estimate the clique. For instance, one can take the $k$
vertices corresponding to the entries of $v_1(\bW)$ with largest
absolute value.

This spectral approach allows to reduce the minimum detectable clique
size by a factor $\sqrt{\log n}$, with respect  to the degree 
heuristics of \ref{lec6-ubound}.
\begin{theorem}[Alon,  Krivelevich and Sudakov, 1998]
There exists an algorithm that returns an estimate $\hS$ of the set
$S$, with the same complexity as computing the principal eigenvector
of $\bW$, and such that the following holds. 
If $k > 100\sqrt{n}$, then $\hS = S$ with probability converging to one
as $n\to\infty$.
\end{theorem}
\begin{proof}[Proof sketch]
Again, we will limit ourselves to explaining the basic argument. The
actual proof requires some additional steps.

Then the matrix $\bW$ has the form
\begin{align}
\bW = u_S u_S^{\sT} + \bZ - \bZ_{SS}, \label{eq:Wbinary}
\end{align}
\index{Wigner}
where $\bZ$ is a Wigner matrix i.e. a matrix with i.i.d. zero-mean
entries $(\bZ_{ij})_{i\le j}$, and $\bZ_{S,S}$ is the restriction of $\bZ$
to indices in $S$. In the present case, the entries
distribution is Bernoulli
\begin{align}
\bZ_{ij} = \begin{cases} +1, &\text{with probability } 1/2 \\-1,
  &\text{with probability } 1/2 \end{cases}.
\end{align}
By the celebrated F\"uredi-Komlos theorem of \cite{furedi1981eigenvalues},
the operator norm of this matrix (i.e. the maximum of the largest
eigenvalue of $\bZ$ and the largest eigenvalue of $-\bZ$)  is upper
bounded as $\|\bZ\|_2\le (2+\eps)\sqrt{n}$, with high probability. By
the same argument $\| \bZ_{SS} \|_2 \le  (2+\eps) \sqrt{k}$, which is
much smaller than $\|\bZ\|_2$.

We view $\bW$ as a perturbation of the matrix $u_Su_S^{\sT}$ (whose
principal, normalized, eigenvector is $u_S/\sqrt{k}$). Matrix
perturbation theory implies that the largest eigenvector is perturbed
by an amount proportional to the norm of the perturbation and
inversely proportional to the gap between top eigenvalue and second
eigenvalue of the perturbed matrix. More precisely, Davis-Kahan
`sin theta' theorem yields (for $v_1= v_1(\bW)$) 
\begin{align}
\sin\theta(v_1,u_S) \le
\frac{\|\bZ-\bZ_{S,S}\|_2}{\lambda_1(u_Su_S^{\sT})- \lambda_2(\bW)}\, ,
\end{align}
where $\lambda_\ell(\bA)$ denotes the $\ell$-th largest eigenvalue of
matrix $\bA$, and $\theta(a,b)$ is the angle between vectors $a$ and
$b$.  We of course have 
$\lambda_1(u_Su_S^{\sT}) = \|u_S\|_2^2 = k$, and $\lambda_2(\bW) \le
\lambda_2(u_Su_S^{\sT})+\|\bZ-\bZ_{S,S}\|_2$.
Therefore, for 
\begin{align}
\sin\theta(v_1,u_S) &\le
\frac{\|\bZ-\bZ_{S,S}\|_2}{k- \|\bZ-\bZ_{S,S}\|_2} \\
&\le \frac{2.1\sqrt{n}}{k-2.1\sqrt{n}}\le frac{1}{45}\, .
\, ,
\end{align}
where the last inequality holds with high probability by
F\"uredi-Komlos theorem.
Using standard trigonometry, this bound can be immediately converted in a bound
on the $\ell_2$ distance between $v_1(\bW)$ and the unperturbed\index{Norm $\ell_2$}
eigenvector:
\begin{align}
\Big\| v_1 - \frac{u_S}{\sqrt k} \Big\|_2 \leq \frac{1}{40}\, .
\end{align}
We can then select the set $B$ of  $k$ vertices that correspond to the
$k$ entries of $v_1$ with largest absolute value.
The last bound does not guarantee that $B$ coincide with $S$, but it
implies that $B$ must have a substantial overlap with $S$.
The estimator $\hS$ is constructed by selecting the $k$ vertices in 
$\{1,2,\dots,p\}$ that have the largest number of neighbors in $B$.
\end{proof}

It is useful to pause for a few remarks on this result.
\begin{remark}
The complexity of the above algorithm is the same as the one of
computing the principale eigenvector $v_1(\bW)$. Under the assumptions of
the theorem, this is non-degenerate and in fact, there is a large gap
between the first eigenvalue and the second one, say
$\max(\lambda_2(\bW),|\lambda_n(\bW)|)\le (1/2)\lambda_1(\bW)$. 

Hence,
$v_1(\bW)$ can be computed efficiently through power iteration, i.e. by
computing the sequence of vectors $v^{(t+1)} = \bW v^{(t)}$. 
Each operation takes at most $n^2$ operations, and due to the fast
convergence, 
$O(\log n)$ iterations are sufficient for implementing the above
algorithm.
We will revisit power iteration in the following.
\end{remark}

\begin{remark}
The eigenvalues and eigenvectors of a random matrix of the form
(\ref{eq:Wbinary}) have been studied in detail in statistics (under
the name of `spiked model') and probability theory (as `low-rank
perturbation of Wigner matrices'), see e.g.\index{Wigner}
\cite{feral2007largest,capitaine2009largest,capitaine2012central}.
These work  unveil a phase transition phenomenon that,\index{Phase Transition}
in the present application, can be stated as  follows. Assume $k,
n\to\infty$
with $k/\sqrt{n}=\kappa\in (0,\infty)$. Then
\begin{align}
\lim_{n\to\infty}|\<v_1(\bW),u_S/\sqrt{k}\>| = 
\left\{
\begin{array}{ll}
0 & \mbox{ if $\kappa\le 1$,}\\
\sqrt{1-\kappa^{-2}} & \mbox{ otherwise.}
\end{array}
\right.
\end{align}
In other words, for $k\le (1-\eps)\sqrt{n}$ the principal eigenvector of $\bW$
is essentially uncorrelated with the hidden set $S$. The barrier at
$k$ of order $\sqrt{n}$ is not a proof artifact, but instead a
fundamental limit related to this phase transition.

On the other hand, a more careful analysis of the spectral method can
possibly show that it succeeds for all $k\ge (1+\eps)\sqrt{n}$.
(Here and above $\eps>0$ is an  arbitrary constant).
\end{remark}
\begin{remark}
A clever trick by Alon and collaborators \citeyear{alon1998finding}, allow
to find cliques of size $k\ge \delta\sqrt{n}$ for any fixed constant
$\delta>0$ in polynomial time. The price to pay is that the
computational  complexity increases rapidly as $\delta$ gets smaller.
ore precisely,  we can identify sets of size $k \geq \delta \sqrt n$
for any  with time complexity of order $n^{O(\log(1/\delta))}$.

To see this, we use the spectral method as a routine that is able to
find the clique with high probability provided $k\ge c\sqrt{n}$ for
some constant $c$.
First assume  that an oracle gives us one node in the clique.
We can solve the problem with $k \gtrsim c\sqrt{n/2}$. Indeed we can
focus our attention on the set of neighbors of the node provided by
the oracle. There is about $n/2$ such neighbors, and they contain a
clique of size $k-1$, hence the spectral method will succeed under the
stated condition.

We then observe that we do not need an such: we can search for the
vertex that the oracle would tell us by 
blowing up the runtime by a factor at most  $n$ (indeed only a
$\sqrt{n}$ factor is sufficient,  since one every $\sqrt{n}$
vertices is in the clique). In this way we can trade a factor of
$\sqrt{2}$ in $k$ by an $n$-fold increase of the runtime.
This construction can be repeated $O(\log(1/\delta))$ times to achieve
the trade-off mentioned above.
\end{remark}

\subsection{An iterative thresholding approach}

Throughout this section, we shall normalize the data and work with the
matrix $\bA = \bW/\sqrt{n}$.
As we saw in the previous section, the principal eigenvector of $\bA$
carries important information about the set $S$, and in particular it
is correlated with the indicator vector $u_S$, if the hidden set $S$
is large enough.
Also an efficient way to compute the principal eigenvector is through
power iteration
\begin{align}
v^{(t+1)} = \bA\, v^{(t)}\, .
\end{align}
Note that the resulting vector $v^{(t)}$ will not --in general-- be
sparse, if not, a posteriori,  because of the correlation with $u_S$.  
It is therefore a natural idea to modify the power iteration by
introducing a non-linearity that enforces sparsity:
\begin{align}
\theta^{(t + 1)} = \bA\,  f_{t} (\theta^{(t)}),\label{eq:NonlinearPower}
\end{align}
where $\theta \in \reals^n$ and 
$f_{t}: \reals^n\to\reals^n$ is a non-linear function that enforces
sparsity.  To be definite, we will assume throughout that the
initialization is $\theta^{(0)} = (1,1,\dots ,1)$, the all-ones vector.

For ease of exposition, we shall focus on separable functions 
and denote by $f_t$ the action of this function on each component.
In other words, with a slight abuse of notation, we will write
$f_t(v) = (f_t(v_1),f_t(v_2),\dots,f_t(v_n))$ when $v =
(v_1,v_2,\dots,v_n)$. Example of such a function might be
\begin{itemize}
\item \emph{Positive soft thresholding:} $f_t(x) = (x-\lambda_t)_+$ for some
  iteration-dependent threshold $\lambda_t$. The threshold can be
  chosen so that, on average $f_t(\theta^{(t)})$ has a number of
  non-zeros of order $k$.
\item \emph{Positive hard thresholding:} $f_t(x) = x\, \ind(x\ge
  \lambda_t)$ (here $\ind$ is the indicator function: $\ind(B) =1$ if
  $B$ is true and $=0$ otherwise). Again $\lambda_t$ is a threshold.
\item \emph{Logistic nonlinearity:} 
\begin{align}
f_t(x) = \frac{1}{1+\exp(-a_t(x-\lambda_t))}\, ,
\end{align}
where $\lambda_t$ plays the role of a `soft threshold.'
\end{itemize}
Which function should we choose? Which thresholds? Will this approach
beat the simple power iteration (i.e. $f_t(x) = x$)?

In order to address these questions, we will carry out a simple
heuristic analysis of the above  non-linear power
iteration. Remarkably, we will see in the next section that this
analysis yields the correct answer for a modified version of the same
algorithm --a message passing algorithms.
Our discussion is based on \cite{deshpande2013finding}, and we refer
to that paper for al omitted details, formal statements and derivations.

 The heuristic analysis
requires to consider separately vertices in $S$ and outside $S$:
\begin{enumerate}
\item For $i\not\in S$, the non-linear power iteration
  (\ref{eq:NonlinearPower})  reads
\begin{align}
\theta_i^{(t+1)} = \sum_{j = 1}^n \bA_{ij} f_t(\theta_j^{(t)})\, . \label{eq:ByCoordinate}
\end{align}
Since in this case the variables $\{\bA_{ij}\}_{j\in [n]}$ are
i.i.d. with mean zero and variance $1/n$, it is natural to guess --by
central limit theorem-- 
$\theta_i^{(t+1)}$ to be approximately normal with mean $0$ and
variance $(1/n) \sum_{j = 1}^n f_t^2(\theta_j^{(t)})$. 
Repeating this argument inductively, we conclude that 
 that $\theta_i^t \sim \normal{0}{ \sigma_t^2}$, where --by the law of
 large numbers applied to  $(1/n) \sum_{j = 1}^n
 f_t^2(\theta_j^{(t)})$--
\begin{align}
	\sigma_{t+1}^2 = \E\big\{f_t(\sigma_t \,Z)^2\big\}\, ,
\label{lect6-eq-sigma}
\end{align}
where the expectation is taken with respect to $Z \sim \normal{0}{1}$.
The initialization $\theta^{(0)} = u$ implies $\sigma_1^2 = f_0(1)^2$
\item  For $i \in S$, we have $\bA_{ij} = \kappa$  if $j\in S$ as well,
  and $\bA_{ij} = \bZ_{ij}/\sqrt{n}$ having zero mean  and variance
  $1/n$ otherwise. Hence
\begin{align}
\theta_i^{(t+1)} = \kappa\sum_{j \in S }f_t(\theta_j^{(t)})+\frac{1}{\sqrt{n}}\sum_{j
  \in [n]\setminus S }\bZ_{ij}f_t(\theta_j^{(t)})\, .
\end{align}
By the same argument as above, the second part gives rise to a
zero-mean Gaussian contribution, with variance $\sigma_t^2$, and  the
first has non-zero mean and negligible variance. 
We conclude that $\theta_i^{(t)}$ is approximately
$\normal{\mu_t}{\sigma_t^2}$ with $\sigma_t$ given recursively by Eq.~(\ref{lect6-eq-sigma}).
Applying the law of large numbers to the non-zero mean contribution,
we get the recursion
\begin{align}\label{lect6-eq-mu}
\mu_{t+1} = \kappa \, \E\big\{f_t(\mu_t + \sigma_t \, Z)\big\}\, ,
\end{align}
where the expectation is taken with respect to $Z \sim \normal{0}{1}$,
and the initialization
$\theta^{(0)} = u$ implies $\mu_1 = \kappa\, f_0(1)$ (recall that
$\kappa$ is defined as the limit of $k/\sqrt{n}$.
\end{enumerate}
A few important remarks.

\emph{The above derivation is of course incorrect!}
The problem is that the central limit theorem cannot be applied to the
right-hand side of Eq.~(\ref{eq:ByCoordinate}) because the summands
are not independent. Indeed, each term $f_t(\theta^{(t)}_j)$ depends
on all the entries of the matrix $\bA$. 

\emph{The conclusion that we reached is incorrect}. It is not
true that, asymptotically, $\theta^{(t)}_i$ is approximately
Gaussian, with the above mean and variance. 

\emph{Surprisingly, the conclusion is correct for a slightly modified
  algorithm}, namely a message passing algorithm that will be
introduced in the next section. This is a highly non-trivial
phenomenon

\subsection{A message passing algorithm}

We modify the non-linear power iteration (\ref{eq:ByCoordinate})
by transforming it into a message passing algorithm, whose 
underlying graph is the complete graph with $n$ vertices. 
The iteration variables are `messages' $\theta^{(t)}_{i\to j}$ for 
each $i\neq j$ (with $\theta^{(t)}_{i\to j}\neq \theta^{(t)}_{j\to
  i}$. These are updated using the rule
\begin{align}
\theta_{i\to j}^{(t)} = \sum_{k \in [n]\setminus j} \bA_{ik}
f_t(\theta_{k\to i}^{(t)})\, .\label{eq:MPClique}
\end{align}
The only difference with respect to the iteration 
(\ref{eq:ByCoordinate}) is that we exclude the term 
$k=j$ from the sum. Despite this seemingly negligible change (one out of $n$ terms is dropped),
the statistical properties of this algorithm are significantly
different from the ones of the nonlinear power iteration 
(\ref{eq:ByCoordinate}), even in the limit $n\to\infty$.
In particular, the Gaussian limit derived heuristically in the
previous section, holds for the message passing algorithm. Informally,
we have, as $n\to\infty$,
\begin{align}
\theta_{i\to j}^{(t)} \sim \begin{cases}
\normal{\mu_t}{\sigma_t^2} & \mbox{ if $i\in S$},\\
\normal{0}{\sigma_t^2} & \mbox{ if $i\not\in S$},
\end{cases}
\end{align}
where $\mu_t,\sigma_t$ are determined by the state evolution equations
(\ref{lect6-eq-sigma}) and (\ref{lect6-eq-sigma}).

Let us stress that we did not yet choose the functions $f_t(\,\cdot\,)$:
we defer this choice, as well as an analysis of state evolution to the
next section.
Before this, we note that --as in the case of sparse regression-- an
approximate message passing (AMP) version of this algorithm can be
\index{Approximate Message Passing}
derived by writing $\theta^{(t)}_{i\to j}
=\theta^{(t)}_i+\delta\theta^{(t)}_{i\to j}$ and linearizing in the
latter correction.
This calculation leads to the simple AMP iteration
\begin{align}
\theta^{(t+1} = \bA f_t(\theta^{(t)}) -\ons_t f_{t-1}(\theta^{(t-1)}), \label{eq:AMPClique}
\end{align}
where  the  `Onsager term' $\ons_t$ is given in this case by
\index{Onsager}
\begin{align}
\ons_t = \frac{1}{n}\sum_{i=1}^n f'_t(\theta^{(t)}_i)\, .
\end{align}

\subsection{Analysis and optimal choice of $f_t(\,\cdot\,)$}

We now consider the implications of state evolution for the
performance of the above message
passing algorithms. For the sake of simplicity, we will refer to the
AMP form (\ref{eq:AMPClique}),
\index{Approximate Message Passing}
but analogous statements hold for the message passing version
(\ref{eq:MPClique}).
Informally, state evolution implies that
\begin{align}
\theta^{(t)} \approx \mu_t\, u_S + \sigma_t\, z\, ,\label{eq:GaussianLimit}
\end{align}
where $z\sim\normal{0}{\id_n}$, and this statement holds
asymptotically in the sense of finite-dimensional marginals. 

In other words, we can interpret $\theta^{(t)}$ as a noisy
observation of the unknown vector  $u_S$, corrupted by Gaussian
noise. 
This suggest to choose $f_t(\,\cdot\,)$ as the posterior expectation
denoiser. Namely, for $y\in \reals$ 
\begin{align}
f^{{\rm opt}}_t(y) = \E\{U |\, \mu_t\, U+\sigma_t \, Z = y\big\}\, ,
\end{align}
where $U\sim \Bernoulli(p)$ for $p =  k/n = \kappa/\sqrt{n}$, and
$Z\sim \normal{0}{1}$ independently of $U$.
A simple  calculation yields the explicit expression
\begin{align}
f^{{\rm opt}}_t(y) = \frac{\delta}{\delta + (1-\delta)\,
  \exp\Big(-\frac{\mu_t}{\sigma_t^2} y+ \frac{\mu_t^2}{\sigma_t^2}
  \Big)}\, . \label{eq:OptimalF}
\end{align}
This is indeed empirically the best choice for the non-linearity
$f^{{\rm opt}}_t(\, \cdot\, )$.
We shall next rederive it from a different point of view, which also
allow to characterize its behavior.

Reconsider again the Gaussian limit (\ref{eq:GaussianLimit}). It is
clear that  the quality of the information contained in $\theta^{(t)}$
depends on signal to noise ratio  $\mu_t/\sigma_t$. Note that $u_S$ is
very sparse, hence the vector $\theta^{(t)}$ is undistinguishable from
a zero-mean Gaussian vector unless $(\mu_t/\sigma_t)\to \infty$.
Indeed, unless this happens, the entries $\theta^{(t)}_i$, $i\in S$,
are hidden in the tail of the zero-mean entries $\theta^{(t)}_i$,
$i\in [n]\setminus S$, see figure below.
It turns out that, by optimally choosing $f_t(\,\cdot\,)$, 
this happens if and only if $\kappa > 1/\sqrt{e}$.
In other words, the message passing algorithm succeeds with high
probability as long as $k$
is larger than $(1+\eps)\sqrt{n/e}$, for any $\eps>0$.
\vspace{0.5cm}
\begin{center}
   \def\svgwidth{0.5\textwidth}
\begingroup%
  \makeatletter%
  \providecommand\color[2][]{%
    \errmessage{(Inkscape) Color is used for the text in Inkscape, but the package 'color.sty' is not loaded}%
    \renewcommand\color[2][]{}%
  }%
  \providecommand\transparent[1]{%
    \errmessage{(Inkscape) Transparency is used (non-zero) for the text in Inkscape, but the package 'transparent.sty' is not loaded}%
    \renewcommand\transparent[1]{}%
  }%
  \providecommand\rotatebox[2]{#2}%
  \ifx\svgwidth\undefined%
    \setlength{\unitlength}{282.95128bp}%
    \ifx\svgscale\undefined%
      \relax%
    \else%
      \setlength{\unitlength}{\unitlength * \real{\svgscale}}%
    \fi%
  \else%
    \setlength{\unitlength}{\svgwidth}%
  \fi%
  \global\let\svgwidth\undefined%
  \global\let\svgscale\undefined%
  \makeatother%
  \begin{picture}(1,0.90810687)%
    \put(0,0){\includegraphics[width=\unitlength]{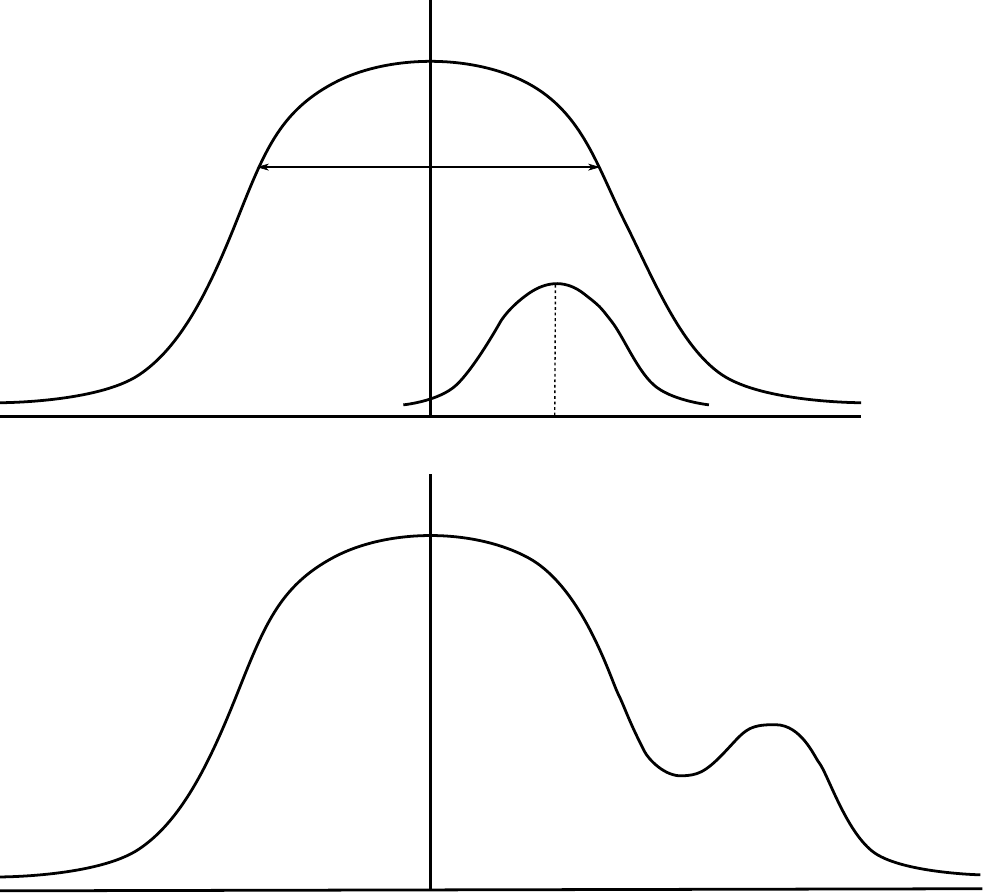}}%
    \put(0.30689231,0.70587013){\color[rgb]{0,0,0}\makebox(0,0)[lb]{\smash{$\sigma_t$}}}%
    \put(0.54108816,0.45239592){\color[rgb]{0,0,0}\makebox(0,0)[lb]{\smash{$\mu_t$}}}%
    \put(0.59392497,0.85081454){\color[rgb]{0,0,0}\makebox(0,0)[lb]{\smash{not distinguishable}}}%
    \put(0.59392497,0.36814272){\color[rgb]{0,0,0}\makebox(0,0)[lb]{\smash{distinguishable}}}%
  \end{picture}%
\endgroup%

\end{center}
\vspace{0.5cm}

In order to determine the whether $\tmu_t \equiv (\mu_t / \sigma_t)\to
\infty$, note that --without loss of generality-- we can  rescale the functions $f_t(\,\cdot\,)$ so that
$\sigma_t =1$ for all $t$ (simply replacing $f_t(z)$ by
$f_t(z)/\E\{f_t(Z)^2\}^{1/2}$ in Eq.~(\ref{eq:MPClique}), or in Eq.~(\ref{eq:AMPClique})).
After this normalization,  Eq.~(\ref{lect6-eq-mu}) yields
\begin{align}
\tmu_{t+1} = \kappa \frac{\E\left\{f_t(\tmu_t +
    Z)\right\}}{\E\{f_t(Z)^2\}^{1/2}}\ .
\end{align}
Note that
\begin{align}
\E\left\{f_t(\tmu_t +   Z)\right\} &= \int f_t(x) \,
\frac{1}{\sqrt{2\pi}}\, e^{-(x-\mu_t)^2/2} \, \de x\\
&= e^{-\tmu_t^2/2}\E\big\{f_t(Z) \, e^{\tmu_t x}\big\}\le
e^{\tmu_t^2} \E\big\{f_t(Z)^2\big\}^{1/2}\, .
\end{align}
where the last inequality follows from  Cauchy-Schwartz inequality.
The inequality is saturated by taking $f_t(x) = e^{\tmu_t x -
  (\tmu_t^2/2)}$,   that yields the state evolution recursion 
\begin{align}
\tmu_{t+1} = \kappa\,  e^{\tmu_t^2/2}\, .
\end{align}
It is immediate to study this recursion, and conclude that
$\tmu_t\to\infty$ if and only if 
$\kappa> 1/\sqrt{e}$.

The above analysis indeed yields the correct threshold for a message
passing algorithm, as proved in \cite{deshpande2013finding}.
(For proving the theorem below, a `cleaning' step is added to the
message passing algorithm.)
\begin{theorem}[Deshpande, Montanari, 2014]
There exists an algorithm with time complexity $O(n^2 \log n)$, that
outputs an estimate $\hS$ such that --if $k > (1 + \epsilon)
\sqrt{\frac{n}{e}}$--
then $\hS = S$ with probability  converging to one as $n\to\infty$.
\end{theorem}
In other words, the message passing algorithm is able to find 
cliques smaller by a factor $1/\sqrt{e}$ with respect to spectral
methods, with no increase in complexity.
A natural research question is the following:
\begin{quote}
\emph{Is it possible to planted find cliques of size $(1-\eps)\sqrt{n/e}$ in
time $O(n^2\log n)$?}
\end{quote}
The paper \cite{deshpande2013finding} provides a partially positive
answer to this question, by 
showing that no `local algorithm' (a special class of linear-time
algorithm) can beat message passing algorithms for a sparse-graph
version of the planted clique problem.

Let us conclude by showing how the last derivation 
agrees in fact with the guess (\ref{eq:OptimalF}) for the optimal
non-linearity.
Note that $\delta =\kappa/\sqrt{n}\to 0$ as $n\to\infty$. In this
limit
\begin{align}
f^{{\rm opt}}_t(y) \approx C_t \, \exp\Big\{\frac{\mu_t}{\sigma_t^2}\, y\Big\}\,.
\end{align}
This coincides with the choice optimizing the state evolution
threshold, once we set $\sigma_t=1$ (that entails no loss of generality).

\label{lecture_6}

\bibliographystyle{IEEEtran}
\bibliography{ChapMontanari/references}

\end{document}